\documentclass[10pt,a4paper]{article}
\usepackage{tikz}
\usepackage{amssymb}
\usepackage{mathtools}
\usepackage{amsmath}
\usepackage{amsthm}
\usepackage{latexsym}
\usepackage[dvips]{epsfig}
\usepackage{mathrsfs}
\usepackage{eufrak}
\usepackage{graphicx}
\usepackage{bm}
\usepackage{enumerate}
\usepackage{authblk}

\usepackage{framed}

\theoremstyle{plain}
\newtheorem{proposition}{Proposition}
\newtheorem{lemma}{Lemma}
\newtheorem{theorem}{Theorem}

\newtheorem{corollary}{Corollary}

\newtheorem{remark}{Remark}

\def\bmg{{\bm g}}
\def\Hone{{\mathcal{H}_{1}}}
\def\Htwo{{\mathcal{H}_{2}}}
\def\cZ{{\mathcal{Z}}}

\usepackage[normalem]{ulem}
\usepackage{color}

\definecolor{DGREEN}{rgb}{0,0.7,0.3}
\definecolor{DRED}{rgb}{.7,0.1,0.3}

\newcounter{mnotecount}

\newcommand{\mnotex}[1]
{\protect{\stepcounter{mnotecount}}$^{\mbox{\footnotesize $\bullet$\themnotecount}}$ 
\marginpar{
\raggedright\tiny\em
$\!\!\!\!\!\!\,\bullet$\themnotecount: #1} }

\begin{document}

\title{\textbf{Killing spinor data on distorted black hole horizons and the uniqueness of stationary vacuum black holes}}


\author[,1]{M.J. Cole \footnote{E-mail
    address:{\tt m.j.cole@qmul.ac.uk}}}
\author[,2,3]{I. R\'acz \footnote{E-mail address:{\tt racz.istvan@wigner.mta.hu}}}
\author[,1]{J.A. Valiente Kroon \footnote{E-mail address:{\tt j.a.valiente-kroon@qmul.ac.uk}}}

\affil[1]{School of Mathematical Sciences, Queen Mary, University of London,
Mile End Road, London E1 4NS, United Kingdom.}

\affil[2]{Faculty of Physics, University of Warsaw, Ludwika Pasteura 5, 02-093 Warsaw, Poland}

\affil[3]{Wigner RCP, H-1121 Budapest, Konkoly Thege Mikl\'{o}s \'{u}t
29-33, Hungary}

\maketitle

\begin{abstract}
We make use of the black hole holograph construction of
 [I. R\'acz, Stationary black holes as holographs, Class. Quantum
  Grav. 31, 035006 (2014)]
to analyse the existence of Killing spinors in the
domain of dependence of the horizons of distorted black holes. In
particular, we provide conditions on the bifurcation sphere ensuring
the existence of a Killing spinor. These conditions can be understood
as restrictions on the curvature of the bifurcation sphere and ensure
the existence of an axial Killing vector on the 2-surface. We obtain
the most general 2-dimensional metric on the bifurcation sphere for
which these curvature conditions are satisfied. Remarkably,
these conditions are found to be so restrictive that, in the
considered particular case, the free data on the bifurcation surface
(determining a distorted black hole spacetime) is completely
determined by them.  In addition, we formulate further conditions on
the bifurcation sphere ensuring that the Killing vector associated to
the Killing spinor is Hermitian. Once the existence of a Hermitian
Killing vector is guaranteed, one can use a characterisation of the
Kerr spacetime due to Mars to identify the particular subfamily of
2-metrics giving rise to a member of the Kerr family in the
black hole holograph construction. Our analysis sheds light
on the role of asymptotic flatness and curvature conditions on the
bifurcation sphere in the context of the problem of uniqueness of
stationary black holes. The Petrov type of the considered
distorted black hole spacetimes is also determined.
\end{abstract}


\tableofcontents 


\section{Introduction}

In \cite{Rac14} it is shown that the 4-dimensional geometry of a
spacetime admitting a pair of expansion- and shear-free null
hypersurfaces $\Hone$ and $\Htwo$ intersecting on a two-surface
$\cZ\equiv\Hone\cap\Htwo$ can uniquely be determined in the domain of
dependence of $\Hone\cup\Htwo$, once suitable data ---consisting of
three complex functions--- has been prescribed on
$\cZ=\Hone\cap\Htwo$. This set-up provides a basis for the use of the
\emph{characteristic initial value problem} in the investigation of a
variety of black hole configurations by inspecting the freedom in
the specification of the data on the bifurcation surface $\cZ$ of the
horizons only. In the following, we will often refer to this set-up
as \emph{R\'{a}cz's black hole holograph construction}. In fact, the
set-up introduced in \cite{Rac07, Rac14} is suitable to
host all of the \emph{stationary distorted electrovacuum black hole
spacetimes} ---within the class of solutions to the Einstein-Maxwell
equations with non-zero cosmological constant.

\medskip
As it was proposed already in \cite{Rac07, Rac14} the black hole
holograph construction should  open a new
avenue in the black hole uniqueness problem. To this end note that the
Kerr-Newman family of solutions (describing a charged, rotating black
hole) is an example of a family of exact solutions to the
Einstein-Maxwell equations satisfying these conditions, and so it
belongs this class of distorted black hole solutions. Thus, one can
naturally ask what further conditions are necessary to impose on the
horizons in order to single out the Kerr-Newman family from the more
general class, and how restrictive these conditions are.

\medskip
In the present article, we make use of a characterisation of the Kerr solution by
Killing spinors to to identify the appropriate set of
conditions on the data at the bifurcation sphere $\cZ$. Killing spinors are known to
represent \emph{hidden symmetries} of a spacetime, and the existence
of such a field on the Kerr spacetime is directly related to the
existence of the Carter constant, which allows the geodesic equations
to be completely integrated \cite{Car68b} ---see also
\cite{Car73a,WalPen70}. In recent work, it has been shown that the
existence of a Killing spinor on a spacetime, along with the
assumption of asymptotic flatness, can be used to identify the
spacetime as a member of the Kerr or Kerr-Newman families
\cite{ColVal16a}. These ideas have been used in previous work to
determine whether initial data corresponds to exact Kerr data. The
assumption on the existence of a Killing spinor can be recast as an
initial value problem, producing a set of \emph{Killing spinor initial
data equations} that must be satisfied on a spacelike initial
hypersurface. These constraint equations can be used, for example, to
determine whether the initial data set on the hypersurface corresponds
to initial data for the exact Kerr spacetime. In a similar way, in the
present article, we show it is possible to guarantee the existence of
a Killing spinor on the domain of dependence
$D(\Hone\cup\Htwo)$ of the intersecting expansion
and shear-free horizons $\Hone\cup\Htwo$ by prescribing data for the
Killing spinor and, in accordance with the black hole
holograph construction, this data need only be given on the
intersection surface $\cZ$. The only restriction on the background
spacetime is the prescription of the only gauge invariant
Weyl spinor component $\Psi_{2}$ in terms of this initial data.

\medskip
In this article, we consider the vacuum case and set the goal of
identifying the Kerr family of solutions to the Einstein equations
from the general class of \emph{stationary} distorted vacuum black
hole spacetimes. We give a set of conditions which must be satisfied
on the bifurcation sphere $\cZ$ to ensure the existence of a Killing
spinor on the domain of dependence $D(\Hone\cup\Htwo)$ of
$\Hone\cup\Htwo$ ---which is nothing but the interior of the black hole
region in the smooth setting, whereas it also contains the domain of
outer communication if analyticity is allowed--- and describe further
conditions which must be given there to single out the Kerr
solution. Our main results can be described as follows:

\begin{enumerate}[(i)]
\item \emph{We identify the conditions that need to be imposed on the
initial data surface ---comprised by two expansion- and shear-free}
horizons intersecting on a two-surface $\cZ$--- to ensure the
existence of a Killing spinor in the domain of dependence of the
horizons, $\mathcal{H}_1\cup \mathcal{H}_2$. These conditions are
stated in Lemmas \ref{KillingSpinorH1}, \ref{KillingSpinorH2},
\ref{KillingVectorZ}, \ref{KillingVectorH1} and
\ref{KillingVectorH2}. These conditions set restrictions on both the
free specifiable data for the Killing spinor on
$\mathcal{H}_1\cup\mathcal{H}_2$ and on the components of the Weyl
spinor and some of the spin connection coefficients.

\item We show that the conditions obtained in (i) can be imposed by
  satisfying a set of \emph{Killing spinor constraints} at the
  bifurcation sphere $\mathcal{Z}$. In particular, it turns out that the
  whole Killing spinor data can be propagated along
  $\mathcal{H}_1\cup\mathcal{H}_2$ from some \emph{basic Killing spinor data}
  on $\mathcal{Z}$. The Killing spinor constraints on $\cZ$ are given
  in  Proposition \ref{HSZeroReduceToZ}. 

\item Using R\'{a}cz's black hole holograph construction, it follows that
\emph{if the Killing spinor constraints are satisfied, then one can
ensure the existence of a Killing spinor everywhere in the domain of
dependence of the horizons of distorted black hole.} This result is
stated more precisely in Theorem
\ref{Theorem:ExistenceKillingSpinors}.  

\item All of the above results are local ---i.e.~independent of the
topology of $\cZ$. Note, however, if one restrict considerations to
black holes, in virtue of Hawking's black hole topology theorem
\cite{Hawk72,HE73} (see also Corollary 4.2 of \cite{Rac07} relevant
for generic distorted black holes) $\cZ$ has to have the topology of a
2-sphere. \emph{Using this assumption, it is shown that the Killing
spinor constraints imply, in particular, that the Killing
vector field ---that always comes together with the existence of a
generic Killing spinor--- gives rise to be an axial Killing vector
field on the bifurcation surface}. This result is given in
Proposition \ref{AxialSymmetry:BifurcationSphere}.

\item \emph{We show how to encode, in terms of further constraints on $\cZ$,
  that the Killing vector field determined by the Killing spinor is
  Hermitian (i.e.~real)} in the domain of dependence of
  $\mathcal{H}_1\cup \mathcal{H}_2$. The relevant conditions are spelled out in full details
  in Lemma \ref{Lemma:HermiticityRefined}.

\item \emph{We determine the most general (regular) two-metric and
associated curvature scalar on the bifurcation sphere $\cZ$ ensuring
that the Killing spinor constraints are satisfied.} This result is
presented in Proposition \ref{Proposition:MetricsZ}.

\item \emph{ Notably, the aforementioned Killing spinor constraints,
can be seen to be geometrically equivalent to the freedom one has in
choosing initial data in R\'{a}cz's black hole holograph
construction. } The corresponding results are presented in Section
\ref{DeterninancyOnGeom}.

\item \emph{It is also shown that the existence of the Killing spinor,
in the generic case, implies that the domain of dependence of
$\mathcal{H}_1\cup \mathcal{H}_2$ must be of Petrov type D.}  This
result is presented in Subsection~\ref{PetrovType} (see also
Corollary~\ref{cor: PetrovType}).

\item \emph{Finally, we also give a clear identification of that
subclass of basic initial data on $\cZ$ which gives rise to a member
of the Kerr family of spacetimes in the domain of dependence of
$\mathcal{H}_1\cup \mathcal{H}_2$.}  The basic idea behind this
calculation is to make use of Mars's invariant characterisation of the
Kerr spacetime that is summarised in Theorem
\ref{Theorem:SpinorialCharacterisationKerr}. The conditions on the
freely specifiable data leading to a development isomorphic to Kerr
are spelled out in Proposition \ref{Proposition:Kerr}.

\end{enumerate}

\subsection*{Overview} 
This paper is structured as follows: in Section
\ref{Section:KillingSpinors}, we recall the results of
\cite{ColVal16a}, illustrating how the existence of a Killing spinor
can be used to characterise the Kerr spacetime. This is done in the form of a
local result requiring the evaluation of two constants. 

\smallskip
In Section
\ref{Section:CharacteristicInitialValueProblem}, we summarise the
construction of the characteristic problem in \cite{Rac14}, used to
define the class of distorted black holes to be considered in this
article. 

\smallskip
In Section \ref{Section:KillingSpinorData}, we decompose the
wave equation for the Killing spinor into equations intrinsic to the
horizons, providing a system of transport equations for the components
of the Killing spinor. Furthermore, by finding a system of homogeneous
wave equations for a collection of zero-quantity fields and imposing
appropriate initial data for the system, we find further conditions
(differential and algebraic constraints) for the components of the
Killing spinor and their first derivatives on the bifurcate horizon
$\Hone\cup\Htwo$ ---the \emph{Killing spinor data conditions} on $\mathcal{Z}$.

\smallskip
In Section \ref{Section:ConstraintsZ}, we
investigate these constraints. We show that the conditions
intrinsic to the bifurcation sphere $\cZ$ imply a specific form for
the components of the Killing spinor. We also show that the constraints
intrinsic to $\Hone$ or $\Htwo$ satisfy ordinary differential
equations along the generators of the relevant horizons, and so can be
replaced with conditions on the bifurcation sphere, or become
redundant. In this way, conditions on the extended horizon
construction are reduced to conditions only on the bifurcation surface
$\cZ$. 

\smallskip
In Section \ref{AxialSymmetry:BifurcationSphere} it is shown
that the Killing spinor data conditions on $\cZ$ imply that the bifurcation sphere is an
axially symmetri 2-surface. 

\smallskip
In Section \ref{Section:HermiticityKillingVector} further conditions
on $\cZ$ are obtained which ensure that the Killing vector associated
to the Killing spinor is a Hermitian (i.e. real) vector. 

\smallskip
Section \ref{Section:DeterminationKappa1} discusses the more general
solution to the constraints on $\cZ$. This solution fixes the Gaussian
curvature of the bifurcation sphere $\mathcal{Z}$ and, in turn, also
the functional form of the metric of the 2-surface and the spin
coefficient $\tau$. Using this metric one can use R\'{a}cz's black hole
holograph construction to (locally in a neighbourhood of $\cZ$) obtain
the most general family of vacuum (with vanishing cosmological
constant) distorted black holes with Killing spinors.

\smallskip
Section \ref{Section:Kerr} is devoted to the task of identifying the Kerr out of
the class of spacetimes constructed in the previous section. The main
tools for this is a characterisation of Kerr spacetime in terms of
Killing spinors based on a more general result by Mars \cite{Mar00} ---see Theorem
\ref{Theorem:SpinorialCharacterisationKerr} in Section \ref{Section:KillingSpinors}. 

\smallskip
We provide some concluding remarks in Section \ref{Section:Conclusions}.

\subsection*{Notation and conventions}
In what follows $(\mathcal{M},\bmg)$ will denote a vacuum
spacetime. The metric $\bmg$ is assumed to have signature
$(+,-,-,-)$. The Latin letters $a,\,b,\ldots$ are used as abstract
tensorial spacetime indices. The script letters $\mathcal{A},\,
\mathcal{B},\ldots$ are used to denote \emph{angular coordinates}. The
Latin capital letters $A,\,B,\ldots$ are used as abstract spinorial
indices.

\medskip
We make systematic use of the standard Newman-Penrose (NP) formalism as discussed
in, say, \cite{PenRin84,Ste91}. Standard NP notation and conventions will be used ---see
e.g. \cite{Ste91}. In particular, if $\eta$ is a smooth scalar on a 2-surface $\cZ$ with
spin-weight $s$, the action of the $\eth$ and  $\overline{\eth}$ operators on
is defined by
\begin{equation}
\eth \eta = \delta \eta + s \,(\overline{\alpha}-\beta)\,\eta\,, \qquad \overline{\eth}
\eta = \overline{\delta}\,\eta - s\,(\alpha-\overline{\beta})\,\eta\,.
\label{eths}
\end{equation}
One also has that
\begin{equation}
( \overline{\eth} \eth - \eth \overline{\eth})\, \eta = s\,K_{\mathcal{G}} \,\eta\,,
\label{laplace}
\end{equation}
where $K_{\mathcal{G}}$ denotes the Gaussian curvature of $\cZ$.

\medskip
We shall also make use of the representation of $\eth$ and
$\overline{\eth}$ operators following the construction in Section 4.14
of \cite{PenRin84}. In particular, by choosing an arbitrary
holomorphic function $z$ the 2-metric $\bm \sigma$ on $\cZ$ can be
given as
\begin{equation}\label{metric_sigma}
{\bm \sigma} = -\frac{1}{P\overline{P}}\,\big( \mathbf{d}z\otimes
\mathbf{d}\overline{z} + \mathbf{d}\overline{z}\otimes \mathbf{d}z  \big),
\end{equation}
where $P$ is a complex function on $\mathcal{Z}$. If $\mathcal{Z}$ was
the unit sphere $\mathbb{S}^2$, then the coefficient $P$ would have
the form $P=\frac{1}{2}(1+z\overline{z})$.

\medskip
The operators $\eth$ and $\overline\eth$---acting on a scalar $\eta$ of spin-weight $s$---are defined as (see (4.14.3)-(4.14.4) in \cite{PenRin84})
\begin{equation}\label{eth_P}
\eth \eta \equiv P\overline{P}^{-s} \frac{\partial}{\partial {z}}\big(\overline{P}^s \eta 
\big)\,, \qquad \overline\eth \eta \equiv \overline{P}P^{s} \frac{\partial}{\partial
	\overline{z}}\big(  P^{-s} \eta\big)\,. 
\end{equation}
As the complex coordinates $z$ and $\overline{z}$ have no spin-weight direct 
calculations readily verify that 
\[
\eth {z} =P, \qquad \eth \overline{z} =0\,, 
\]
and that 
\[
\eth \overline{P} =0\,, \qquad \overline{\eth} P=0\,.
\]

\medskip
Note, finally, that in the generic setup for the characteristic
initial value problem the initial data is given on a pair of
intersecting null hypersurfaces $\Hone$ and $\Htwo$. The solution to
the Einstein's equations is known to exist in certain domains. We
shall denote the domain of dependence of $\mathcal{H}_1\cup
\mathcal{H}_2$ by $D(\mathcal{H}_{1}\cup\mathcal{H}_{2})$. The extent
of $D(\mathcal{H}_{1}\cup\mathcal{H}_{2})$ is known to depend on the
techniques used in verifying the existence of solutions. According to
the claims in \cite{Ren90} it is covering only a neighbourhood
$\mathcal{O}$ of the spacelike 2-surface $\cZ$ (indicated by the blue
coloured area on Fig.\ref{fig}). Nevertheless, when techniques of
energy estimates are used, as e.g.~in \cite{Luk12}, the domain of
dependence can be seen to be larger covering (at least certain)
neighbourhood of the two intersecting null hypersurfaces $\Hone$ and
$\Htwo$ (indicated by the green coloured area---including the blue
one---on Fig.\ref{fig}). Hereafter, we shall refer to the domain of
dependence without explicit mentioning of its extent. This is done to
simplify the presentation. The size of this domain does not play a
significant role in our discussions.

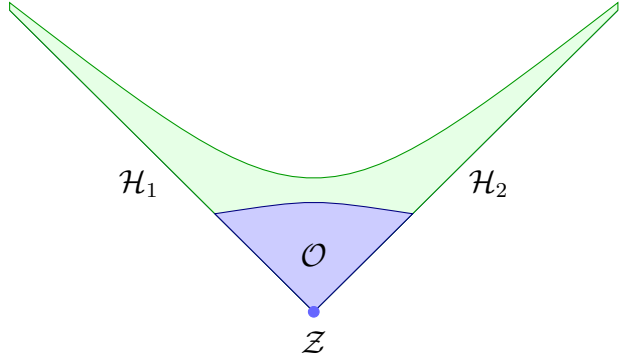
\begin{figure}\label{fig}
	\centering
	\begin{tikzpicture}
	\draw [blue!50!black] (0,0) --(4,4);
	\draw [blue!50!black] (0,0) --(-4,4);
	\filldraw[fill=green!10!white, draw=green!60!black]
	(0,0) -- (-4,4) -- (-4,4.1) .. controls (0,1) and (0,1) .. (4,4.1) -- (4,4) -- cycle;
	\filldraw[fill=blue!20!white, draw=blue!50!black]
	(0,0) -- (-1.3,1.3) .. controls (0,1.5) and (0,1.5) .. (1.3,1.3) -- cycle;
	\draw (2.3,1.7) node {\large $\Htwo$};
	\draw (-2.3,1.7) node {\large $\Hone$};
	\draw (0,0.75) node {\large $\mathcal{O}$};
	\draw (0,-0.4) node {\large $\cZ$};
	\filldraw [blue!60!white] (0,0) circle (2pt);
	\end{tikzpicture}
	\caption{
		{
	The possible extents of the domain of dependence of the initial data surface, comprised by a pair of null hypersurfaces $\Hone$ and $\Htwo$  intersecting on a two-surface $\cZ=\mathcal{H}_1\cap\mathcal{H}_2$ in the characteristic initial value problem, is indicated. As described in the main text these depend on the techniques used to control the existence of solutions to the specific initial value problem.}
	}
\end{figure}

\section{An invariant characterisation of the Kerr spacetime}
\label{Section:KillingSpinors}
In this section we provide a brief overview of a characterisation of
the Kerr spacetime by means of Killing spinors.

\subsection{Killing spinors}

A \emph{Killing spinor} is a symmetric rank 2 spinor
$\kappa_{AB}=\kappa_{(AB)}$ satisfying the \emph{Killing spinor
  equation}
\begin{equation}
\nabla_{A'(A} \kappa_{BC)}=0.
\label{KillingSpinorEquation}
\end{equation}
Given a spinor $\kappa_{AB}$, the spinor 
\begin{equation}
\xi_{AA'} \equiv
\nabla^P{}_{A'} \kappa_{AP}
\label{KillingSpinorToKillingVector}
\end{equation}
 is the spinorial counterpart of a
(possibly complex) Killing vector. Thus, it satisfies the equation
\[
\nabla_{AA'} \xi_{BB'} + \nabla_{BB'} \xi_{AA'}=0.
\]

\medskip
Conditions on a spacelike hypersurface $\mathcal{S}$ ensuring the
existence of a Killing spinor on the future domain of dependence of
$\mathcal{S}$, $D^+(\mathcal{S})$, have
been analysed in \cite{BaeVal10b,GarVal08b}. In view of the subsequent
discussion it will be  convenient to define the following
\emph{zero quantities:}
\begin{eqnarray*}
&& H_{A'ABC} \equiv 3 \nabla_{A'(A} \kappa_{BC)}\,, \label{eq: defH}\\
&& S_{AA'BB'}
\equiv \nabla_{AA'} \xi_{BB'} + \nabla_{BB'}\xi_{AA'}\,.\label{eq: defS}
\end{eqnarray*}
A straightforward consequence of the Killing spinor equation is the
wave equation
\begin{equation}
\square \kappa_{AB} +\Psi_{ABCD} \kappa^{CD} =0,
\label{KillingSpinorWaveEquation}
\end{equation}
where $\Psi_{ABCD}$ denotes the Weyl spinor.

\medskip
A calculation then yields the following:

\begin{proposition}
\label{Proposition:KillingSpinorPropagationSystem}
Let $\kappa_{AB}$ be a solution to equation
\eqref{KillingSpinorWaveEquation}. Then the spinor fields $H_{A'ABC}$
and $S_{AA'BB'}$ satisfy the system of wave equations
\begin{subequations}
\begin{eqnarray}
&& \square H_{AA'BC} = 4 \big( \Psi_{(AB}{}^{PQ} H_{C)PQA'} +
\nabla_{(A}{}^{Q'} S_{BC)Q'A'} \big),  \label{WaveEquationH}\\
&& \square S_{AA'BB'} = -\nabla_{AA'} ( \Psi_B{}^{PQR} H_{B'PQR}) -
\nabla_{BB'} (\Psi_A{}^{PQR} H_{A'PQR})  \nonumber\\
&& \hspace{3cm}+ 2 \Psi_{AB}{}^{PQ}
S_{PA'QB'} + 2 \overline{\Psi}_{A'B'}{}^{P'Q'} S_{AP'BQ'}. \label{WaveEquationS}
\end{eqnarray}
\end{subequations}
\end{proposition}

\begin{remark}
{\em As the above equations constitute a system
  of  homogeneous linear wave equations for the fields $H_{A'ABC}$ and $S_{AA'BB'}$, it
follows that they readily imply conditions for the
existence of a Killing spinor in the development of a given initial
value problem for the Einstein field equations.} 
\end{remark}

\subsection{The Killing form and the Ernst potential}
In this section let $\xi_{AA'}$ denote the spinorial counterpart of a
real Killing vector $\xi^a$. Accordingly, $\xi_{AA'}$ is assumed to be
Hermitian. The spinorial counterpart of the Killing form of $\xi^a$,
namely, 
\[
H_{ab} \equiv \nabla_{[a} \xi_{b]} = \nabla_a \xi_b
\]
is given by
\[
H_{AA'BB'} \equiv \nabla_{AA'} \xi_{BB'}.
\]
As a consequence of the antisymmetry in the pairs ${}_{AA'}$ and
${}_{BB'}$, the latter can be decomposed into irreducible parts as 
\[
H_{AA'BB'}\equiv \eta_{AB}\epsilon_{A'B'} + \overline{\eta}_{A'B'} \epsilon_{AB},
\]
where $\eta_{AB}$ is a symmetric spinor. In the following we will make
use of the self-dual part of $H_{AA'BB'}$, denoted by
$\mathcal{H}_{AA'BB'}$, and defined by
\[
\mathcal{H}_{AA'BB'} \equiv H_{AA'BB'} + \mbox{i} H^*_{AA'BB'}.
\]
It follows readily that 
\[
\mathcal{H}_{AA'BB'} = 2 \eta_{AB} \epsilon_{A'B'}.
\]
The spinor $\eta_{AB}$ can be expressed in terms of $\xi_{AA'}$ as
\[
\eta_{AB} = \frac{1}{2}\nabla_{AA'} \xi_B{}^{A'}.
\]
If, moreover, $\xi_{AA'}$ is obtained from a Killing spinor
$\kappa_{AB}$ using formula \eqref{KillingSpinorToKillingVector}, then
one has that
\[
\eta_{AB} = -\frac{3}{4}\Psi_{ABCD} \kappa^{CD}.
\]
For later use we also define
\[
\mathcal{H}^2 \equiv \mathcal{H}_{ab}\mathcal{H}^{ab} =8 \eta_{AB}\eta^{AB}.
\]

\medskip
The \emph{Ernst form} of the Killing vector $\xi^a$ is defined as
\[
\chi_a = 2 \xi^b\mathcal{H}_{ba}.
\]
It is well-known that in vacuum, the Ernst form closed and thus,
locally exact. Therefore, there exists a scalar, the \emph{Ernst
  potential} $\chi$, that satisfies
\[
\chi_a = \nabla_a \chi.
\]
A calculation then readily yields that
\[
\chi_{AA'} = 4 \eta_{AB} \xi^B{}_{A'} = 3 \kappa^{CF}\Psi_{ABCF}
\nabla_{DA'} \kappa^{DB}.
\]

\subsection{Mars's characterisation of the Kerr spacetime}

In \cite{Mar00} is has been shown that the Kerr spacetime can be
characterised in terms of an alignment condition of the Weyl tensor
and the Killing form of the stationary Killing vector of the
spacetime. This invariant characterisation admits both  local and 
semi-global versions. In \cite{ColVal16a} it has been shown that the
alignment condition follows if the spacetime is assumed to have a
Killing spinor. More precisely, one has the following:

\begin{theorem}
\label{Theorem:SpinorialCharacterisationKerr}
Let $(\mathcal{M},\bmg)$ denote a smooth vacuum spacetime endowed with
a Killing spinor $\kappa_{AB}$ satisfying $\kappa_{AB}\kappa^{AB}\neq 0$,
such that the spinor $\xi_{AA'}\equiv \nabla^B{}_{A'} \kappa_{AB}$ is
Hermitian. Then there exist two complex constants $\mathfrak{l}$ and
$\mathfrak{c}$ such that 
\begin{equation}
\mathcal{H}^2 = - \mathfrak{l}( \mathfrak{c}-\chi)^4.
\label{Condition:Mars}
\end{equation}
If, in addition, $\mathfrak{c}=1$ and $\mathfrak{l}$ is real positive,
then $(\mathcal{M},\bmg)$ is locally isometric to the Kerr spacetime. 
\end{theorem}

\section{The characteristic initial value problem on expansion and
  shear-free hypersurfaces}
\label{Section:CharacteristicInitialValueProblem}

In \cite{Rac14}, by adopting and slightly generalising results of
\cite{Fri81a}, a systematic analysis of the null characteristic
initial value problem for the Einstein-Maxwell equations in terms of
the Newman-Penrose formalism was carried out. In particular, it was
shown how to obtain a system of reduced evolution equations forming a
first order symmetric hyperbolic system of equations. Moreover, it was
shown that the solutions to these evolution equations imply, in turn,
a solution to the full Einstein-Maxwell system provided that the inner
(constraint) equations on the initial null hypersurfaces hold. For
this type of setting, the theory for the characteristic initial value
problem developed in \cite{Ren90} applies and ensures the local
existence and uniqueness of a solution of the reduced evolution equations.

 The general results described in the previous paragraph were then
used to investigate electrovacuum spacetimes $(\mathcal{M},\bmg,{\bm
F})$ possessing a pair of null hypersurfaces $\mathcal{H}_1$ and
$\mathcal{H}_2$ generated by expansion and shear-free geodesically
complete null congruences, with intersection on a two dimensional spacelike
hypersurface $\mathcal{Z}\equiv \mathcal{H}_1\cap \mathcal{H}_2$. The
configuration formed by $\mathcal{H}_1$ and $\mathcal{H}_2$ constitute
a \emph{bifurcate horizon}. In general, the freely specifiable data on
$\mathcal{Z}$ do not possess any symmetry in addition to the horizon
Killing vector (implied by the non-expanding character of these
horizons). Thus, these spacetimes constitute the generic class of 
\emph{stationary} distorted electrovacuum spacetimes.  The key
observation resulting from the analysis in \cite{Rac14} is, for the
\emph{vacuum case}, summarised in the following:

\begin{theorem}
\label{Theorem:ExistenceCIVPSpacetime}
Suppose that $(\mathcal{M},\bmg)$ is a vacuum spacetime with a
vanishing Cosmological constant possessing a pair of null
hypersurfaces $\mathcal{H}_1$ and $\mathcal{H}_2$ generated by
expansion and shear-free geodesically complete null congruences,
intersecting on a 2-dimensional spacelike hypersurface
$\mathcal{Z}\equiv \mathcal{H}_1 \cap \mathcal{H}_2$. Then, the metric
$\bmg$ is uniquely determined (up to diffeomorphisms) on 
a neighbourhood $\mathcal{O}$ of $\mathcal{Z}$ contained in
the domain of
dependence $D(\mathcal{H}_1 \cap \mathcal{H}_2)$ of $\mathcal{H}_1$
and $\mathcal{H}_2$, once a complex vector field $\zeta^{\mathcal{A}}$
determining the induced metric ${\bm
\sigma}$ on $\mathcal{Z}$ and the spin connection coefficient $\tau$
are specified on $\mathcal{Z}$.
\end{theorem}

\subsection{Summary of the construction}
In the remainder of this article we will require further information
concerning the construction in \cite{Rac14}. Throughout, let
$(\mathcal{M},\bmg)$ denote a vacuum spacetime and let $\mathcal{H}_1$
and $\mathcal{H}_2$ denote two null hypersurfaces in
$(\mathcal{M},\bmg)$ intersecting on a spacelike 2-surface
$\mathcal{Z}$.

\begin{remark}
{\em In the remaining of this section, the topology of $\mathcal{Z}$ will
not be relevant in the discussion. The situation will, however, change
when we attempt to single out the Kerr spacetime.}
\end{remark}

Let $n^a$ denote a smooth future-directed null vector
on $\mathcal{Z}$ tangent to $\mathcal{H}_2$, which is extended to
$\mathcal{H}_2$ by requiring it to satisfy $n^b \nabla_b n_a=0$ on
$\mathcal{H}_2$. Moreover, let $u$ be an affine parameter along the
null generators of $\mathcal{H}_2$, so that $u=0$ on $\mathcal{Z}$
and $\mathcal{Z}_u$ are the associated 1-parameter family of smooth cross
sections of $\mathcal{H}_2$. We choose a further null vector $l^a$ as
the unique future-directed null vector field on $\mathcal{H}_2$ which
is orthogonal to the 2-dimensional cross sections $\mathcal{Z}_u$ and
satisfies the normalisation condition $n_{a}l^{a}=1$. Consider now the
null geodesics starting on $\mathcal{H}_2$ with tangent
$l^a$. Since $\mathcal{H}_2$ is assumed to be smooth and the vector
fields $n^a$ and $l^a$ are smooth on $\mathcal{H}_2$ by construction,
these geodesics do not intersect in a sufficiently small open
neighbourhood $\mathcal{O}\subset \mathcal{M}$ of $\mathcal{H}_2$. Let
now $r$ denote the affine parameter along the null geodesics starting
on $\mathcal{H}_2$ with tangent $l^a$, chosen such
that $r=0$ of $\mathcal{H}_2$. By construction one has that $l^a
=(\partial/\partial r)^a$. The affine parameter defines a smooth
function $r:\mathcal{O}\rightarrow \mathbb{R}$. The function
$\mathcal{H}_2\rightarrow \mathbb{R}$ defined by the affine parameter
of the integral curves of $n^a$ can be extended to a smooth function
$u:\mathcal{O}\rightarrow\mathbb{R}$ by requiring it to be constant
along the null geodesics with tangent $l^a$. 

The construction of the previous paragraph is complemented by choosing
suitable coordinates $(x^{\mathcal{A}})$ on patches of $\mathcal{Z}$
and extending them to $\mathcal{O}$ by requiring them to be constant
along the integral curves of the vectors $l^a$ and $n^a$. In this
manner one obtains a system of \emph{Gaussian null coordinates}
$(u,r,x^{\mathcal{A}})$ on patches of $\mathcal{O}$. In
each of these patches the spacetime metric $\bmg$ takes the form
\begin{equation}
\bmg = g_{00} \mathbf{d}u\otimes \mathbf{d}u +( \mathbf{d}u\otimes
\mathbf{d}r+\mathbf{d}r\otimes\mathbf{d}u) + g_{0\mathcal{A}} (
\mathbf{d}u\otimes\mathbf{d}x^{\mathcal{A}} + \mathbf{d}x^{\mathcal{A}}\otimes\mathbf{d}u) +
g_{\mathcal{A}\mathcal{B}}
\mathbf{d}x^{\mathcal{A}}\otimes\mathbf{d}x^{\mathcal{B}},
\label{LineElement}
\end{equation}
where $g_{00}$, $g_{0\mathcal{A}}$, $g_{\mathcal{A}\mathcal{B}}$ are
smooth functions of the coordinates $(u,r,x^{\mathcal{A}})$ such that 
\begin{equation}
g_{00}= g_{0\mathcal{A}}=0, \quad \mbox{on} \quad \mathcal{H}_2,
\label{Metric:BehaviourH2}
\end{equation}
and
$g_{\mathcal{A}\mathcal{B}}$ is a negative definite $2\times 2$
matrix. Observe that by construction in $\mathcal{O}$ the $u=0$ and $r=0$ hypersurfaces coincide with $\mathcal{H}_1$ and $\mathcal{H}_2$, respectively. 

\medskip

In the following it will be convenient to consider the components of
the contravariant form of the metric associated to the line element
\eqref{LineElement}. A calculation shows that components of the contravariant metric $g^{ab}$ in the {Gaussian null coordinates}
$(u,r,x^{\mathcal{A}})$ can be given as 
\[
g^{ab} \rightleftharpoons
\left(
\begin{array}{ccc}
0 & 1 & 0 \\
1 & g^{11} & g^{1\mathcal{B}} \\
0 & g^{\mathcal{A}1} & g^{\mathcal{A}\mathcal{B}}
\end{array}
\right).
\]
The metric functions $g^{11}$, $g^{1\mathcal{A}}$ and $
g^{\mathcal{A}\mathcal{B}}$ can be conveniently parametrised in terms
of real-valued functions $U$, $X^{\mathcal{A}}$ and complex-valued
functions $\omega$, $\zeta^{\mathcal{A}}$ on $\mathcal{O}$ such that 
\[
g^{11}=2(U-\omega\overline{\omega}), \qquad g^{1\mathcal{A}} =
X^{\mathcal{A}}-(\overline{\omega}\zeta^{\mathcal{A}} +\omega
\overline{\zeta}^{\mathcal{A}}), \qquad g^{\mathcal{A}\mathcal{B}}
=-(\zeta^{\mathcal{A}} \overline{\zeta}^{\mathcal{B}} + \zeta^{\mathcal{B}} \overline{\zeta}^{\mathcal{A}}).
\]
Accordingly, setting
\[
l^a = (\partial_r)^a, \qquad n^a =  (\partial_u)^a + U\,
(\partial_r)^a + X^{\mathcal{A}}\, (\partial_{x^\mathcal{A}})^a, \qquad
m^a = \omega (\partial_r)^a + \zeta^{\mathcal{A}}\, (\partial_{x^\mathcal{A}})^a,
\]
one obtains a complex (NP) null tetrad $\{l^a,\, n^a,\, m^a,\, \overline{m}^a\}$
in $\mathcal{O}$. As a result of the conditions in
\eqref{Metric:BehaviourH2} one has that
\[
U = X^{\mathcal{A}}=\omega =0, \quad \mbox{on} \quad \mathcal{H}_2.
\]
It follows from the previous discussion that $m^a$ and $\overline{m}^a$ are
everywhere tangent to the sections $\mathcal{Z}_u$ of
$\mathcal{H}_2$. In general, the complex null vectors $m^a$ and
$\overline{m}^a$ are not parallelly propagated along the null generators of
$\mathcal{H}_2$. 

Associated to the NP null tetrad $\{l^a,\, n^a,\, m^a,\, \overline{m}^a\}$
in $\mathcal{O}$ one has the directional derivatives
\begin{eqnarray*}
&& D = \frac{\partial}{\partial r}, \\
&& \Delta = \frac{\partial}{\partial u} + U\frac{\partial}{\partial
   u}+ X^{\mathcal{A}}\frac{\partial}{\partial x^{\mathcal{A}}},\\
&& \delta = \omega \frac{\partial}{\partial r} + \zeta^{\mathcal{A}}
   \frac{\partial}{\partial x^{\mathcal{A}}}.
\end{eqnarray*}

\begin{remark}
{\em By construction, one has that $D$ is an intrinsic derivative to
  $\mathcal{H}_1$ pointing along the null generators of this
  hypersurface. Similarly, $\Delta$ is intrinsic to $\mathcal{H}_2$
  and points in the direction of its null generators. Finally, $\{
  m^a,\, \overline{m}^a\}$ are differential operators which on
  $\mathcal{H}_2$ are intrinsic to the sections of constant $u$,
  $\mathcal{Z}_u$. Observe, however, that while $\delta$ restricted to
$\mathcal{H}_1$ is still intrinsic to the null hypersurface, it is not
intrinsic to the sections of constant $r$.}
\end{remark}

The NP null tetrad constructed in the previous paragraph can be specialised further to simplify the associate spin-connection
coefficients. By parallelly propagating $\{l^a,\, n^a,\, m^a,\,
\overline{m}^a\}$ along the null geodesics with tangent $l^a$ one finds that
\begin{subequations}
\begin{eqnarray}
&\kappa=\pi =\epsilon=0,& \label{NPGaugeCondition1}\\
&\rho= \overline{\rho}, \qquad \tau=\overline{\alpha} +\beta, \qquad
  \mbox{everywhere on }\quad \mathcal{O}.& \label{NPGaugeCondition2}
\end{eqnarray} 
\end{subequations}
Moreover, from the condition $n^b\nabla_b n^a=0$ on $\mathcal{H}_2$ it
follows that 
\begin{equation}
\nu =0 \qquad \mbox{on} \quad \mathcal{H}_2.
\label{NPGaugeCondition3}
\end{equation}
Also, using that $u$ is an affine parameter of the generators of
$\mathcal{H}_2$ one finds that $\gamma+\overline{\gamma}=0$ along these
generators. One can specialise further by suitably rotating the
vectors $\{ m^a,\, \overline{m}^a \}$ so as to obtain 
\begin{equation}
\gamma=0, \qquad \mbox{on} \quad \mathcal{H}_2.
\label{NPGaugeCondition4}
\end{equation}

\subsubsection{Solving the NP constraint equations}
The NP Ricci and Bianchi identities split into a subset of intrinsic
(constraint) equations  to $\mathcal{H}_1\cup\mathcal{H}_2$ and a
subset of transverse (evolution) equations. In \cite{Rac14} the gauge
introduced in the previous subsection was used to systematically
analyse the constraint equations on $\mathcal{H}_1\cup\mathcal{H}_2$
with the aim of identifying the freely specifiable data on this pair
of intersecting
hypersurfaces under the assumption that it is expansion and
shear-free. The results from this analysis can be conveniently
presented in the form of a table ---see Table \ref{Table:CharacteristicInitialData}. 
{
	\renewcommand{\arraystretch}{1.5}
\begin{table}[t]
\centering \small \hskip-.15cm
\begin{tabular}{|l|l|l|} 
\hline ${{\mathcal{H}}_1}$ &  ${{\mathcal{Z}}}$
$\phantom{\frac{\frac12}{A}}$ & ${{\mathcal{H}}_2}$ \\ \hline \hline
 $\mathrm{D}\zeta^{\mathcal{A}} =0$ &  $\zeta^{\mathcal{A}}$ $\phantom{\frac{\frac12}{A}}$ \hfill (data)  & $\Delta\zeta^{\mathcal{A}} =0$ \\  \hline 
 $\omega =-r\,\tau$ &  $\omega =0$  $\phantom{\frac{\frac12}{A}}$ \hfill  & $\omega = 0$  \hfill (geometry) \\  \hline 
 $X^{\mathcal{A}}=r\,[\tau\,\overline{\zeta}^{\mathcal{A}}+\overline{\tau}\,\zeta^{\mathcal{A}}]$ $\phantom{\frac{\frac12}{A}}$ &  $X^{\mathcal{A}}=0$ $\phantom{\frac{\frac12}{A}}$ \hfill   & $X^{\mathcal{A}}=0$ \hfill (geometry) \\  \hline 
 $U = -r^2\,\left[\,2\,\tau\,\overline{\tau}+\frac12\left(\Psi_2+\overline{\Psi}_2\right)\right]$ &  $U =0$ $\phantom{\frac{\frac12}{A}}$ \hfill   &  $U = 0$ \hfill (geometry) \\  \hline
 $\rho =0$ &  $\rho =0$ $\phantom{\frac{\frac12}{A}}$ & $\rho = u\, \left(\,\overline{\delta}\tau
 - 2\,\alpha\,\tau - \Psi_2\right)$ \\  \hline 
 $\sigma=0$  &
 $\sigma=0$ $\phantom{\frac{\frac12}{A}}$ &
 $\sigma=u \, \left(\,\delta\tau - 2\,\beta\,\tau\,\right)$ \\ \hline 
 $\mathrm{D}\tau =0$ &  $\tau$ $\phantom{\frac{\frac12}{A}}$ \hfill (data)  & $\Delta\tau =0$ \\  \hline 
 $\mathrm{D}\alpha=\mathrm{D}\beta=0$ & $\alpha,\beta, \tau=\overline{\alpha}+\beta$ $\phantom{\frac{\frac12}{A}}$ &
  $\Delta \alpha=\Delta \beta=0$ \\ \hline
  $\gamma=r \, \left(\,\tau\,\alpha+ \overline{\tau}\,\beta +\Psi_2\,\right)$ & $\gamma=0$ $\phantom{\frac{\frac12}{A}}$ \hfill  & $\gamma=0$  \hfill (gauge) \\ \hline 
 $\mu = r \,\Psi_2 $ & $\mu =0$ $\phantom{\frac{\frac12}{A}}$ & $\mu = 0$  \\ \hline
 $\lambda=0$ & $\lambda=0$ $\phantom{\frac{\frac12}{A}}$ & $\lambda=0$ \\ \hline
 $\nu=\frac{1}{2}\,r^2\,\left(\,\overline{\delta}\Psi_2 + \overline{\tau}\,\Psi_2\,\right)$ &  $\nu=0$ $\phantom{\frac{\frac12}{A}}$ \hfill  & 
 $\nu=0$ \hfill (gauge) \\ \hline
 $\Psi_0=0$ & $\Psi_0=0$ $\phantom{\frac{\frac12}{A}}$ & $\Psi_0=\frac{1}{2}\,u^2\,(\delta^2\Psi_2 - (7\,\tau+2\,\beta)\,\delta\Psi_2 + 12\,\tau^2 \Psi_2)$ \\ \hline
  $\Psi_1=0$ &
 $\Psi_1=0$ $\phantom{\frac{\frac12}{A}}$ & $\Psi_1=u\,\left(\,\delta\Psi_2 -3\,\tau \,\Psi_2\,\right)$ \\ \hline
 $\mathrm{D} \Psi_2 =0$ &
 $\zeta^A, \tau$ \hskip-.4cm $\phantom{\frac{\frac12}{A}}$
 $\rightarrow \ \alpha,\beta,\Psi_2$ & $\Delta\Psi_2=0$ \\ \hline
 $\Psi_3=r\,\overline{\delta}\Psi_2 $& $\Psi_3=0$  & $\Psi_3=0$ 
  \\ \hline 
 $\Psi_4=\frac{1}{2}\,r^2\,\left(\,\overline{\delta}^2\Psi_2 +2\,\alpha\,\overline{\delta}\Psi_2\,\right)$ & $\Psi_4=0$  & $\Psi_4=0$ 
  \\ \hline
\end{tabular}
\caption{\small The full initial data set on the intersecting null
  hypersurfaces $\mathcal{H}_1\cup\mathcal{H}_2$. }
\label{Table:CharacteristicInitialData}
\end{table}
}

\begin{remark}
\label{Remark:GeometryZ}
{\em As already mentioned, in what follows we will mostly be
  interested in the situation where $\cZ$ is diffeomorphic to a unit
  2-sphere, i.e.~$\cZ\approx \mathbb{S}^2$. From the definition of the
  operators $\eth$ and $\overline{\eth}$ as given in \eqref{eths},
  along with those of the NP spin connection coefficients $\alpha$ and
  $\beta$, it follows that the connection on $\mathcal{Z}$ is encoded
  in the combination $\overline\alpha-\beta$.  As discussed in
  \cite{Rac14}, given the freely specifiable data $\zeta^A$ and $\tau$
  one can readily compute the NP coefficients $\alpha$,
  $\beta$. These, in turn, can be used, together with the NP Ricci
  equation
\begin{equation}
\Psi_2 =-\delta{\alpha} + \overline{\delta} \beta +
\alpha\,\overline{\alpha} -2\, \alpha\, \beta +
\beta\,\overline{\beta}
\label{Psi2}
\end{equation}
to determine the Weyl spinor component $\Psi_2$ on $\cZ$. From the latter it is straightforward to deduce
\begin{equation}
2\,\mbox{Re}( \Psi_2)=\Psi_2+\overline\Psi_2 = -\delta\,(\,{\alpha}-\overline\beta\,) - \overline{\delta} \,(\,\overline{\alpha}-\beta \,) + 2\,(\,{\alpha}-\overline\beta\,)\,(\,\overline{\alpha}-\beta \,)
\label{RePsi2}
\end{equation}	
which implies that the real part of $\Psi_2$ ---in accordance with the
fact that $-2 \,\mbox{Re}( \Psi_2)$ is the Gaussian curvature
$K_{\mathcal{G}}$ of $\cZ$ (see, i.e.~Proposition 4.14.21 in
\cite{PenRin84})--- depends only on the combination
$\overline\alpha-\beta$ which is completely intrinsic to
$\cZ$. Analogously, by making use of \eqref{Psi2}  $\tau$ and the
imaginary part of $\Psi_2$ can be seen to be closely related to each other via
\begin{equation}
2\,\mbox{i\,Im}( \Psi_2)=\Psi_2-\overline\Psi_2=\overline\delta\tau - \delta \overline\tau - 2\,\left(\, \beta\,\overline\tau - \overline\beta\,\tau \,\right) \,.
\label{ImPsi2}
\end{equation}
}
\end{remark}

\section{The Killing spinor data conditions for a characteristic
  initial problem}
\label{Section:KillingSpinorData}
In this section we adapt the analysis of Killing spinor initial data
in \cite{BaeVal10b} to the setting of a
characteristic initial data set ---see also
\cite{GarVal08a}. 

\subsection{Construction of the Killing
  spinor candidate}
\label{Section:KillingSpinorCandidate}
In this section we investigate the characteristic initial value problem for
the wave equation, equation \eqref{KillingSpinorWaveEquation},  governing the evolution of the Killing spinor
candidate $\kappa_{AB}$. An approach to the formulation of the characteristic initial value problem for
wave equations on intersecting null hypersurfaces $\mathcal{H}_1$ and
$\mathcal{H}_2$ has been analysed in \cite{Ren90}. Our discussion
follows the ideas of this analysis closely.

\subsubsection{Basic set-up}
Let $\{o^A,\iota^A\}$ denote a spin dyad normalised according to $o_A
\iota^A=1$. The spinor $\kappa_{AB}$ can be written as
\[
\kappa_{AB} = \kappa_2 o_A o_B - 2\kappa_1 o_{(A} \iota_{B)} +
\kappa_0 \iota_A \iota_B.
\]
so that
\[
\kappa_0 \equiv \kappa_{AB} o^A o^B, \qquad \kappa_1 \equiv
\kappa_{AB} o^A \iota^B, \qquad \kappa_2\equiv \kappa_{AB}
\iota^A\iota^B.
\]
It can be readily verified that the scalars $\kappa_2$, $\kappa_1$ and
$\kappa_0$ have, respectively, spin weights $-1,\, 0, 1$ --- i.e. they
transform as
\[
\kappa_j \mapsto  e^{-2(j-1)\mbox{i}\vartheta}\kappa_j
\]
under a rotation $\{o^A,\iota^A\}\mapsto \{ e^{\mbox{i}\vartheta}o^A,\, e^{-\mbox{i}\vartheta}\iota^A\}$.

\medskip
A direct decomposition of the wave equation
\eqref{KillingSpinorWaveEquation} using the NP formalism readily
yields the following equations for the independent components
$\kappa_0,\, \kappa_1$ and $\kappa_2$ of the spinor $\kappa_{AB}$:
\begin{subequations}
\begin{eqnarray}
&& D\Delta \kappa_2 + \Delta D\kappa_2 -\delta\overline{\delta}\kappa_2 -
\overline{\delta}\delta \kappa_2 \nonumber \\
&& \hspace{1cm} + (\mu+\overline{\mu} +3 \gamma-\overline{\gamma}) D\kappa_2
-(\rho+\overline{\rho})\Delta \kappa_2 +(\overline{\tau}-3\alpha
-\overline{\beta})\delta \kappa_2 +(\overline{\alpha}-5\beta +\tau) \overline{\delta}
\kappa_2 \nonumber\\
&& \hspace{1cm} +(\Psi_2 + 2\alpha\overline{\alpha} -8 \alpha\beta -2\beta
\overline{\beta} -2\gamma\rho + 2 \mu\rho -2\gamma\overline{\rho} + 2\lambda
\sigma + 2 \alpha \tau + 2\beta \overline{\tau} + 2 D\gamma-2\delta \alpha
-2 \overline{\delta}\beta ) \kappa_2 \nonumber\\
&& \hspace{1cm}+(\Psi_4 -4\lambda\mu)\kappa_0=0, \label{KillingSpinorWaveEquationNP1}\\
&& D\Delta \kappa_1 + \Delta D\kappa_1-\delta\overline{\delta}\kappa_1
-\overline{\delta}\delta \kappa_1 \nonumber\\
&& \hspace{1cm}-2\tau D\kappa_2
+(\mu+\overline{\mu}-\gamma-\overline{\gamma}) D\kappa_1 + 2\nu D\kappa_0
-(\rho+\overline{\rho})\Delta\kappa_1 + 2\rho \delta\kappa_2
+(\alpha-\overline{\beta}+\overline{\tau})\delta \kappa_1 \nonumber\\
&& \hspace{1cm}-2\lambda \delta
\kappa_2 + 2\sigma \overline{\delta}\kappa_2 + (\overline{\alpha}-\beta +\tau
)\overline{\delta}\kappa_1-2\mu \overline{\delta}\kappa_0 \nonumber\\
&& \hspace{1cm} +(-\Psi_1 -\overline{\alpha}\rho + 3\beta\rho +
\alpha\sigma + \overline{\beta}\sigma \overline{\rho}\tau -\sigma\overline{\tau}-D\tau
+\delta\rho \overline{\delta}\sigma)\kappa_2 \nonumber\\
&& \hspace{1cm}+(-\Psi_3 +\overline{\alpha}\lambda +\beta \lambda +
3\alpha\mu -\overline{\beta} \mu -\nu \rho -\nu\overline{\rho} +\lambda\tau
+\mu\overline{\tau} + D\nu -\delta \lambda -\overline{\delta}\mu)\kappa_0 =0, \label{KillingSpinorWaveEquationNP2}\\
&& D\Delta \kappa_0 + \Delta D\kappa_0-\delta\overline{\delta} \kappa_0
-\overline{\delta}\delta \kappa_0 \nonumber\\
&& \hspace{1cm} + (\mu+\overline{\mu} -5 \gamma-\overline{\gamma})
D\kappa_0 -(\rho+\overline{\rho}) \Delta \kappa_0 +(5\alpha -\overline{\beta} +
\overline{\tau})\delta \kappa_0 +(\overline{\alpha}+3\beta +
\tau)\overline{\delta}\kappa_0 \nonumber\\
&& \hspace{1cm} + (\Psi_2 -2 \alpha\overline{\alpha} -8 \alpha\beta +
2\beta\overline{\beta} +2 \gamma \rho +2 \mu \rho + 2\gamma\overline{\rho} + 2
\lambda \sigma -2 \alpha \tau -2 \beta \overline{\tau} -2 D\gamma + 2\delta
\alpha + 2 \overline{\delta}\beta ) \kappa_0 \nonumber\\
&& \hspace{1cm} + (\Psi_0-4\rho\sigma )\kappa_2=0. \label{KillingSpinorWaveEquationNP3}
\end{eqnarray}
\end{subequations}
The above expressions are completely general: no assumption on the
spacetime (other than satisfying the vacuum field equations) or the
gauge has been made.

\begin{remark}\label{tetrad}
{\em  In the sequel we investigate the consequences of
these equations on the hypersurfaces $\mathcal{H}_1$ and
$\mathcal{H}_2$. For this we consider a spin dyad $\{ o^A,\,\iota^A  \}$
adapted to the NP null tetrad $\{ l^a,n^a,\, m^a, \overline{m}^a \}$ ---if
$\{ l^{AA'},n^{AA'},\, m^{AA'}, \overline{m}^{AA'} \}$ denote the
spinorial counterparts of the null tetrad, one has the correspondences
\[
l^{AA'} = o^A \overline{o}^{A'}, \qquad n^{AA'} = \iota^A \overline{\iota}^{A'},
\qquad m^{AA'} =o^A\overline{\iota}^{A'}, \qquad \overline{m}^{AA'} =\iota^A \overline{o}^{A'},
\]
and the gauge conditions
\eqref{NPGaugeCondition1}-\eqref{NPGaugeCondition2},
\eqref{NPGaugeCondition3} and \eqref{NPGaugeCondition4} hold when
computing the corresponding NP spin-connection coefficients by means
of derivatives of the spin dyad.}
\end{remark}

\subsubsection{The transport equations on $\mathcal{H}_1$}
Consider now the restriction of equations
\eqref{KillingSpinorWaveEquationNP1}-\eqref{KillingSpinorWaveEquationNP3}
to the null hypersurface $\mathcal{H}_1$ with tangent $l^a$. It
follows then that $D$ is a directional derivative along the null
generators of $\mathcal{H}_1$, while $\Delta$ is a
directional derivative transversal to $\mathcal{H}_1$. Using the
NP commutator $[D,\Delta]$ equation to rewrite $\Delta D \kappa_0$,
$\Delta D \kappa_1$, $\Delta D \kappa_2$ in terms of $D \Delta
\kappa_0$, $D \Delta \kappa_1$ and $D \Delta \kappa_2$, equations
\eqref{KillingSpinorWaveEquationNP1}-\eqref{KillingSpinorWaveEquationNP3}
take the form:
\begin{subequations}
\begin{eqnarray}
&& 2 D \Delta \kappa_0 - \delta \overline{\delta} \kappa_0 -\overline{\delta}
\delta \kappa_0 +( \overline{\alpha} + 3\beta ) \overline{\delta}\kappa_0  +
(5\alpha -\overline{\beta})\delta \kappa_0 +(\mu+\overline{\mu}-4\gamma)
D\kappa_0 + 4\tau D\kappa_1 \nonumber \\
&& \hspace{1cm} + 2 \kappa_1 D\tau +(\Psi_2 -2 \alpha\overline{\alpha} - 8 \alpha\beta +
2\beta\overline{\beta} -2 \alpha \tau -2\beta \overline{\tau} - 2 D\gamma +
2\delta \alpha + 2 \overline{\delta}\beta) \kappa_0=0,\label{KillingSpinorWaveEquationNP1H1}\\
&& 2 D \Delta \kappa_1 - \delta \overline{\delta} \kappa_1
-\overline{\delta}\delta \kappa_1 -2 \nu D \kappa_0 + (\mu +
\overline{\mu})D\kappa_1 + 2 \tau D\kappa_2 + (\alpha -\overline{\beta})
\delta\kappa_1 + 2 \mu \overline{\delta}\kappa_0 +
(\overline{\alpha}-\beta)\overline{\delta}\kappa_1 \nonumber \\
&& \hspace{1cm} + (\Psi_3 -3 \alpha\mu + \overline{\beta}\mu - \mu
\overline{\tau
} -D\nu + \overline{\delta}\mu)\kappa_0 -2 \Psi_2 \kappa_1 + \kappa_2
D\tau=0,\label{KillingSpinorWaveEquationNP2H1}\\
&& 2 D \Delta \kappa_2 -\delta \overline{\delta} \kappa_2
-\overline{\delta}\delta \kappa_2 -4 \nu D\kappa_1 +(4\gamma+
\mu+\overline{\mu})D\kappa_2 -(3\alpha + \overline{\beta})\delta\kappa_2 + 4 \mu
\overline{\delta}\kappa_1 +(\overline{\alpha}-5\beta)\overline{\delta}\kappa_2
\nonumber \\
&& \hspace{1cm} + (\Psi_2 + 2 \alpha\overline{\alpha} - 8 \alpha\beta -2
\beta \overline{\beta} + 2 \alpha \tau + 2\beta\overline{\tau} + 2 D\gamma -2
\delta\alpha -2 \overline{\delta}\beta)\kappa_2 \nonumber \\
&& \hspace{1cm} +(2\alpha \mu -2 \Psi_3 + 2\overline{\beta}\mu -2\mu
\overline{\tau} -2 D\nu + 2\overline{\delta}\mu) \kappa_1 + \Psi_4 \kappa_0=0.
\label{KillingSpinorWaveEquationNP3H1}
\end{eqnarray}
\end{subequations}
If the value of the components $\kappa_0,\, \kappa_1,\, \kappa_2$ are
known on $\mathcal{H}_1$, then the above equations can be read as a
system of ordinary differential equations for the transversal derivatives
\[
\Delta \kappa_0, \quad \Delta \kappa_1, \quad \Delta \kappa_2,
\]
 along the null generators of
$\mathcal{H}_1$. Initial data for these transport equations is
naturally prescribed on $\mathcal{Z}$. 

\subsubsection{The transport equations on $\mathcal{H}_2$}
Similarly, one can consider the restriction of equations
\eqref{KillingSpinorWaveEquationNP1}-\eqref{KillingSpinorWaveEquationNP3}
to the null hypersurface $\mathcal{H}_2$ with tangent $n^a$. Thus,
$\Delta$ is a directional derivative along the null generators of
$\mathcal{H}_2$, $\delta$ and $\overline{\delta}$ are intrinsic derivatives
while $D$ is transversal to $\mathcal{H}_2$.  In this
case one uses the NP commutator $[D,\Delta]$ to rewrite $D
\Delta\kappa_0$, $D \Delta\kappa_1$, $D \Delta\kappa_2$ in terms of
$\Delta D \kappa_0$, $\Delta D \kappa_1$, $\Delta D \kappa_2$ and
lower order terms so that equations
\eqref{KillingSpinorWaveEquationNP1}-\eqref{KillingSpinorWaveEquationNP3}
take the form
\begin{subequations}
\begin{eqnarray}
&&2 \Delta D \kappa_0 -\delta\overline{\delta}\kappa_0 -\overline{\delta} \delta
\kappa_0 -(\rho + \overline{\rho}) \Delta \kappa_0 + 4\tau D\kappa_1 +
(5\alpha -\overline{\beta} + 2\overline{\tau}) \delta\kappa_0 + (\overline{\alpha}+
3\beta + 2\tau)\overline{\delta}\kappa_0 \nonumber \\
&& \hspace{1cm} +
4\sigma \overline{\delta}\kappa_1 -4\rho
\delta\kappa_1 +(\Psi_2 - 2\alpha\overline{\alpha} -8\alpha\beta +2
\beta\overline{\beta} -2 \alpha \tau -2 \beta \overline{\tau} + 2 \delta \alpha
+ 2\overline{\delta}\beta)\kappa_0 \nonumber \\
&& \hspace{1cm} + (2\overline{\alpha}\rho + 2\beta \rho + 6\alpha\sigma
-2\overline{\beta}\sigma -2 \overline{\rho}\tau + 2\sigma\overline{\tau} + 2D\tau -2
\delta\rho -2\overline{\delta}\sigma -2\Psi_1) \kappa_1\nonumber \\
&& \hspace{1cm}+
(\Psi_0-4\rho\sigma)\kappa_2=0,\label{KillingSpinorWaveEquationNP1H2}\\
&& 2 \Delta D \kappa_1 -\delta\overline{\delta}\kappa_1 -\overline{\delta} \delta
\kappa_1-(\rho+\overline{\rho})\Delta\kappa_1 +2\tau D\kappa_2
+(\alpha-\overline{\beta}+2\overline{\tau})\delta\kappa_1 + (\overline{\alpha}-\beta+2\tau)\overline{\delta}\kappa_1\nonumber\\
&& \hspace{1cm} -2 \rho \delta
\kappa_2 +2 \sigma \overline{\delta}\kappa_2-2\Psi_2 \kappa_1 \nonumber\\
&& \hspace{1cm}+(\Psi_1 -\overline{\alpha}\rho
-3\beta\rho -\alpha\sigma -\overline{\beta}\sigma -\overline{\rho}\tau +
\sigma\overline{\tau} + D\tau -\delta\rho -\overline{\delta}\sigma)\kappa_2=0, \label{KillingSpinorWaveEquationNP2H2}
\\
&& 2 \Delta D \kappa_2 -\delta\overline{\delta}\kappa_2 -\overline{\delta} \delta
\kappa_2-(\rho+\overline{\rho})\Delta\kappa_2
+(2\overline{\tau}-3\alpha-\overline{\beta}) \delta\kappa_2 +
(\overline{\alpha}-5\beta +2\tau)\overline{\delta}\kappa_2 \nonumber \\
&& \hspace{1cm} +(\Psi_2 + 2\alpha\overline{\alpha} -8\alpha\beta -2
\beta\overline{\beta} +2 \alpha\tau + 2\beta\overline{\tau} -2 \delta\alpha -2
\overline{\delta}\beta)\kappa_2=0.
\label{KillingSpinorWaveEquationNP3H2}
\end{eqnarray}
\end{subequations}
If the values of $\kappa_0$, $\kappa_1$, $\kappa_2$ are known on
$\mathcal{H}_2$ then the above equations can be read as a system
of ordinary differential equations for the transversal derivatives
\[
D\kappa_0, \qquad D\kappa_1, \qquad D\kappa_2,
\]
along the null generators of $\mathcal{H}_2$. Initial data for these
transport equations is naturally prescribed on $\mathcal{Z}$.

\subsubsection{Summary: existence of the Killing spinor candidate}

The discussion of the previous subsections combined with the methods
of \cite{Ren90} ---see also \cite{Kan96b}---  allows to formulate the
following existence result:

\begin{proposition}
Let $(\mathcal{M},\bmg)$ denote a spacetime satisfying the assumptions
of Theorem \ref{Theorem:ExistenceCIVPSpacetime}. Then, given a smooth
choice of fields $\kappa_0$, $\kappa_1$ and $\kappa_2$ on
$\mathcal{H}_1\cup \mathcal{H}_2$, there exists a neighbourhood
$\mathcal{O}$ of $\mathcal{Z}$ in $D(\mathcal{H}_1\cup\mathcal{H}_2)$
on which the wave equation \eqref{KillingSpinorWaveEquation} has a
unique solution $\kappa_{AB}$.
\label{ExistenceKSWaveEquation}
\end{proposition}

\begin{remark}
{\em The assumption of smoothness of the fields $\kappa_0$, $\kappa_1$ and
$\kappa_2$ require, in particular, that the limits of these fields as
one approaches to $\mathcal{Z}$ on either $\mathcal{H}_1$ or
$\mathcal{H}_2$ coincide.}
\end{remark}

\subsection{The NP decomposition of the Killing spinor data conditions}

The conditions on the initial data for the Killing spinor candidate
$\kappa_{AB}$ constructed in the previous section which ensure that it
is, in fact, a Killing spinor follow from requiring that the
propagation system \eqref{WaveEquationH}-\eqref{WaveEquationS} of
Proposition \ref{Proposition:KillingSpinorPropagationSystem} has as a
unique solution ---the trivial (zero) one.

The purpose of this section is to analyse the characteristic initial
value problem for the Killing spinor equation propagation system
\eqref{WaveEquationH}-\eqref{WaveEquationS}.

\subsubsection{Basic observations}
We are interested in solutions to the system
\eqref{WaveEquationH}-\eqref{WaveEquationS} ensuring the existence of
a Killing spinor on $D(\mathcal{H}_1\cup\mathcal{H}_2)$. The
homogeneity of these equations on the fields $H_{AA'BC}$ and
$S_{AA'BB'}$ allows to formulate the following result:

\begin{lemma}
\label{Lemma:KillingSpinorEquationEvolutionSystem}
Let $(\mathcal{M},\bmg)$ denote a spacetime satisfying the assumptions
of Theorem \ref{Theorem:ExistenceCIVPSpacetime}. Further, assume that 
\[
H_{AA'BC}=0, \qquad S_{AA'BB'}=0 \qquad \mbox{on} \quad \mathcal{H}_1\cup\mathcal{H}_2.
\]
Then there exists a neighbourhood $\mathcal{O}$ of $\mathcal{Z}$ in
 $H_{AA'BC}$ and
$S_{AA'BB'}$ vanish everywhere on the domain of dependence of $\mathcal{H}_1\cup
	\mathcal{H}_2$.
\end{lemma}

\begin{proof}
The result follows from using the methods of Section
\ref{Section:KillingSpinorCandidate} on the equations
\eqref{WaveEquationH}-\eqref{WaveEquationS}, and the uniqueness of the
solutions to the characteristic initial value problem. 
\end{proof}

From the above lemma and the observations in Section \ref{Section:KillingSpinors} one
directly obtains the following result concerning the existence of
Killing spinors on $D(\mathcal{H}_1\cup\mathcal{H}_2)$:

\begin{proposition}
\label{KillingSpinorExistence}
Let $(\mathcal{M},\bmg)$ denote a spacetime satisfying the assumptions
of Theorem \ref{Theorem:ExistenceCIVPSpacetime}. Assume that initial
data $\kappa_0$, $\kappa_1$, $\kappa_2$ on
$\mathcal{H}_1\cup \mathcal{H}_2$ for the wave equation
\eqref{KillingSpinorWaveEquation} can be found such that 
\[
H_{AA'BC}=0, \qquad S_{AA'BB'}=0 \qquad \mbox{on} \quad \mathcal{H}_1\cup\mathcal{H}_2.
\]
Then the resulting Killing spinor candidate $\kappa_{AB}$ is, in fact,
a Killing spinor everywhere on the domain of dependence of
$\mathcal{H}_1\cup\mathcal{H}_2$. 
\end{proposition}

\begin{remark}
{\em A standard computation shows that the condition
\[
H_{AA'BC}=0
\]
 is
equivalent to the equations
\begin{subequations}
\begin{eqnarray}
&& D \kappa_0 - 2\epsilon \kappa_0 + 2\kappa\kappa_1 =0, \label{NPH1}\\
&& \delta \kappa_0- 2\beta \kappa_0+2 \sigma\kappa_1=0, \label{NPH2}\\
&& \overline{\delta}\kappa_0+ 2D\kappa_1 -2 \pi \kappa_0-2 \alpha \kappa_0
+2 \kappa_1 \rho
+2 \kappa \kappa_2=0, \label{NPH3}\\
&& \Delta \kappa_0 +2 \delta \kappa_1+ 2 \sigma \kappa_2 -2 \mu
\kappa_0 +2\tau \kappa_1-2 \gamma
\kappa_0 =0, \label{NPH4}\\
&& D\kappa_2 + 2\overline{\delta}\kappa_1+ 2\rho \kappa_2-2\lambda
\kappa_0-2\pi\kappa_1+2 \epsilon \kappa_2=0, \label{NPH5}\\
&& \delta \kappa_2 +2 \Delta \kappa_1 +2 \tau \kappa_2 + 2\beta \kappa_2 -2\mu\kappa_1-2 \nu \kappa_0
=0, \label{NPH6}\\
&& \overline{\delta} \kappa_2 + 2\alpha \kappa_2 -2 \lambda\kappa_1=0, \label{NPH7}\\
&& \Delta \kappa_2 + 2 \gamma \kappa_2-2\nu\kappa_1=0.  \label{NPH8}
\end{eqnarray} 
\end{subequations}}
\end{remark}

\begin{remark}
{\em Using the notation
\[
\xi_{AA'} = \xi_{11'} o_A \overline{o}_{A'} + \xi_{10'} o_A \overline{\iota}_{A'} +
\xi_{01'} \iota_{A} \overline{o}_{A'} + \xi_{00'} \iota_A \overline{\iota}_{A'},
\]
equation \eqref{KillingSpinorToKillingVector} takes the form
\begin{subequations}
\begin{eqnarray}
&& \xi_{11'} = \Delta \kappa_1 - \delta \kappa_2 -2 \beta\kappa_2 +\tau \kappa_2+2\mu\kappa_1
-\nu \kappa_0, \label{XiExpanded1}\\
&& \xi_{10'} = D\kappa_2-\overline{\delta}\kappa_1 + 2 \epsilon \kappa_2 -
\rho \kappa_2 -2\pi\kappa_1+ \lambda \kappa_0 , \\ 
&& \xi_{01'} = \delta \kappa_1-\Delta\kappa_0  +2 \gamma \kappa_0 -
\mu \kappa_0 -2\tau\kappa_1+ \sigma \kappa_2, \\
&& \xi_{00'} = \overline{\delta}\kappa_0 -D\kappa_1 -2\alpha \kappa_0 + \pi
\kappa_0+2 \rho\kappa_1 -
\kappa \kappa_2.  \label{XiExpanded4}
\end{eqnarray}
\end{subequations}

If $\xi_{AA'}$ is required to be Hermitian so that it corresponds to
the spinorial counterpart of a real vector $\xi^a$ then one has the
reality conditions
\[
\xi_{00'} = \overline{\xi}_{0'0}, \qquad \xi_{11'} = \overline{\xi}_{1'1},
\qquad \xi_{01'} = \overline{\xi}_{1'0}, \qquad \xi_{10'} = \overline{\xi}_{0'1}
\]

A further calculation shows that the equation $S_{AA'BB'}=0$ takes, in
NP notation the form
\begin{subequations}
\begin{eqnarray}
&& D\xi_{00'}- \xi_{00'} \epsilon -  \xi_{00'} \overline{\epsilon} -  \xi_{10'} 
\kappa -  \xi_{01'} \overline{\kappa} = 0, \label{KVEqnNP1}\\
&&  \Delta\xi_{11'}+ \xi_{11'} \gamma + \xi_{11'} \overline{\gamma} +  \xi_{01'} \nu 
+  \xi_{10'} \overline{\nu} = 0, \label{KVEqnNP2}\\
&& 
   D\xi_{11'} + \Delta\xi_{00'}- \xi_{00'} \gamma -  \xi_{00'} \overline{\gamma} + \xi_{11'} 
\epsilon + \xi_{11'} \overline{\epsilon} + \xi_{01'} \pi + \xi_{10'} 
\overline{\pi} -  \xi_{10'} \tau -  \xi_{01'} \overline{\tau} = 0, \label{KVEqnNP3}\\
&&  \delta\xi_{11'}  -
   \Delta\xi_{01'}+ \overline{\alpha}\xi_{11'} + \xi_{11'} \beta + \xi_{01'} \gamma -  
\xi_{01'} \overline{\gamma} + \xi_{10'} \overline{\lambda} + 
\xi_{01'} \mu -  \xi_{00'} \overline{\nu} + \xi_{11'} \tau= 0, \label{KVEqnNP4}\\
&&\delta\xi_{01'} + \xi_{01'} \overline{\alpha} - \xi_{01'} \beta + \xi_{00'} 
\overline{\lambda} - \xi_{11'} \sigma =0,\label{KVEqnNP5}\\
&&  \delta\xi_{00'}-  D\xi_{01'} - \xi_{00'} \overline{\alpha} -  \xi_{00'} \beta + \xi_{01'} \epsilon 
-  \xi_{01'} \overline{\epsilon} + \xi_{11'} \kappa -  \xi_{00'} 
\overline{\pi} -  \xi_{01'} \overline{\rho} -  \xi_{10'} \sigma = 0, \label{KVEqnNP6}\\
&&  \overline{\delta}\xi_{11'}-
   \Delta\xi_{10'} +\xi_{11'} \alpha + \xi_{11'} \overline{\beta} -  \xi_{10'} \gamma + 
\xi_{10'} \overline{\gamma} + \xi_{01'} \lambda + \xi_{10'} 
\overline{\mu} -  \xi_{00'} \nu + \xi_{11'} \overline{\tau} = 0, \label{KVEqnNP7}\\
&& \overline{\delta}\xi_{10'} + \xi_{10'} \alpha - \xi_{10'} \overline{\beta} + \xi_{00'} 
\lambda - \xi_{11'} \overline{\sigma} = 0, \label{KVEqnNP8}
  \\
&&
   \delta\xi_{10'} +  \overline{\delta}\xi_{01'}- \xi_{01'} \alpha - \xi_{10'} \overline{\alpha} +  \xi_{10'} \beta +  
\xi_{01'} \overline{\beta} + \xi_{00'} \mu +  \xi_{00'} 
\overline{\mu} - \xi_{11'} \rho - \xi_{11'} \overline{\rho} = 0, \label{KVEqnNP9}\\
&&  \overline{\delta}\xi_{00'}-  D\xi_{10'} - \xi_{00'} \alpha -  \xi_{00'} \overline{\beta} -  \xi_{10'} 
\epsilon + \xi_{10'} \overline{\epsilon} + \xi_{11'} 
\overline{\kappa} -  \xi_{00'} \pi -  \xi_{10'} \rho -  \xi_{01'} 
\overline{\sigma} = 0. \label{KVEqnNP10}
\end{eqnarray}
\end{subequations}}
\end{remark}

In the remainder of this section we investigate these conditions on
$\mathcal{H}_1\cup\mathcal{H}_2$. 

\subsubsection{The condition $H_{AA'BC}=0$ on $\mathcal{Z}=\mathcal{H}_1\cap\mathcal{H}_2$} 

On  $\mathcal{Z}=\mathcal{H}_1\cap\mathcal{H}_2$ equations
\eqref{NPH1}-\eqref{NPH8} reduce to:
\begin{subequations}
\begin{eqnarray}
&& D\kappa_0 =0, \label{NPHZ1}\\
&& \Delta \kappa_2 =0, \label{NPHZ2}\\
&& \delta \kappa_0 -2\beta \kappa_0=0, \label{NPHZ3}\\
&& \Delta \kappa_0 + 2 \delta \kappa_1 + 2\tau \kappa_1 =0, \label{NPHZ4}\\
&& 2\Delta \kappa_1 +\delta \kappa_2 +2\beta \kappa_2 +2 \tau \kappa_2
=0, \label{NPHZ5}\\
&& 2 D\kappa_1+ \overline{\delta}\kappa_0  -2\alpha \kappa_0=0, \label{NPHZ6}\\
&& D\kappa_2 + 2 \overline{\delta}\kappa_1 =0, \label{NPHZ7}\\
&& \overline{\delta}\kappa_2 + 2 \alpha \kappa_2 =0. \label{NPHZ8}
\end{eqnarray}
\end{subequations}

In what follows, we regard equations \eqref{NPHZ3} and
\eqref{NPHZ8} as intrinsic to $\mathcal{Z}$. Making use of  the
operators $\eth$ and $\overline{\eth}$ (see \eqref{eths} for their
explicit form) these conditions can be concisely rewritten as
\begin{subequations}
\begin{eqnarray}
&&\eth \kappa_0 =\tau \kappa_0, \label{Kappa0Kappa2Za}\\
&& \overline{\eth}\kappa_2=-\overline{\tau}\kappa_2. \label{Kappa0Kappa2Zb}
\end{eqnarray}
\end{subequations}

\begin{remark}
{\em Equations \eqref{NPHZ1}-\eqref{NPHZ8} do not constrain the value of
the coefficient $\kappa_1$ on $\mathcal{Z}$. Instead, given an arbitrary (smooth)
choice of $\kappa_1$ and coefficients $\kappa_0$ and $\kappa_2$
satisfying the equations in \eqref{Kappa0Kappa2Za}-\eqref{Kappa0Kappa2Zb}, we regard equations
\eqref{NPHZ2}, \eqref{NPHZ4} and \eqref{NPHZ5} as prescribing the
initial values of the derivatives $\Delta \kappa_0$, $\Delta\kappa_1$
and $\Delta \kappa_2$ that need to be provided for the transport
equations
\eqref{KillingSpinorWaveEquationNP1H1}-\eqref{KillingSpinorWaveEquationNP3H1}
along $\mathcal{H}_1$. Similarly, we use equations \eqref{NPHZ1},
\eqref{NPHZ6} and \eqref{NPHZ7} to prescribe the initial values of the
derivatives $D \kappa_0$, $D\kappa_1$
and $D \kappa_2$ which are used, in turn, to solve the transport
equations
\eqref{KillingSpinorWaveEquationNP1H2}-\eqref{KillingSpinorWaveEquationNP3H2}
along $\mathcal{H}_2$.} 
\end{remark}

\subsubsection{The condition $H_{AA'BC}=0$ on $\mathcal{H}_1$} 

On  $\mathcal{H}_1$ equations
\eqref{NPH1}-\eqref{NPH8}  reduce to:
\begin{subequations}
\begin{eqnarray}
&& D\kappa_0 =0, \label{NPHH11}\\
&& \Delta \kappa_2 -2 \nu\kappa_1 + 2\gamma\kappa_2 =0, \label{NPHH12}\\
&& \delta \kappa_0 -2\beta \kappa_0=0, \label{NPHH13}\\
&& \Delta \kappa_0 + 2\delta \kappa_1 -2(\gamma +\mu)\kappa_0 + 2\tau
\kappa_1 =0, \label{NPHH14}\\
&& 2 \Delta \kappa_1 +\delta \kappa_2 + 2(\beta +\tau)\kappa_2 -2
\mu\kappa_1 -2 \nu \kappa_0 =0, \label{NPHH15}\\
&& 2 D\kappa_1 +\overline{\delta}\kappa_0 -2\alpha \kappa_0=0, \label{NPHH16}\\
&& D\kappa_2 + 2\overline{\delta}\kappa_1 =0, \label{NPHH17}\\
&& \overline{\delta}\kappa_2 + 2\alpha\kappa_2=0. \label{NPHH18}
\end{eqnarray}
\end{subequations}
Equations \eqref{NPHH11}, \eqref{NPHH16} and \eqref{NPHH17} are
interpreted as propagation equations along the null generators of
$\mathcal{H}_1$ which are used to propagate the initial values of $\kappa_0$,
$\kappa_1$ and $\kappa_2$  at $\mathcal{Z}$. In order to understand
the role equations \eqref{NPHH13} and \eqref{NPHH18} we consider the
expressions
\[
D(\delta \kappa_0 -2\beta \kappa_0), \qquad D(\overline{\delta}\kappa_2 + 2\alpha\kappa_2).
\]
A direct computation using the NP commutators shows that
\[
D(\delta \kappa_0 -2\beta \kappa_0) =-2 \kappa_0 D\beta \qquad
D(\overline{\delta}\kappa_2 + 2\alpha\kappa_2)=2\kappa_2 D\alpha -2
(\alpha-\overline{\beta})\overline{\delta}\kappa_1 - 2 \overline{\delta}^2\kappa_1.
\]
Evaluating the Ricci identities on $\mathcal{H}_1$ one finds
that $D\alpha=D\beta=0$ ---see also Table
\ref{Table:CharacteristicInitialData}. Thus, it follows that
\[
D(\delta \kappa_0 -2\beta \kappa_0) =0, \qquad D(\overline{\delta}\kappa_2
+ 2\alpha\kappa_2)=-2 \overline{\eth}^2 \kappa_1.
\]
Accordingly, equation \eqref{NPHH13} holds along $\mathcal{H}_1$ if it
is satisfied on $\mathcal{Z}$ ---this is equivalent to requiring
condition \eqref{Kappa0Kappa2Za} on $\mathcal{Z}$. Observe, however, that in order to
obtain the same conclusion for equation \eqref{NPHH18} one needs
$\overline{\eth}^2 \kappa_1=0$ on $\mathcal{H}_1$. 

\medskip
It remains to consider equations \eqref{NPHH12}, \eqref{NPHH14} and
\eqref{NPHH15}. These prescribe the value of the transversal
derivatives $\Delta \kappa_0$, $\Delta \kappa_1$ and $\Delta
\kappa_2$. Recall, however, that from the discussion of Section \ref{Section:KillingSpinorCandidate}
these derivatives satisfy transport equations along the generators of
$\mathcal{H}_1$. Thus, some compatibility conditions will
arise. Substituting the value of $\Delta \kappa_0$, given by equation
\eqref{NPHH14} into the transport equation
\eqref{KillingSpinorWaveEquationNP1H1}, and then using the NP
commutators, NP Ricci identities and equations
\eqref{NPHH11}, \eqref{NPHH16} and \eqref{NPHH17} to simplify one
obtains the condition 
\[
\Psi_2 \kappa_0=0.
\]
Similarly, substituting the value of $\Delta \kappa_1$ given by
equation \eqref{NPHH15} into the transport equation
\eqref{KillingSpinorWaveEquationNP2H1} and proceeding in similar
manner one finds the further condition
\[
\Psi_3 \kappa_0=0.
\]
Finally, the substitution of the value of $\Delta \kappa_2$ as given
by equation \eqref{NPHH12} eventually leads to the condition
\[
\Psi_4 \kappa_0 + 2 \Psi_3 \kappa_1 - 3\Psi_2 \kappa_2=0.
\]

\medskip
One can summarise the discussion of this subsection as follows:

\begin{lemma}
\label{KillingSpinorH1}
Assume that equations \eqref{NPHH11}, \eqref{NPHH16} and
\eqref{NPHH17} hold along $\mathcal{H}_1$ with initial data for
$\kappa_0$ and $\kappa_2$ on $\mathcal{Z}$ satisfying equations
\eqref{Kappa0Kappa2Za} and \eqref{Kappa0Kappa2Zb}, respectively, and that, in addition, 
\[
\overline{\eth}^2 \kappa_1=0, \quad \Psi_2 \kappa_0=0, \quad \Psi_3\kappa_0=0,
\quad \Psi_4 \kappa_0 + 2\Psi_3\kappa_1 - 3\Psi_2 \kappa_2=0, \quad
\mbox{on} \quad \mathcal{H}_1.
\] 
Then, one has that 
\[
H_{A'ABC}=0 \quad \mbox{on} \quad \mathcal{H}_1.
\]
\end{lemma}

\subsubsection{The condition $H_{AA'BC}=0$ on $\mathcal{H}_2$} 

On  $\mathcal{H}_2$ equations
\eqref{NPH1}-\eqref{NPH8} reduce to:
\begin{subequations}
\begin{eqnarray}
&& D\kappa_0 =0, \label{NPHH21}\\
&& \Delta \kappa_2 =0, \label{NPHH22}\\
&& \delta \kappa_0 -2\beta \kappa_0+2\sigma \kappa_1=0, \label{NPHH23}\\
&& \Delta \kappa_0 + 2\delta \kappa_1 + 2\tau
\kappa_1 + 2\sigma\kappa_2=0, \label{NPHH24}\\
&& 2 \Delta \kappa_1 +\delta \kappa_2 + 2(\beta +\tau)\kappa_2 =0, \label{NPHH25}\\
&& 2 D\kappa_1 +\overline{\delta}\kappa_0 -2\alpha \kappa_0+ 2\rho\kappa_1=0, \label{NPHH26}\\
&& D\kappa_2 + 2\overline{\delta}\kappa_1+2 \rho\kappa_2 =0, \label{NPHH27}\\
&& \overline{\delta}\kappa_2 + 2\alpha\kappa_2=0. \label{NPHH28}
\end{eqnarray}
\end{subequations}
In analogy with the analysis on $\mathcal{H}_2$, in what follows we
regard equations \eqref{NPHH22}, \eqref{NPHH24} and \eqref{NPHH25} as
propagation equations for the components $\kappa_0$, $\kappa_1$ and
$\kappa_2$ along the generators of $\mathcal{H}_2$. Initial data for
these equations is naturally prescribed on $\mathcal{Z}$. 

\medskip
Now, regarding equation \eqref{NPHH28}, a direct computation shows that
\[
\Delta (\overline{\delta}\kappa_2 + 2\alpha\kappa_2) =0.
\]
Thus, if equation \eqref{NPHH28} is satisfied on $\mathcal{Z}$ then it holds along
the generators of $\mathcal{H}_2$ ---this equivalent to requiring \eqref{Kappa0Kappa2Zb}. A similar computation with equation
\eqref{NPHH23} yields the more complicated relation
\[
\Delta(\delta \kappa_0 -2\beta \kappa_0+2\sigma \kappa_1) =  -2 \eth^2
\kappa_1 -2 \kappa_2 \delta \sigma-3\sigma \delta
\kappa_2-2\overline{\alpha}\sigma \kappa_2.
\]
Observe that if $\kappa_2=0$ along $\mathcal{H}_2$, then the
obstruction to the propagation of equation \eqref{NPHH23} reduces to
the simple condition $\eth^2
\kappa_1=0$ which is somehow complementary to the condition
$\overline{\eth}^2\kappa_1=0$ on $\mathcal{H}_1$. 

\medskip
It remains to analyse the compatibility of equations \eqref{NPHH21},
\eqref{NPHH26} and \eqref{NPHH27} with the transport equations
\eqref{KillingSpinorWaveEquationNP1H2}-\eqref{KillingSpinorWaveEquationNP3H2}.
Substituting $D\kappa_{1}$, $\Delta\kappa_{0}$,  $D\kappa_{0}$
and $\delta\kappa_{0}$ given by equations \eqref{NPHH21},
\eqref{NPHH24}, \eqref{NPHH26} and \eqref{NPHH23} into equation
\eqref{KillingSpinorWaveEquationNP1H2} one obtains after some
manipulations the condition
\[
\Psi_{0}\kappa_{2} + 2\Psi_{1}\kappa_{1} - 3\Psi_{2}\kappa_{0} = 0.
\]
Similarly, after substituting $D\kappa_{1}$, $\Delta\kappa_{1}$ and
$D\kappa_{2}$ given by equations \eqref{NPHH26}, \eqref{NPHH25} and
\eqref{NPHH27} into equation \eqref{KillingSpinorWaveEquationNP2H2}
one obtains the condition
\[
\Psi_1 \kappa_2 =0.
\]

Finally, by substituting  $D\kappa_{2},
\Delta\kappa_{2}$ and $\overline{\delta}\kappa_{2}$ given
by \eqref{NPHH27}, \eqref{NPHH22} and
\eqref{NPHH28} into equation \eqref{KillingSpinorWaveEquationNP3H2},
one obtains the condition
\[
\Psi_{2}\kappa_{2} = 0.
\]

One can summarise the discussion of this subsection as follows:

\begin{lemma}
\label{KillingSpinorH2}
Assume that equations \eqref{NPHH22}, \eqref{NPHH24} and
\eqref{NPHH25} hold along $\mathcal{H}_2$ with initial data for
$\kappa_0$ and $\kappa_2$ on $\mathcal{Z}$ satisfying conditions
\eqref{Kappa0Kappa2Za} and \eqref{Kappa0Kappa2Zb}, respectively, and that, in addition, 
\[
\eth^2 \kappa_1 + \kappa_{2}\delta\sigma + \frac{3}{2}\sigma\delta\kappa_{2} + \overline{\alpha}\sigma\kappa_{2}=0, \quad \Psi_2 \kappa_2=0, \quad \Psi_1\kappa_2=0,
\quad \Psi_0 \kappa_2 + 2\Psi_1\kappa_1 - 3\Psi_2 \kappa_0=0, \quad
\mbox{on} \quad \mathcal{H}_2.
\] 
Then, one has that 
\[
H_{A'ABC}=0 \quad \mbox{on} \quad \mathcal{H}_2.
\]
\end{lemma}

\begin{remark}\label{IntegrabilityCondition}
{\em One can show that the curvature conditions in Lemmas \ref{KillingSpinorH1} and \ref{KillingSpinorH2} are in fact components of the equation
\begin{equation*}
\Psi_{(ABC}{}^{F}\kappa_{D)F}=0.
\end{equation*}
The other components of this equation are trivially satisfied. As this
is a basis independent expression, the curvature conditions are
satisfied in all spin bases, not just the parallelly propagated
one. One can check this by considering Lorentz transformations and
null rotations about $l^a$ and $n^a$, and show that these conditions are
preserved. The equation above can be shown to be an integrability
condition for the Killing spinor equation, so it is unsurprising to
find components of it arising naturally from the analysis.}
\end{remark}

\subsubsection{The condition $S_{AA'BB'}=0$ at $\mathcal{Z}$}\label{KVFonZ}

Using the properties of $\mathcal{Z}$, as given explicitly in Table
\ref{Table:CharacteristicInitialData},  together with the conditions
\eqref{NPHZ1}-\eqref{NPHZ8} 
implied by the equation $H_{AA'BC}=0$ on $\mathcal{Z}$, equations \eqref{XiExpanded1}-\eqref{XiExpanded4} reads as
\begin{subequations}
\begin{eqnarray}
&& \xi_{11'} = -\frac{3}{2}(\eth\kappa_2 +\tau\kappa_2)\,, \label{KVatZ1}
\\
&& \xi_{10'} = -3 \overline{\eth} \kappa_1\,, \label{KVatZ2}\\
&& \xi_{01'} = 3 \eth\kappa_1\,, \label{KVatZ3}\\
&& \xi_{00'} = \frac{3}{2}(\eth \kappa_0 - \overline{\tau} \kappa_0)\,, \label{KVatZ4}
\end{eqnarray}
\end{subequations}
while on $\cZ$ equations \eqref{KVEqnNP1}-\eqref{KVEqnNP10} reduce to
\begin{subequations}
\begin{eqnarray}
&& D\xi_{00'} =0, \label{KVEqnNPZ1}\\
&& \Delta \xi_{11'} =0, \label{KVEqnNPZ2}\\
&& D\xi_{11'}+\Delta \xi_{00'}  -\tau\xi_{10'} - \overline{\tau}\xi_{01'} = 0, \label{KVEqnNPZ3}\\
&& \Delta\xi_{01'} - \delta\xi_{11'}- 2\tau \xi_{11'} = 0,\label{KVEqnNPZ4}\\
&& \delta\xi_{01'}  + (\overline{\alpha} -  \beta) \xi_{01'} = 0, \label{KVEqnNPZ5}\\
&& D\xi_{01'} - \delta \xi_{00'} + \tau \xi_{00'}  = 0, \label{KVEqnNPZ6}\\
&& \Delta\xi_{10'} - \overline{\delta}\xi_{11'} - 2 \overline{\tau} \xi_{11'} =
   0, \label{KVEqnNPZ7}\\
&& \overline{\delta}\xi_{10'}+ (\alpha -\overline{\beta}) \xi_{10'} = 0, \label{KVEqnNPZ8}\\
&&  \delta\xi_{10'} +  \overline{\delta}\xi_{01'} -( \overline{\alpha} -  \beta)
   \xi_{10'} -( \alpha- \overline{\beta}) \xi_{01'} = 0, \label{KVEqnNPZ9}\\
&&D\xi_{10'} - \overline{\delta}\xi_{00'} + \overline{\tau}\xi_{00'} = 0. \label{KVEqnNPZ10}
\end{eqnarray}
\end{subequations}
Equations \eqref{KVEqnNPZ5}, \eqref{KVEqnNPZ8} and \eqref{KVEqnNPZ9}
can be read as intrinsic equations for $\xi_{01'} $ and
$\xi_{10'}$. Expressing these in terms of the $\eth$ and $\overline{\eth}$
operators, observing that the spin-weight of $\xi_{01'}$ and
$\xi_{10'}$ are respectively $-1$ and $1$, one has that
\begin{subequations}
\begin{eqnarray}
&& \eth \xi_{01'} =0, \label{ConstraintAtZ1}\\
&& \overline{\eth} \xi_{10'} =0, \label{ConstraintAtZ2}\\
&& \eth \xi_{10'} + \overline{\eth}\xi_{01'} =0. \label{ConstraintAtZ3}
\end{eqnarray} 
\end{subequations}
Substituting conditions \eqref{KVatZ2}-\eqref{KVatZ3} into
conditions \eqref{ConstraintAtZ1}-\eqref{ConstraintAtZ2} above yield the simple conditions

\begin{equation}\label{preKVFonZ}
\eth^2\kappa_1 =0, \qquad \overline{\eth}^2 \kappa_1=0\,,
\end{equation}
whereas \eqref{ConstraintAtZ3}---as $\kappa_1$ is of zero spin-weight
quantity---reduces to the commutation relation
\[
\eth\overline{\eth}\kappa_1 - \overline{\eth}\eth\kappa_1=0\,.
\]

\begin{remark}
 {\em The above expressions indicate that the component
$\kappa_1$ has a very specific multipolar structure. Note, however,
that the $\eth$ and $\overline{\eth}$ above are not the ones corresponding
to $\mathbb{S}^2$ but of a 2-manifold diffeomorphic to it. Thus, in
order to further the discussion one needs to consider the conformal
properties of the operators. }
\end{remark}

Crucially, one can also show that equations
\eqref{KVEqnNPZ1}-\eqref{KVEqnNPZ4},
\eqref{KVEqnNPZ6}-\eqref{KVEqnNPZ7} and \eqref{KVEqnNPZ10} are implied
by equations \eqref{NPHZ1}-\eqref{NPHZ8}, the Ricci equations, and the
conditions of Lemmas 2 and 3 (which must be satisfied on
$\mathcal{Z}=\mathcal{H}_{1}\cap\mathcal{H}_{2}$. Summarising:

\begin{lemma}
\label{KillingVectorZ}
Assume that equations \eqref{NPHZ1}-\eqref{NPHZ8} hold on
$\mathcal{Z}$ and that, in addition,
\[
\eth^2\kappa_1 =0, \qquad \overline{\eth}^2 \kappa_1=0, \quad
\mbox{on} \quad \mathcal{Z}.
\]
Then one has that 
\[
S_{AA'BB'}=0 \quad \mbox{on} \quad \mathcal{Z}.
\]
\end{lemma}

\subsubsection{The Killing vector equation on $\mathcal{H}_{1}$}

On $\mathcal{H}_{1}$, equations \eqref{KVEqnNP1}-\eqref{KVEqnNP10} reduce to:
\begin{subequations}
\begin{eqnarray}
&& D\xi_{00'} =0, \label{KVEqnNPH11}\\
&& \Delta \xi_{11'} + (\gamma + \overline{\gamma})\xi_{11'} + \nu\xi_{01'} + \overline{\nu}\xi_{10'}=0, \label{KVEqnNPH12}\\
&& D\xi_{11'}+\Delta \xi_{00'}  -\tau\xi_{10'} - \overline{\tau}\xi_{01'} - (\gamma + \overline{\gamma})\xi_{00'} = 0, \label{KVEqnNPH13}\\
&& \Delta\xi_{01'} - \delta\xi_{11'} - (\gamma - \overline{\gamma} + \mu)\xi_{01'} + \overline{\nu}\xi_{00'} - 2\tau \xi_{11'} = 0,\label{KVEqnNPH14}\\
&& \delta\xi_{01'}  + (\overline{\alpha} -  \beta) \xi_{01'} = 0, \label{KVEqnNPH15}\\
&& D\xi_{01'} - \delta \xi_{00'} + \tau \xi_{00'}  = 0, \label{KVEqnNPH16}\\
&& \Delta\xi_{10'} - \overline{\delta}\xi_{11'} - (\overline{\gamma}-\gamma + \overline{\mu})\xi_{10'} + \nu\xi_{00'} - 2 \overline{\tau} \xi_{11'} =
   0, \label{KVEqnNPH17}\\
&& \overline{\delta}\xi_{10'}+ (\alpha -\overline{\beta}) \xi_{10'} = 0, \label{KVEqnNPH18}\\
&&  \delta\xi_{10'} +  \overline{\delta}\xi_{01'} + (\mu + \overline{\mu})\xi_{00'} -( \overline{\alpha} -  \beta)
   \xi_{10'} -( \alpha- \overline{\beta}) \xi_{01'} = 0, \label{KVEqnNPH19}\\
&&D\xi_{10'} - \overline{\delta}\xi_{00'} + \overline{\tau}\xi_{00'} = 0. \label{KVEqnNPH110}
\end{eqnarray}
\end{subequations}

Substituting the components $\xi_{00'}, \xi_{01'}, \xi_{10'}$ and
$\xi_{11'}$, as given by \eqref{XiExpanded1}-\eqref{XiExpanded4}, into
these relations (being careful not to discard the $\Delta$ derivatives
of quantities which vanish on $\mathcal{H}_{1}$), and using equations
\eqref{NPHH11}-\eqref{NPHH18} and the Ricci equations, one finds that
\eqref{KVEqnNPH11}-\eqref{KVEqnNPH110} reduce to:
\begin{subequations}
\begin{eqnarray}
&& \eth^2\kappa_{1} = \kappa_{0}(\delta\mu+\mu\tau), \label{KVEqnNPH1Red1} \\
&&\overline{\eth}^2 \kappa_{1}=0, \label{KVEqnNPH1Red2}\\
&&\Psi_{2}\kappa_{0}=0, \label{KVEqnNPH1Red3}\\
&&\Psi_{3}\kappa_{0}=0, \label{KVEqnNPH1Red4}\\
&&\Psi_{4}\kappa_{0} + 2\Psi_{3}\kappa_{1} - 3\Psi_{2}\kappa_{2} = 0.\label{KVEqnNPH1Red5}
\end{eqnarray}
\end{subequations}

\begin{remark}
{\em The conditions \eqref{KVEqnNPH1Red2}-\eqref{KVEqnNPH1Red5} are exactly
the conditions of Lemma 2. The additional condition
\eqref{KVEqnNPH1Red1} must be satisfied on all of
$\mathcal{H}_{1}$. Note, however, that after some manipulations the condition
\[
D\left(\eth^2\kappa_{1} - \kappa_{0}(\delta\mu+\mu\tau)\right) = -2\delta(\Psi_{2}\kappa_{0})+4\beta\Psi_{2}\kappa_{0}=0
\]
can be shown to hold, where in the last step \eqref{KVEqnNPH1Red3} was
used. Accordingly, it suffices to guarantee \eqref{KVEqnNPH1Red1} on
$\mathcal{Z}$ as then it is satisfied on the whole of
$\mathcal{H}_{1}$ if condition \eqref{KVEqnNPH1Red3} holds on
$\mathcal{H}_1$. Furthermore, on $\mathcal{Z}$ the spin coefficient
$\mu$ vanishes, so \eqref{KVEqnNPH1Red1} reduces to $\eth^2\kappa_1
=0$ on $\mathcal{Z}$. Note that this is one of the
conditions appearing in Lemma \ref{KillingVectorZ}.}
\end{remark}

Summarising, we have the following lemma:
\begin{lemma}
\label{KillingVectorH1}
Assume that equations \eqref{NPHH11}-\eqref{NPHH18} hold on
$\mathcal{H}_1$, and the conditions of Lemmas \ref{KillingSpinorH1} and
\ref{KillingVectorZ} are satisfied. Then one has that
\[
S_{AA'BB'}=0 \quad \mbox{on} \quad \Hone.
\]
\end{lemma}

\subsubsection{The Killing vector equation on $\mathcal{H}_{2}$}

On $\mathcal{H}_{2}$, equations \eqref{KVEqnNP1}-\eqref{KVEqnNP10} reduce to:
\begin{subequations}
\begin{eqnarray}
&& D\xi_{00'} =0, \label{KVEqnNPH21}\\
&& \Delta \xi_{11'} = 0, \label{KVEqnNPH22}\\
&& D\xi_{11'}+\Delta \xi_{00'}  -\tau\xi_{10'} - \overline{\tau}\xi_{01'} = 0, \label{KVEqnNPH23}\\
&& \Delta\xi_{01'} - \delta\xi_{11'} - 2\tau \xi_{11'} = 0,\label{KVEqnNPH24}\\
&& \delta\xi_{01'}  + (\overline{\alpha} -  \beta) \xi_{01'} - \sigma\xi_{11'} = 0, \label{KVEqnNPH25}\\
&& D\xi_{01'} - \delta \xi_{00'} + \tau \xi_{00'} + \sigma\xi_{10'} + \rho\xi_{01'} = 0, \label{KVEqnNPH26}\\
&& \Delta\xi_{10'} - \overline{\delta}\xi_{11'} - 2 \overline{\tau} \xi_{11'} =
   0, \label{KVEqnNPH27}\\
&& \overline{\delta}\xi_{10'}+ (\alpha -\overline{\beta}) \xi_{10'} - \overline{\sigma}\xi_{11'} = 0, \label{KVEqnNPH28}\\
&&  \delta\xi_{10'} +  \overline{\delta}\xi_{01'} -( \overline{\alpha} -  \beta)
   \xi_{10'} -( \alpha- \overline{\beta}) \xi_{01'} -2\rho\xi_{11'} = 0, \label{KVEqnNPH29}\\
&&D\xi_{10'} - \overline{\delta}\xi_{00'} + \overline{\tau}\xi_{00'} + \overline{\sigma}\xi_{01'} + \rho\xi_{10'} = 0. \label{KVEqnNPH210}
\end{eqnarray}
\end{subequations}
Substituting the components $\xi_{00'}, \xi_{01'}, \xi_{10'}$ and
$\xi_{11'}$, as given by \eqref{XiExpanded1}-\eqref{XiExpanded4}, into
these relations (being careful not to discard the $D$ derivatives of
quantities which vanish on $\mathcal{H}_{2}$), and using equations
\eqref{NPHH21}-\eqref{NPHH28} and the Ricci equations, one finds that
\eqref{KVEqnNPH21}-\eqref{KVEqnNPH210} reduce to:
\begin{subequations}
\begin{eqnarray}
&&\eth^2 \kappa_1 + \kappa_{2}\delta\sigma + \frac{3}{2}\sigma\delta\kappa_{2} + \overline{\alpha}\sigma\kappa_{2}=0,  \label{KVEqnNPH2Red1} \\
&&\overline{\eth}^2 \kappa_{1} + \kappa_{2}\overline{\delta}\sigma - \frac{1}{2}\overline{\sigma}\delta\kappa_{2} - 3\overline{\alpha}\kappa_{2}\overline{\sigma}-\overline{\Psi}_{1}\kappa_{2} = 0, \label{KVEqnNPH2Red2}\\
&&\Psi_{1}\kappa_{2}=0, \label{KVEqnNPH2Red3}\\
&&\Psi_{2}\kappa_{2}=0, \label{KVEqnNPH2Red4}\\
&&\Psi_{0}\kappa_{2} + 2\Psi_{1}\kappa_{1} - 3\Psi_{2}\kappa_{0} = 0.\label{KVEqnNPH2Red5}
\end{eqnarray}
\end{subequations}

The conditions \eqref{KVEqnNPH2Red1} and
\eqref{KVEqnNPH2Red3}-\eqref{KVEqnNPH2Red5} are exactly the conditions
of Lemma 3. The additional condition \eqref{KVEqnNPH2Red2} must be
satisfied on all of $\mathcal{H}_{2}$.  Summarising, we have the
following lemma:

\begin{lemma}
\label{KillingVectorH2}
Assume that equations \eqref{NPHH21}-\eqref{NPHH28} hold on $\Htwo$, the conditions of Lemma \ref{KillingSpinorH2} are satisfied, and that in addition, 
\begin{equation*}
\overline{\eth}^2 \kappa_{1} + \kappa_{2}\overline{\delta}\sigma - \frac{1}{2}\overline{\sigma}\delta\kappa_{2} - 3\overline{\alpha}\kappa_{2}\overline{\sigma} -\overline{\Psi}_{1}\kappa_{2} =0\quad\mbox{on}\quad\Htwo.
\end{equation*}
Then one has that 
\[
S_{AA'BB'}=0 \quad \mbox{on} \quad \Htwo.
\]
\end{lemma}

\section{Analysis the constraints on $\mathcal{Z}$}
\label{Section:ConstraintsZ}

In this section we analyse the constraints on $\mathcal{Z}$ obtained
in the previous section.

\subsection{Determining  $\kappa_{2}$ on $\cZ$}
\label{subsection:kappa2}

Consider now the restrictions we have concerning $\kappa_2$ on
$\mathcal{Z}$. To satisfy the condition $\Psi_{2}\kappa_{2}=0$ on
$\mathcal{H}_{2}$, applied in Lemma 3, we have that
$\Psi_{2}\kappa_{2}=0$ has to vanish on
$\mathcal{Z}\subset\mathcal{H}_{2}$, as well. Consistent with this
condition the following subcases can be seen to arise:

\medskip
\noindent
\textbf{\em i. Assume first that $\kappa_{2}$ is nowhere vanishing on
$\mathcal{Z}$.}  In this case $\Psi_{2}$ must vanish throughout
$\mathcal{Z}$. Note also that in virtue of Table
\ref{Table:CharacteristicInitialData} all the other Weyl spinor
components vanish on $\mathcal{Z}$, and thereby
\[
\Psi_{ABCD}|_{\mathcal{Z}}=0\,.
\]	
As shown in Table \ref{Table:CharacteristicInitialData}, $\Psi_{0}$
and $\Psi_{1}$ vanish on $\mathcal{H}_{1}$, and $\Psi_{3}$ and
$\Psi_{4}$ vanish on $\mathcal{H}_{2}$, respectively. Further, observe
that the Bianchi identities imply the following relations on
$\mathcal{H}_{1}$:
\begin{align*}
D\Psi_{2} &= 0, \\
D\Psi_{3} &= \overline{\delta}\Psi_{2}, \\
D\Psi_{4} &= 2\alpha\Psi_{3} + \overline{\delta}\Psi_{3}.
\end{align*}
As $\Psi_{2}$ vanishes on $\mathcal{Z}$ and $D$ is the directional
derivative along the geodesics generating $\mathcal{H}_{1}$, the first
of these equations imply that $\Psi_{2}=0$ on $\mathcal{H}_{1}$. By the same
argument, because the right hand side of the second of the above
relations has shown to vanish on $\mathcal{H}_{1}$, we have that
$\Psi_{3}=0$ on $\mathcal{H}_{1}$. In turn, this also implies that
$\Psi_{4}=0$ on $\mathcal{H}_{1}$ as a consequence of the last
relation. Therefore, along with the vanishing of $\Psi_{0}$ and
$\Psi_{1}$ on $\Hone$ all the Weyl spinor components vanish there
---that is one has 
\[
\Psi_{ABCD}|_{\mathcal{H}_{1}}=0\,.
\]
Similarly, the Bianchi identities imply the following relations on $\mathcal{H}_{2}$:
\begin{align*}
\Delta\Psi_{0} &= \delta\Psi_{1} - (4\tau + 2\beta)\Psi_{1} + 3\sigma\Psi_{2}, \\
\Delta\Psi_{1} &= \delta\Psi_{2} - 3\tau\Psi_{2}, \\
\Delta\Psi_{2} &= 0.
\end{align*}
As $\Psi_{2}$ vanishes on $\mathcal{Z}$, and $\Delta$ is the
directional derivative along the geodesics generating
$\mathcal{H}_{2}$, the third of these equations imply that $\Psi_{2}=0$ on
$\Htwo$. Thus, the right hand side of the second of the above
relations vanishes on $\Htwo$, and by the same argument we must have
$\Psi_{1}=0$ on $\Htwo$. The first relation then implies that
$\Psi_{0}=0$ on $\Htwo$. Therefore, along with the vanishing of
$\Psi_{3}$ and $\Psi_{4}$ on $\Htwo$ all the Weyl spinor components
vanish there. Thus, one has that 
\begin{equation*}
\Psi_{ABCD}|_{\mathcal{H}_{2}}=0\,.
\end{equation*}
Summarising, the non-vanishing of $\kappa_{2}$ on $\cZ$ implies that
all the Weyl spinor components vanish identically on the union of
$\cZ$, $\Hone$ and $\Htwo$. This, in the vacuum case,
implies that all components of the Riemann curvature tensor
vanish on $\mathcal{H}_1\cup\mathcal{H}_2$. It follows then that the
spacetime obtained in Theorem \ref{Theorem:ExistenceCIVPSpacetime}
is diffeomorphic to a portion of the Minkowski spacetime 
and the pair intersecting null hypersurfaces has to contains a
bifurcate Killing horizon corresponding to a choice of a boost Killing
vector field.

\medskip
\medskip
\noindent
\textbf{\em ii. $\kappa_{2}$ vanishes somewhere on $\cZ$:} 
It follows from the discussion in the previous subsection that, unless
the spacetime is Minkowski,  $\kappa_{2}$ must vanish somewhere on
$\cZ$. It turns out that that if this is the case, then $\kappa_{2}$ must vanish on
some open subset of $\cZ$. To see this assume, on contrary, that
$\kappa_{2}$ vanishes only at isolated points. Choose one of them, say
$z\in\cZ$ with $\kappa_{2}(z)=0$ and a Cauchy sequence $\{z_{n}\}$
converging to $z$ in the metric topology of $\cZ\approx
\mathbb{S}^2$. Since $\kappa_{2}$ is assumed to vanish only at
isolated points to ensure $\Psi_{2}\kappa_{2}=0$ on $\cZ$, the
sequence $\{\Psi_{2}(z_{n})\}$ must be the identically zero sequence
in $\mathbb{R}$ which by continuity implies that $\Psi_{2}(z)=0$. As
an analogous argument apply to any of the isolated points where
$\kappa_{2}$ vanishes we get that  $\Psi_{2}$ must be identically zero
on  $\cZ$. As we saw in the previous section this would imply that the
spacetime is Minkowski ---in conflict with our assumption that the
geometry is not flat. This, in turn, verifies that whenever $\kappa_{2}$
vanishes somewhere on $\cZ$ it has to vanish on some (non-empty) open
subset of $\cZ$.

\medskip

\noindent
\textbf{\em iii. $\kappa_{2}$ vanishes on a (non-empty)
open subset of $\cZ$.} 
It follows from \eqref{Kappa0Kappa2Zb}, and from equation 
\eqref{eth_P}, that
\[
\overline{\eth}\kappa_{2} = - \overline{\tau}\,\kappa_{2}\,,
\]
can be written as
\begin{equation}
\overline{P}P\,\partial_{\overline{z}}\,({P}^{-1}\kappa_2)= - \overline{\tau}\,P\,(P^{-1}\,\kappa_{2})
\label{vanishkappa2}
\end{equation}
implying, in turn, that $\kappa_{2}$ has to be of the form
\begin{equation}
\kappa_2= {P}\cdot\exp\left(- \int\overline{\tau}\,\overline{P}^{-1}\,d\overline{z}+\varphi(z)\,\right)\,,
\label{vanishkappa3}
\end{equation}
where $\varphi(z)$ is an arbitrary holomorphic function. This,
however, in virtue of the non-vanishing of $P$, implies that
$\kappa_2$ cannot vanish on an open subset of $\mathcal{Z}$ unless it
is identically zero on $\mathcal{Z}$, i.e.
\[
\kappa_{2}|_{\mathcal{Z}}=0\,
\]
as we intended to show. Note also that the condition \eqref{NPHH22} requires then the
vanishing of $\kappa_{2}$ along the generators of $\Htwo$, thereby we
have
\[
\kappa_{2}|_{\Htwo}=0\,.
\]

\medskip
Summarising, in this subsection we have shown the following:
\begin{lemma}
Assume that 
\[
\Psi_2 \,\kappa_2 =0 \qquad \mbox{on}\quad \mathcal{Z}.
\]
Then, if $\kappa_{2}$ is nowhere vanishing on $\cZ$, then the solution
to the characteristic initial value problem must be diffeomorphic to
the Minkowski spacetime in the domain of dependence of
$D(\mathcal{H}_1\cup\mathcal{H}_2)$. Otherwise, $\kappa_2=0$ holds on
$\cZ$, and then it is also identically zero on $\Htwo$.
\end{lemma}

\subsection{Determining  $\kappa_{0}$ on $\cZ$}

The analysis of the previous section can be adapted, \emph{mutatis mutandi}, 
to the component $\kappa_{0}$ by noting that the vanishing of
$\Psi_{2}\kappa_{0}$ on $\Hone$, one of the conditions in Lemma
\ref{KillingSpinorH1}, can be traced back to the vanishing of
$\Psi_{2}\kappa_{0}$ on $\cZ$. Indeed, it can be shown that unless the
spacetime is Minkowski $\kappa_{0}$ must vanish on a non-empty subset
of $\cZ$. The only difference in the analysis lies on the analogue of
equation \eqref{vanishkappa2}. From equations
\eqref{Kappa0Kappa2Za} and \eqref{eth_P}
along with the fact that $\kappa_{0}$ is of spin weight $-1$, it follows that
\[
{\eth}\kappa_{0} =  {\tau}\,\kappa_{0}\,,
\]
can be written as
\begin{equation}
P\overline{P}\,\partial_{{z}}\,(\overline{P}^{-1}\kappa_0)= {\tau}\,\overline{P}\,(\overline{P}^{-1}\kappa_{0})
\label{vanishkappa4}
\end{equation}
which implies, in turn, that
$\kappa_{0}$ has to be of the form
\begin{equation}
\kappa_0= \overline{P}\cdot \exp\left( \int\,{\tau}\,{P}^{-1}\,d{z}+\varsigma(\overline{z})\,\right)\,,
 \label{vanishkappa5}
\end{equation}
where $\varsigma(\overline{z})$ is an arbitrary antiholomorphic function on $\mathcal{Z}$. From
here, by an argument analogous to that used for  $\kappa_{2}$ one
concludes that 
\[
\kappa_{0}|_{\cZ}=0
\]
and, moreover, as a consequence of equation  \eqref{NPHH11}, also that
\begin{equation*}
\kappa_{0}|_{\Hone}=0.
\end{equation*}

\medskip
Summarising:
\begin{lemma}
Assume that 
\[
\Psi_2 \,\kappa_0 =0 \qquad \mbox{on}\quad \mathcal{Z}.
\]
Then, if $\kappa_{0}$ is nowhere vanishing on $\cZ$, then the solution to the characteristic initial value problem must be diffeomorphic to the Minkowski spacetime in the domain of dependence of $D(\mathcal{H}_1\cup\mathcal{H}_2)$. Otherwise, $\kappa_0=0$ holds on $\cZ$, and then it is also identically
	zero on $\Hone$.
\end{lemma}

\subsection{Eliminating redundant conditions on $\Hone$ and $\Htwo$}
The first condition in Lemma \ref{KillingSpinorH1} was
\begin{equation*}
\overline{\eth}^{2}\kappa_{1}=0\quad\mbox{on}\quad\Hone.
\end{equation*}
In theory, one would have to solve this constraint on the whole of
$\Hone$. However, one can show that on $\Hone$
\begin{align*}
D(\overline{\eth}^{2}\kappa_{1}) = &-\frac{1}{2}\overline{\delta}\overline{\delta}\overline{\delta}\kappa_{0} + \frac{3}{2}\overline{\tau}\overline{\delta}\overline{\delta}\kappa_{0} + \overline{\delta}\kappa_{0}\left(-\alpha^2 - 4\alpha\overline{\beta} - \overline{\beta}^2 + \frac{5}{2}\overline{\delta}\alpha + \frac{1}{2}\overline{\delta}\overline{\beta}\right) \\
& \,+ \kappa_{0}\left(2\alpha\overline{\beta}\overline{\tau} - 2\alpha\overline{\delta}\alpha - 3\overline{\beta}\overline{\delta}\alpha - \alpha\overline{\delta}\overline{\beta} + \overline{\delta}\overline{\delta}\alpha\right).
\end{align*}
Note that as $\kappa_{0}$ vanishes on $\Hone$ (under the assumption that the spacetime is not diffeomorphic to Minkowski), the right hand side of
this equation also vanishes on $\Hone$. Therefore, if $\kappa_{1}$
satisfies $\overline{\eth}^{2}\kappa_{1}=0$ on $\cZ$, then it also
satisfies the same condition on the whole of $\Hone$. This was a
condition on $\cZ$ already present from the requirement that
$S_{AA'BB'}|_{\cZ}=0$. Summarising:

\begin{lemma}
If $\kappa_{0}|_{\Hone}=0$ and $\overline{\eth}^2\kappa_{1}|_{\cZ}=0$, then the condition $\overline{\eth}^2\kappa_{1}|_{\Hone}=0$ from Lemma \ref{KillingSpinorH1} is automatically satisfied.
\end{lemma}

\medskip
A similar procedure can be performed on $\Htwo$. The first condition
from Lemma \ref{KillingSpinorH2} was
\begin{equation*}
\eth^{2}\kappa_{1} + \kappa_{2}\delta\sigma + \frac{3}{2}\sigma\delta\kappa_{2} + \overline{\alpha}\sigma\kappa_{2} = 0
\end{equation*}
which must be satisfied on $\Htwo$. As we already we have shown that
necessarily $\kappa_{2}|_{\Htwo}=0$  unless the spacetime is  diffeomorphic to the
Minkowski spacetime. Therefore, the aforementioned condition reduces to
\begin{equation*}
\eth^{2}\kappa_{1}=0\quad\mbox{on}\quad\Htwo.
\end{equation*}
Now, one can show that on $\Htwo$,
\begin{align*}
\Delta\left(\eth^{2}\kappa_{1}\right) = &-\frac{1}{2}\delta\delta\delta\kappa_{2} - \frac{3}{2}\tau\delta\delta\kappa_{2} + \delta\kappa_{2} \left(-\overline{\alpha}^{2}-\overline{\alpha}\beta + 2\beta^{2} - 2\delta\overline{\alpha}-4\delta\beta\right) \\
 &\,+ \kappa_{2}\left(-\overline{\alpha}\delta\overline{\alpha} + \beta\delta\overline{\alpha} - 2\overline{\alpha}\delta\beta + 2\beta\delta\beta - \delta\delta\overline{\alpha} - 2\delta\delta\beta\right).
\end{align*}
The requirement that $\kappa_{2}$ vanishes on $\Htwo$ means that the
right hand side of this equation also vanishes on $\Htwo$. Therefore,
if $\kappa_{1}$ satisfies $\eth^{2}\kappa_{1}=0$ on $\cZ$, then it
also satisfies the same condition on the whole of $\Htwo$. This was a
condition on $\cZ$ already present from the requirement that
$S_{AA'BB'}|_{\cZ}=0$.

\medskip
Finally, the condition from Lemma \ref{KillingVectorH2} says that
\begin{equation*}
\overline{\eth}^2 \kappa_{1} + \kappa_{2}\overline{\delta}\sigma - \frac{1}{2}\overline{\sigma}\delta\kappa_{2} - 3\overline{\alpha}\kappa_{2}\overline{\sigma} -\overline{\Psi}_{1}\kappa_{2} =0\quad\mbox{on}\quad\Htwo
\end{equation*}
which reduces to $\overline{\eth}^{2}\kappa_{1}=0$ due to the fact that  $\kappa_{2}|_{\Htwo}=0$ when the spacetime is not
diffeomorphic to the Minkowski solution. One can show that on $\Htwo$
\begin{align*}
\Delta\left(\overline{\eth}^{2}\kappa_{1}\right) = \;& \delta\kappa_{2}\left(\frac{1}{2}\overline{\delta}\overline{\tau} - \overline{\beta}\overline{\tau}\right) + \kappa_{2}\left( -6\alpha^{2}\beta - 6\alpha\beta\overline{\beta} - 2\alpha\delta\alpha 
 + \overline{\alpha}\overline{\delta}\alpha \right.\\
 & \left.+ 5\beta\overline{\delta}\alpha 
 + 2\alpha\overline{\delta}\overline{\alpha} + \overline{\beta}\overline{\delta}\overline{\alpha} 
 + 7\alpha\overline{\delta}\beta + 2\overline{\beta}\overline{\delta}\beta + \overline{\delta}\delta\alpha 
 - \overline{\delta}\overline{\delta}\alpha - 2\overline{\delta}\overline{\delta}\beta\right) 
\end{align*}
The requirement that $\kappa_{2}$ vanishes on $\Htwo$ means that the
right hand side of this equation also vanishes on $\Htwo$. So if
$\kappa_{1}$ satisfies $\overline{\eth}^{2}\kappa_{1}=0$ on $\cZ$, then it
also satisfies the same condition on the whole of $\Htwo$. This was a
condition already present from the requirement that
$S_{AA'BB'}|_{\cZ}=0$. Summarising, we have 

\begin{lemma}
If $\kappa_{2}|_{\Htwo}=0$ and $\overline{\eth}^2\kappa_{1}|_{\cZ}=\eth^2\kappa_{1}|_{\cZ}=0$, then the conditions
\begin{align*}
\left(\overline{\eth}^2 \kappa_{1} + \kappa_{2}\overline{\delta}\sigma - \frac{1}{2}\overline{\sigma}\delta\kappa_{2} - 3\overline{\alpha}\kappa_{2}\overline{\sigma} -\overline{\Psi}_{1}\kappa_{2}\right)|_{\Htwo}&=0, \\
\left(\eth^{2}\kappa_{1} + \kappa_{2}\delta\sigma + \frac{3}{2}\sigma\delta\kappa_{2} + \overline{\alpha}\sigma\kappa_{2}\right)|_{\Htwo}&=0,
\end{align*}
applied in Lemmas \ref{KillingSpinorH2} and \ref{KillingVectorH2}, are automatically satisfied.
\end{lemma}

\medskip
The only remaining condition on $\Hone$ to be considered is from Lemma
\ref{KillingSpinorH1}, which reduces to
\begin{equation}
\left(2\Psi_{3}\kappa_{1} - 3\Psi_{2}\kappa_{2}\right)|_{\Hone}=0
\label{H1curvcondition}
\end{equation}
due to the requirement that $\kappa_{0}|_{\Hone}=0$. One can also use
this requirement to show that
\begin{equation*}
D^{2}\left(2\Psi_{3}\kappa_{1} - 3\Psi_{2}\kappa_{2}\right)|_{\Hone}=0.
\end{equation*}
In fact, the right hand side of this expression can be shown to be homogeneous in
$\kappa_{0}$ and derivatives of $\kappa_{0}$ intrinsic to $\Hone$.
This can be thought of as a second order ordinary differential
equation along the geodesic generators of $\Hone$. Therefore, equation
\eqref{H1curvcondition} is equivalent to the vanishing of
$\left(2\Psi_{3}\kappa_{1} - 3\Psi_{2}\kappa_{2}\right)$ and its first
$D$-derivative on $\cZ$. This combination vanishes on $\cZ$ if
$\kappa_{2}|_{\Htwo}=0$ as it follows from Table
\ref{Table:CharacteristicInitialData} that $\Psi_{3}|_{\cZ}=0$.
The vanishing of the first derivative on $\cZ$ can be shown
to be equivalent to
\begin{equation}
\overline{\delta}\left(\kappa_{1}^{3}\Psi_{2}\right)|_{\cZ}=0\,.
\label{bardeltakappa1to3Psi2}
\end{equation}
In a similar way, the only remaining condition on $\Htwo$ to be
analysed is from Lemma \ref{KillingSpinorH2}. This condition reduces to
\begin{equation}
\left(2\Psi_{1}\kappa_{1} - 3\Psi_{2}\kappa_{0}\right)|_{\Htwo}=0
\label{H2curvcondition}
\end{equation}
due to the requirement that $\kappa_{2}|_{\Htwo}=0$. One can also use this requirement to show that
\begin{equation*}
\Delta^{2}\left(2\Psi_{1}\kappa_{1} - 3\Psi_{2}\kappa_{0}\right)|_{\Htwo}=0.
\end{equation*}
In fact, the right hand side of this can be shown to be homogeneous in
$\kappa_{2}$ and derivatives of $\kappa_{2}$ intrinsic to
$\Htwo$. This can be thought of as a second order ordinary
differential equation along the geodesic generators of
$\Htwo$. Therefore, equation \eqref{H2curvcondition} is equivalent to
the vanishing of $\left(2\Psi_{1}\kappa_{1} -
3\Psi_{2}\kappa_{0}\right)$ and its first $\Delta$ derivative on
$\cZ$. This combination vanishes on $\cZ$ if $\kappa_{0}|_{\Hone}=0$
as, following \ref{Table:CharacteristicInitialData}, one has that $\Psi_{1}|_{\cZ}=0$. The vanishing of the
first derivative on $\cZ$ can be shown to be equivalent to
\begin{equation}
\delta\left(\kappa_{1}^{3}\Psi_{2}\right)|_{\cZ}=0\,.
\label{deltakappa1to3Psi2}
\end{equation}
It follows then from equations \eqref{deltakappa1to3Psi2} and \eqref{bardeltakappa1to3Psi2}, 
\begin{equation}
\mathfrak{K}\equiv\kappa_{1}^{3}\Psi_{2}
\label{MathfrakConstantRelation}
\end{equation}
is constant $\mathfrak{K}\in\mathbb{C}$ on $\cZ$. 

\medskip
Summarising the discussion of this section one has that:
\begin{lemma}
Assume that $\kappa_{0}|_{\Hone}=\kappa_{2}|_{\Htwo}=0$. Then
$\mathfrak{K}\equiv \kappa_{1}^{3}\Psi_{2}$ is constant on $\cZ$ if and only if
\begin{align*}
\left(2\Psi_{3}\kappa_{1}-3\Psi_{2}\kappa_{2}\right)|_{\Hone} &= 0, \\
\left(2\Psi_{1}\kappa_{1}-3\Psi_{2}\kappa_{0}\right)|_{\Htwo} &=0.
\end{align*}
\end{lemma}

\begin{remark}\label{re: MathfrakConstantRelation}
{\em Note that
\begin{align*}
D\mathfrak{K}|_{\Hone} &= \frac{3}{2}\Psi_{2}\kappa_{1}^{2}\left(-\overline{\delta}\kappa_{0} + 2\alpha\kappa_{0}\right)|_{\Hone} \\
&= 0
\end{align*}
where we have used equation $D\Psi_{2}|_{\Hone}=0$ from Table
\ref{Table:CharacteristicInitialData}, equation \eqref{NPHH16} and the
requirement that $\kappa_{0}|_{\Hone}=0$. Similarly,
\begin{align*}
\Delta\mathfrak{K}|_{\Htwo} &= \frac{3}{2}\Psi_{2}\kappa_{1}^{2}\left(-\delta\kappa_{2} - 2(\beta+\tau)\kappa_{2}\right)|_{\Htwo} \\
&= 0
\end{align*}
where we have used equation $\Delta\Psi_{2}|_{\Htwo}=0$ from Table
\ref{Table:CharacteristicInitialData}, equation \eqref{NPHH25} and the
requirement that $\kappa_{2}|_{\Htwo}=0$. Thus, $\mathfrak{K}$ is
constant not merely on $\cZ$ but on the whole of
$\Hone\cup\Htwo$. Since the Newman-Penrose reduced system coupled to
the wave equation for $\kappa_{AB}$, equation
\eqref{KillingSpinorWaveEquation}, is a well-posed hyperbolic system
we also have that $\mathfrak{K}$ is, in fact, constant throughout the
domain of dependence of $\Hone\cup\Htwo$.}
\end{remark}

\subsection{Summary}
Collecting all the previous  results together one obtains the following:

\begin{proposition}
\label{HSZeroReduceToZ}
Assume that the spacetime obtained from the characteristic initial value problem in 
$D(\mathcal{H}_1\cup\mathcal{H}_2)$ is not diffeomorphic to
the Minkowski spacetime. Then the following two statements are
equivalent:
\begin{align*}
(i)\qquad&\text{Given a spin basis } \{o^A,\,\iota^A\}\text{ on }\cZ,\ \text{there exist a constant}\ \ \mathfrak{K} \in\mathbb{C}\ \text{such that} \\
&\kappa_{0} = 0\,, \quad \eth^2\kappa_{1} = \overline{\eth}{}^2\kappa_{1} = 0\,,\quad \kappa_{2} = 0\quad \text{and}\quad  \kappa_{1}^{3}\Psi_{2}=\mathfrak{K}
   \quad \mbox{on}\quad \cZ\,. \\
(ii)\qquad& H_{A'ABC}=0,\quad S_{AA'BB'}=0\quad\mbox{everywhere on}
            \quad\Hone\cup\Htwo.
\end{align*}
\end{proposition}

Recall that the vanishing of the spinors $H_{A'ABC}$ and $S_{AA'BB'}$
on $\Hone\cup\Htwo$ are precisely the conditions of Proposition
\ref{KillingSpinorExistence}, which along with the assumptions of
Theorem \ref{Theorem:ExistenceCIVPSpacetime} imply that the Killing
spinor candidate $\kappa_{AB}$ is in fact a Killing spinor in the
causal future (or past) of $\cZ$.  By summing up these observations we get:

\begin{theorem}
\label{Theorem:ExistenceKillingSpinors}
Let $(\mathcal{M},\bmg)$ be a vacuum spacetime satisfying the
conditions of Theorem \ref{Theorem:ExistenceCIVPSpacetime}. Given a
spin basis $\{o^A,\,\iota^A\}$ on $\cZ$, assume that there exist a
constant $\mathfrak{K} \in\mathbb{C}$ such that  the conditions 
\begin{eqnarray}
& \kappa_{0} = 0,\quad \eth^2\kappa_{1} = \overline{\eth}{}^2\kappa_{1} = 0\,, \quad
                                  \kappa_{2} = 0 \quad\text{and}\quad \kappa_{1}^{3}\Psi_{2}=\mathfrak{K} 
 \label{CharacteristicKSID} 
\end{eqnarray}
hold on $\cZ$.  Then the
corresponding unique solution $\kappa_{AB}$ to equation
\eqref{KillingSpinorWaveEquation} is a Killing spinor everywhere on the domain of dependence of $\mathcal{H}_1\cup\mathcal{H}_2$.
\end{theorem}

\begin{proof}
Due to Theorem \ref{HSZeroReduceToZ} we have that $H_{A'ABC}$ and
$S_{AA'BB'}$ vanish on $\Hone\cup\Htwo$. Data for $\kappa_{0},
\kappa_{1}, \kappa_{2}$ on $\Hone$ and $\Htwo$ are determined by their
values on $\cZ$ by \eqref{NPHH11}, \eqref{NPHH16}, \eqref{NPHH17},
\eqref{NPHH22}, \eqref{NPHH24} and \eqref{NPHH25}, so Proposition
\ref{ExistenceKSWaveEquation} says that there exists a unique solution
to \eqref{KillingSpinorWaveEquation} on
$D(\Hone\cup\Htwo)$. Proposition \ref{KillingSpinorExistence} then
says that this field $\kappa_{AB}$ satisfies $H_{A'ABC}=0$ on
$D(\Hone\cup\Htwo)$, so is indeed a Killing spinor there.
\end{proof}

\begin{remark}
{\em Condition \eqref{CharacteristicKSID} is a strong restriction on
the form of the Weyl spinor component $\Psi_2$ ---and thus, also of
the curvature of the 2-surface $\cZ$. As it will be seen in 
Section \ref{Section:DeterminationKappa1}, it fixes its
functional form up to some constants. As already discussed in
Remark \ref{Remark:GeometryZ} the Weyl spinor component $\Psi_2$ is
not a basic piece of initial data. In view of \eqref{ImPsi2} condition
\eqref{CharacteristicKSID}, ultimately leads to restrictions on
$\tau$ and $\zeta^{\mathcal{A}}$.
}
\end{remark}

\section{On the role of $\kappa_1$}
\label{Deterninancy}

The purpose of this section is to discuss some of the consequences of
the existence of a Killing spinor field. The extent of these
implications is remarkable.

\subsection{Restrictions on the initial data of distorted black holes}
\label{DeterninancyOnGeom}

Note first that once $\kappa_{1}$ is fixed on $\cZ$, the component
$\Psi_{2}$ gets also to be determined by the relation
\eqref{MathfrakConstantRelation} as
\begin{equation}
\Psi_{2}=\mathfrak{K}\,\kappa_{1}^{-3}
\label{MathfrakConstantRelation2}
\end{equation}
where $\mathfrak{K}$ is some complex number. In turn, we also get
restrictions on the free data ---comprised by the complex vector field
$\zeta^A$ and the spin coefficient $\tau$ on $\cZ$--- as given in
R\'{a}cz's black hole holograph construction in \cite{Rac07,Rac14}.

\medskip
Now, observe that once $\Psi_{2}$ is known, the metric $\bm\sigma$ is
restricted in a great extent. To see this recall first the definition
of $\eth$ and $\overline{\eth}$ given by \eqref{eth_P} in terms of the
function $P$ on $\cZ$. By applying the commutation relation relevant
for $P$, that is of spin-weight one, we get
\[
(\eth\,\overline{\eth}-\overline{\eth}\,\eth)\,P=(\Psi_{2}+\overline{\Psi}_{2})\,P\,,
\]
which by using $\overline{\eth}\,P = 0$, and also by using explicit
$z$- and $\overline{z}$-derivatives \eqref{RePsi2} can be seen to take
the form
\begin{equation}\label{PDEonPbarP}
P\overline{P}\,\partial_z\partial_{\overline{z}}\big(\log(P\overline{P})\big)= -2\,\mbox{Re}(\Psi_{2})\,.
\end{equation}

Similarly, \eqref{ImPsi2} can be seen to impose very strong
restrictions on the spin-coefficient $\tau$. Indeed, using the above
notation \eqref{ImPsi2} takes the form
\begin{equation}\label{Imtau}
 P\,(\partial_z\overline{\tau})-\overline{P}\,(\partial_{\overline{z}}\tau)-\big( \,\overline{\tau}\,P\,(\partial_z\log\overline{P})-\tau\,\overline{P}\,(\partial_{\overline{z}}\log P)\,\big) = -2\,\mbox{i\,Im}(\Psi_{2})\,.
 \end{equation}
By applying the substitutions $\tau\mapsto\tau_{1}+
\mbox{i}\,\tau_{2}$, $P\mapsto P_{1}+ \mbox{i}\,P_{2}$ and
$z\mapsto z_{1}+ \mbox{i}\,z_{2}$ we get, by a direct calculation,
that the real part of \eqref{Imtau} reduces to a homogeneous linear
equation for $\tau_{1}$ and $\tau_{2}$, whereas the vanishing of the
imaginary part can be seen to be a first order linear partial
differential equation for the variables $\tau_{1}$ and $\tau_{2}$ on
$\mathbb{R}^2$, with coordinates $(z_{1},z_{2})$. Once, say,
$\tau_{1}$ is eliminated by the linear algebraic relation, the
corresponding linear partial differential equation can always be
solved for $\tau_{2}$. This completes then the verification of the claim that
  whenever a Killing spinor exists on a distorted vacuum black hole
  spacetime,  the specification of $\kappa_1$ on the bifurcation
  surface is equivalent to the freely specifiable data comprised by
  $\zeta^A$ and $\tau$ there. 

\begin{remark}
{\em It is worth mentioning that under a boost transformation
\[  
l^a\mapsto b\, l^a, \qquad n^a\mapsto b^{-1} n^a,
\]
 where $b$ is a
  smooth positive real function on $\cZ$, the spin connection
  coefficient $\tau$ transforms as 
\[
\tau\mapsto
  \tau + \delta\log b.
\]
This gauge freedom has been
  left open in the black hole holograph construction
\cite{Rac07,Rac14}. Thereby, by solving the first-order
  quasilinear partial differential equation 
\[
\tau_1 + \mbox{Re}(\delta\log b)=0
\]
 for $b$, the
  real part of $\tau$ could be transformed out to the expense of
  having the imaginary part changing as 
\[
\tau_2\mapsto \tau_2 +
  \mbox{Im}(\delta\log b).
\]
 For a simple application of such a gauge
  transformation see the argument below \eqref{Imtau2} in Section
  \ref{Section:Kerr}.}
\end{remark}

\subsection{The explicit form of $\kappa_{AB}$ and $\xi_{AA'}=\nabla^P{}_{A'}\kappa_{AP}$ on the horizon}
\label{DeterninancyOnH1UH2}

It is also instructive to compute the explicit form of the Killing
spinor $\kappa_{AB}$ and the associated Killing vector field
$\xi_{AA'}=\nabla^P{}_{A'}\kappa_{AP}$ on the horizon
$\Hone\cup\Htwo$.

\medskip

As for the explicit form of the Killing spinor note first that, in virtue of Theorem \ref{Theorem:ExistenceKillingSpinors}, on $\cZ$ we have 
\begin{eqnarray}
& \kappa_{0} = 0,\quad \eth^2\kappa_{1} = \overline{\eth}{}^2\kappa_{1} = 0\,, \quad
\kappa_{2} = 0 \,. 
\label{CharacteristicKSIDZ} 
\end{eqnarray}
Using then \eqref{NPHH11}, \eqref{NPHH16} and \eqref{NPHH17}, and by commuting $D$ and $\delta$ derivatives, we get that on $\Hone$  
\begin{eqnarray}
& \kappa_{0} = 0,\quad \kappa_{1} = \kappa_{1}|_{\cZ}  \,, \quad
\kappa_{2} = -2\,r\,\overline{\eth}{}\kappa_{1} \,. 
\label{CharacteristicKSIDH1} 
\end{eqnarray}
Analogously, by \eqref{NPHH22}, \eqref{NPHH24} and \eqref{NPHH25}, and by commuting $\Delta$ and $\delta$  derivatives,  we get on $\Htwo$ 
\begin{eqnarray}
& \kappa_{0} = -2\,u\,({\eth}{}\kappa_{1}+\tau\,\kappa_{1}),\quad \kappa_{1} = \kappa_{1}|_{\cZ}  \,, \quad
\kappa_{2} =  0\,. 
\label{CharacteristicKSIDH2} 
\end{eqnarray}
 These observations are summarised in Table\,\ref{Table: kappa}. 

{\renewcommand{\arraystretch}{1.5}
\begin{table}[t]
	\centering \small \hskip-.15cm
	\begin{tabular}{|c|c|c|} 
		\hline ${{\mathcal{H}}_1}$ &  ${{\mathcal{Z}}}$
		 & ${{\mathcal{H}}_2}$ \\ \hline \hline
		 $\kappa_{0} = 0$ &  $\kappa_{0} = 0$   & \  $\kappa_{0} = -2\,u\,({\eth}{}\kappa_{1}+\tau\,\kappa_{1})$ \ \\  \hline 
		$\kappa_{1} = \kappa_{1}|_{\cZ}$ &  $\quad \kappa_{1}:\  \eth^2\kappa_{1} = \overline{\eth}{}^2\kappa_{1} = 0$ \hfill  & $\kappa_{1} = \kappa_{1}|_{\cZ}$   \\  \hline 
		\quad $\kappa_{2} = -2\,r\,\overline{\eth}{}\kappa_{1}$ \ \quad &  $\kappa_{2} = 0$ & $\kappa_{2} = 0$ 
		\\ \hline
	\end{tabular}
	\caption{\small The components of the Killing spinor field $\kappa_{AB}$ on the null
		hypersurface $\mathcal{H}_1\cup\mathcal{H}_2$. }
	\label{Table: kappa}
\end{table}
}

\medskip
In exactly the same way, the components of the Killing vector field
$\xi_{AA'}=\nabla^P{}_{A'}\kappa_{AP}$ can be determined by equations
\eqref{KVatZ1}-\eqref{KVatZ4} on $\cZ$, by \eqref{KVEqnNPH11},
\eqref{KVEqnNPH13}, \eqref{KVEqnNPH16} and \eqref{KVEqnNPH110} on
$\Hone$, as well as, by \eqref{KVEqnNPH22}, \eqref{KVEqnNPH23},
\eqref{KVEqnNPH24} and \eqref{KVEqnNPH27} on $\Htwo$, along with
commuting derivatives on $\Hone$ and $\Htwo$, respectively. The
corresponding explicit formulas are collected in Table \ref{Table:
xiAA'}.

{\renewcommand{\arraystretch}{1.5}
	\begin{table}[t]
		\centering \small \hskip-.15cm
		\begin{tabular}{|c|c|c|} 
	\hline ${{\mathcal{H}}_1}$ &  ${{\mathcal{Z}}}$
	 & ${{\mathcal{H}}_2}$ 
			\\ \hline \hline
	\ \ $\xi_{11'} =  -3\,r\,(\tau\,\overline{\eth}{}\kappa_{1}-\overline{\tau}\,{\eth}{}\kappa_{1})$\ \  &  $\xi_{11'}=0$   &  $\xi_{11'}=0$ 
			\\  \hline 
	$\xi_{10'} =  -3\,\overline{\eth}{}\kappa_{1}$  & $\xi_{10'} =  -3\,\overline{\eth}{}\kappa_{1}$ & $\xi_{10'} =  -3\,\overline{\eth}{}\kappa_{1}$  
			\\  \hline 
	$\xi_{10'} =  3\,{\eth}{}\kappa_{1}$  & $\xi_{10'} =  3\,{\eth}{}\kappa_{1}$ & 
	$\xi_{10'} =  3\,{\eth}{}\kappa_{1}$   
			\\  \hline 
	$\xi_{00'}=0$ &  $\xi_{00'}=0$   & \ \  $\xi_{00'} =  -3\,u\,(\tau\,\overline{\eth}{}\kappa_{1}-\overline{\tau}\,{\eth}{}\kappa_{1})$\ \ 
			\\ \hline
		\end{tabular}
		\caption{\small The components of the Killing vector field  $\xi_{AA'}$ on the null
			hypersurface $\mathcal{H}_1\cup\mathcal{H}_2$. }
		\label{Table: xiAA'}
	\end{table}
}

\begin{remark}
{\em 
In order to proceed with the interpretation of the above expressions
recall, first, that any of the the distorted black hole configurations
was shown \cite{Rac07,Rac14} to admit a horizon Killing vector field
of the form
\[
K^a=
\begin{cases}
-r\,({\partial/\partial r})^a\,, \hskip2mm \mbox{on} \hskip2mm \Hone\,, \\
\hskip2.5mm u\,({\partial/\partial u})^a\,, \hskip2mm \mbox{on} \hskip2mm \Htwo \,.
\end{cases}
\]
Note also, in Gaussian null coordinates $(u,r,x^\mathcal{A})$ the
coordinate functions $u$ and $r$  are affine parameters along he
generators of $\Hone$ and $\Htwo$, respectively. Accordingly, the
components $\xi_{11'}$ and $\xi_{00'}$ of the Killing vector field are
not constant in these coordinates. They would be constant if they were
to be expressed in terms of the associated Killing parameters instead
of the affine ones. Notably, this behaviour of $\xi_{AA'}$ on the
horizon $\mathcal{H}_1\cup\mathcal{H}_2$ resembles that of the
asymptotically time translational Killing vector field
$(\partial/\partial t)^a$ of the Kerr solution. Recall that the orbits of
$(\partial/\partial t)^a$ repeatedly and periodically intersect the generators
of the horizons $\Hone$ and $\Htwo$, respectively, and also that
$(\partial/\partial t)^a$ reduces to an axial Killing vector field on the
bifurcation surface. }
\end{remark}

\subsection{The Petrov type of the domain of dependence}\label{PetrovType}

It has been known for long \cite{WalPen70} (see also \cite{PenRin84})
that the existence of the Killing spinor $\kappa_{AB}$ satisfying
\eqref{KillingSpinorEquation} imposes strong restrictions on the
self-dual Weyl spinor $\Psi_{ABCD}$ via the integrability condition
\[
\Psi_{(ABC}{}^F\kappa_{D)F}=0\,.
\]   
Namely, if neither $\Psi_{ABCD}$ nor $\kappa_{AB}$ vanishes there
exist some scalar $\psi$ such that
\[
\Psi_{ABCD}=\psi\, \kappa_{AB}\kappa_{CD}\,,
\]
implying, in particular, that $\Psi_{ABCD}$ \emph{must be of Petrov
type D or N}. It is also known that if the Killing spinor field
$\kappa_{AB}$ is generic, i.e.~$\kappa_{AB}=\alpha_{(A}\beta_{B)}$,
for some $\alpha_{A}\not=\beta_{A}$ spinors, $\Psi_{ABCD}$ \emph{is of
Petrov type D}.

\medskip
As this integrability condition had been used (see
Remark\,\ref{IntegrabilityCondition}) in the previous sections in
identifying the conditions on the initial data for $\kappa_{AB}$ on
$\mathcal{H}_1\cup\mathcal{H}_2$, and both $\kappa_{AB}$ and
$\Psi_{ABCD}$ are known to satisfy wave equations that are linear and
homogeneous in these variables (see e.g.~\cite{BaeVal10b}) the
integrability condition immediately holds everywhere in the domain of
dependence of $\mathcal{H}_1\cup\mathcal{H}_2$.

\medskip
Note also that, as for the Killing spinor field $\kappa_{AB} =
\kappa_2 o_A o_B - 2\kappa_1 o_{(A} \iota_{B)} + \kappa_0 \iota_A
\iota_B$ holds, whenever $\kappa_1$ is non identically zero (which
could only happen in the flat case) $\kappa_{AB}$ is guaranteed to be
generic. All these observations verify the following :

\begin{corollary}\label{cor: PetrovType}
Let $(\mathcal{M},\bmg)$ be a vacuum spacetime satisfying the
conditions of Theorem \ref{Theorem:ExistenceCIVPSpacetime}. If
$\kappa_1$ is not identically zero on the bifurcation surface $\cZ$
then the corresponding distorted black hole spacetime is of Petrov
type D everywhere in the domain of dependence of
$\mathcal{H}_1\cup\mathcal{H}_2$.
\end{corollary}

\section{Axial symmetry of the bifurcation surface $\cZ\approx\mathbb{S}^2$}
\label{AxialSymmetry:BifurcationSphere}

As argued below, whenever the bifurcation surface $\mathcal{Z}$
possesses the topology of a two-sphere, $\mathbb{S}^2$, the conditions
in \eqref{CharacteristicKSID} immediately imply the existence of an
axial Killing vector field on $\mathcal{Z}$.

\subsection{Existence of a Killing vector on $\mathcal{Z}$}

 We begin by observing that as a consequence of equations
\eqref{KVatZ1} and \eqref{KVatZ4} then if $\kappa_0=\kappa_2=0$ then
necessarily $\xi_{00'}=\xi_{11'}=0$. Thus, under the assumptions of
Theorem \ref{Theorem:ExistenceKillingSpinors}, the Killing vector
$\xi_{AA'}$ is tangent to $\mathcal{Z}$ ---i.e. its only non-vanishing
components are $\xi_{01'}$ and $\xi_{10'}$.  As we have seen in
Subsection \ref{KVFonZ} the existence of a (possibly complex) Killing
vector field on $\mathcal{Z}$ is equivalent to
\eqref{ConstraintAtZ1}-\eqref{ConstraintAtZ3} on $\mathcal{Z}$ which
had also been seen to be equivalent to the vanishing of $S_{AA'BB'} =
\nabla_{AA'} \xi_{BB'} + \nabla_{BB'}\xi_{AA'}$ on $\cZ$.  Thereby,
the conditions of Theorem \ref{Theorem:ExistenceKillingSpinors} are
equivalent to the existence of a (possibly complex) Killing vector
field on $\mathcal{Z}$.

\begin{remark}\label{axPsi2kappa1}
	{\em As all the geometric quantities including the spin
coefficients, as well as, the Weyl spinor components $\Psi_{ABCD}$ are
constructed from the metric $\bmg$, given by \eqref{LineElement}, and
$\xi_{AA'}=\nabla^P{}_{A'}\kappa_{AP}$ is known to be a Killing vector
field everywhere in $D(\Hone\cup\Htwo)$ we immediately have that
		\[
		\xi^{AA'}\nabla_{AA'} \tau=0\,, \quad \xi^{AA'}\nabla_{AA'} \Psi_{2}=0 
		\]
		and, in virtue of \eqref{MathfrakConstantRelation} and the argument in Remark \ref{re: MathfrakConstantRelation} above, that 
		\[
		\xi^{AA'}\nabla_{AA'} \kappa_{1}=0 
		\] 
		everywhere in $D(\Hone\cup\Htwo)$. Thus we shall use
from now on, without loss of generality, that $\tau, \Psi_{2}$ and
$\kappa_{1}$, when they are restricted to $\cZ$, they all respect the
axial symmetry of the metric on $\mathcal{Z}$.
	}
\end{remark}

\subsection{The axial Killing vector}
 By our assumption on the underlying smoothness of the setting the
Killing vector field $\xi_{AA'}$ is smooth on $\mathcal{Z}$. If
$\xi_{AA'}$ was also Hermitian---i.e.
\[
\xi_{01'} = \overline{\xi}_{10'}\,,
\]
would hold, then, by appealing to the \emph{hairy ball theorem} this
Killing vector field vanished at some point, say at $p\in
\mathcal{Z}$. As ---apart from the trivial case when $\kappa_1$ is
constant on $\mathcal{Z}$--- $\xi_{AA'}$ was not identically zero on
$\mathcal{Z}$, and by applying the argument of Wald---see pages
119-120 in \cite{Wal94}---$\xi_{AA'}$ had to be an axial Killing
vector field on $\mathcal{Z}$ with closed orbits, with some fixed
periodicity, around the \emph{fixed point} $p\in \cZ$.

\medskip
In returning now to the generic case note that the argument just
outlined applies to the real and imaginary parts of $\xi_{AA'}$,
separately. Thereby, whenever $\xi_{AA'}$ is non-Hermitian the metric
on $\mathcal{Z}$ has to admit both $\xi_{AA'}+\overline{\xi}_{AA'}$
and $\mbox{i}\,(\xi_{AA'}-\overline{\xi}_{AA'})$ as real Killing
vector fields. If both of the are non-trivial they either vanish at
the same location on $\mathcal{Z}$ or not. If both vanish at $p\in
\mathcal{Z}$ they must be proportional and the factor of
proportionality is determined by the ration of their individual
periodicities. If their vanishing occurs at two different points of
$\mathcal{Z}$ then $\xi_{AA'}+\overline{\xi}_{AA'}$ and
$\mbox{i}\,(\xi_{AA'}-\overline{\xi}_{AA'})$ must be linearly
independent real axial Killing vector fields on $\mathcal{Z}$
implying, in virtue of \eqref{CharacteristicKSID}, that
$\kappa_1=const$ and, in turn, that $\Psi_2=const$ and $\tau=0$ which
implies then that the metric $\bm\sigma$ on $\mathcal{Z}$ is
spherically symmetric.

\medskip
We can summarise the discussion of this section in the following:

\begin{proposition}
\label{Proposition:AxialSymmetryBifurcationSphere}
Assume that the spacetime obtained from the characteristic initial
value problem in $D(\mathcal{H}_1\cup\mathcal{H}_2)$ admits a Killing
spinor $\kappa_{AB}$ such that \eqref{CharacteristicKSID} hold on
$\mathcal{Z}$. Then $\xi_{AA'}=\nabla^P{}_{A'}\kappa_{AP}$ gives rise
to a (possibly complex) axial Killing vector field on $\mathcal{Z}$.
\end{proposition}

\section{Determining $\kappa_1$ on $\cZ$}
\label{Section:DeterminationKappa1}

As we have seen in Section \ref{AxialSymmetry:BifurcationSphere}, once
$\mathcal{Z}$ is assumed to have the topology of a 2-sphere, the
spinor $\xi_{AA'}$ is guaranteed to be an axial Killing vector field
on $\mathcal{Z}$. By restricting our considerations to this case, the
purpose of this section is to explicitly determine $\kappa_1$
satisfying the equations

\begin{equation}
\eth^2\kappa_1 =0\,, \qquad \overline{\eth}^2 \kappa_1=0 \,.
\label{EthEthKappaEquations}
\end{equation}
This can be done in the most effective way by using \emph{coordinates
  adapted to the axial symmetry} of $\bm\sigma$ as introduced in
\cite{DobLewPaw18}. Therefore, for the sake of completeness, we shall outline
the argument applied Section IV of \cite{DobLewPaw18} in the following
subsection.

\subsection{Coordinates adapted to axial symmetry}

Recall, first, that by assumption the 2-dimensional manifold $\mathcal{Z}$ has the
topology of the 2-sphere $\mathbb{S}^2$. Thus, by the \emph{Riemann
	mapping theorem} the metric $\bm \sigma$ is conformal to the standard
round metric of $\mathbb{S}^2$. Taking into account the assumption of
axial symmetry one writes then
\begin{equation}
{\bm\sigma} = \varpi^2 \big( \mathbf{d}\theta\otimes
\mathbf{d}\theta + \sin^2 \theta
\mathbf{d}\varphi\otimes\mathbf{d}\varphi \big)
\label{OldMetricZ}
\end{equation}
where $(\theta,\varphi)$ are standard \emph{spherical coordinates} on
$\mathbb{S}^2$ and $\varpi=\varpi(\theta)$ is a suitable conformal factor
depending only on the \emph{colatitude} $\theta$. If $\bm \sigma$ is a
smooth metric, then the conformal factor is also a \emph{strictly
	positive} scalar field over $\mathcal{Z}$. The key idea behind
the explicit integration of the equations in \eqref{EthEthKappaEquations} is to introduce
a new coordinate $\psi$ given by the condition
\[
\mathbf{d}\psi = \frac{\varpi^2}{R^2}\sin\theta \mathbf{d}\theta,
\]
where $R$ is the \emph{area radius} defined by 
\begin{eqnarray*}
	&& 4\pi R^2 = \int_{\mathcal{Z}} \varpi^2 \sin \theta \mathbf{d}
	\theta \wedge \mathbf{d}\varphi \\
	&& \phantom{4\pi R^2}= R^2 \int_{\mathcal{Z}} \mathbf{d}\psi \wedge
	\mathbf{d}\varphi = 2\pi R^2(\psi_1-\psi_0).
\end{eqnarray*}
Setting, without loss of generality, $\psi_0=-1$ one has that
$\psi_1=1$ ---these coordinate values correspond, respectively, to
the North and South poles of $\mathcal{Z}$ defined by the conditions
$\theta=0$ and $\theta=\pi$. Thus, the coordinate $\psi$ is defined on the range
$[-1,1]$. Defining, for convenience, the function $Q=Q(\psi)$ by 
\begin{equation}
Q\equiv \frac{\varpi^2 \sin^2 \theta}{R^2},
\label{DefinitionP}
\end{equation}
the metric \eqref{OldMetricZ} takes, in terms of the coordinates
$(\psi,\phi)$ the form
\begin{equation}
{\bm\sigma} = R^2 \left( \frac{1}{Q(\psi)^2} \mathbf{d}\psi \otimes
\mathbf{d}\psi + Q(\psi)^2 \mathbf{d}\varphi\otimes
\mathbf{d}\varphi \right).
\label{NewMetricZ}
\end{equation}
In particular, from \eqref{DefinitionP} we have that
\begin{equation}
Q(-1)=Q(1)=0.
\label{RegularityP}
\end{equation}

A direct computation then shows that the Levi-Civita connection of
${\bm\sigma}$ ---encoded in the combination $\alpha-\overline\beta$
[see e.g.~\eqref{eths}]--- is given
in terms of the function $Q$ by
\begin{equation}
\alpha -\overline\beta = - \frac{1}{\sqrt{2}R} \partial_\psi Q \equiv - \frac{1}{\sqrt{2}R} Q'.
\label{ConnectionZ}
\end{equation}

\subsection{Integration of the equations for $\kappa_1$}
We now make use of the coordinates introduced in the previous
subsection to integrate the equations in
\eqref{EthEthKappaEquations}. 

\medskip
Consistent with the discussion in Section \ref{AxialSymmetry:BifurcationSphere} we look for solutions
which are axially symmetric. To this end we observe that, in terms of
the coordinates $(\psi,\varphi)$, the directional derivatives $\delta$
and $\overline\delta$ acting on scalars are given by
\[
\delta = \frac{1}{\sqrt{2}R}\bigg( Q \partial_\psi +
\frac{\mbox{i}}{Q}\partial_\varphi\bigg), \qquad \overline\delta = \frac{1}{\sqrt{2}R}\bigg( Q \partial_\psi -
\frac{\mbox{i}}{Q}\partial_\varphi\bigg).
\]
As it follows from the argument applied in Remark \ref{axPsi2kappa1} $\kappa_1$ is also axially symmetric,
i.e. $\partial_\varphi \kappa_1=0$. Therefore the two conditions $\eth^2
\kappa_1=0$ and $\overline\eth^2 \kappa_1=0$ are no longer independent (in
fact they are equivalent!). Then, in virtue of \eqref{eths},
\begin{eqnarray*}
	&& \eth^2 \kappa_1 = \left( \frac{1}{\sqrt{2}R} Q \partial_\psi -
	\frac{1}{\sqrt{2}R}\partial_\psi Q \right)\left(\frac{1}{\sqrt{2}R}
	Q \partial_\psi \kappa_1\right)\\
	&& \phantom{\eth^2 \kappa_1} = \frac{Q^2}{2R^2}\partial^2_\psi \kappa_1=0\,,
\end{eqnarray*} 
from which one readily obtains the solution
\begin{equation}
\kappa_1 = \mathfrak{c} \psi + \mathfrak{b}, \qquad
\mathfrak{c},\, \mathfrak{b}\in \mathbb{C}. 
\label{AxisymmetricKappa}
\end{equation}
From this solution, recalling the relation \eqref{MathfrakConstantRelation} one readily
obtains the following expression for $\Psi_2$:
\begin{equation}
\Psi_2 = \frac{\mathfrak{K}}{(\mathfrak{c} \psi + \mathfrak{b})^3}.
\label{Psi2AxialSymmetry}
\end{equation}

\subsubsection{The Gauss-Bonnet condition}
The Weyl scalar is related to the Gaussian curvature of
2-surfaces ---see \cite{PenRin84}, Section 4.14. In particular, for
the 2-surface $\mathcal{Z}$ one has that it is given by
$K_{\mathcal{G}}=-2 \mbox{Re}\, \Psi_2$. It follows then that the
Gauss-Bonnet formula applied to $\mathcal{Z}\approx \mathbb{S}^2$
implies
\begin{equation}
\int_{\mathcal{Z}} \Psi_2 \mbox{d}S = -2\pi 
\label{GaussBonnet}
\end{equation}
---see equation (4.14.44) in \cite{PenRin84}. Taking into account the
line element \eqref{NewMetricZ} one finds that
\begin{eqnarray*}
&& \int_{\mathcal{Z}} \Psi_2 \mbox{d}S = R^2 \mathfrak{K} \int_{-1}^1
\int_0^{2\pi} \frac{\mbox{d}\varphi \mbox{d}\psi }{(\mathfrak{b}+ \mathfrak{c}
\psi)^3}\\
&& \phantom{\int_{\mathcal{Z}} \Psi_2 \mbox{d}S}= \frac{4\pi R^2
\mathfrak{K} \mathfrak{b} }{(\mathfrak{b}^2-\mathfrak{c}^2)^2}.  
\end{eqnarray*}
Thus, from \eqref{GaussBonnet} one obtains the condition
\begin{equation}
\frac{2 R^2 \mathfrak{K}\mathfrak{b} }{(\mathfrak{b}^2-\mathfrak{c}^2)^2}=-1.
\label{ConditionGaussBonnet}
\end{equation}
	
\begin{remark}
\label{Remark:GaussBonnet}
{\em Condition \eqref{ConditionGaussBonnet}, being a consequence of
the Gauss-Bonnet formula, is a necessary condition for $-2 \mbox{Re}\,
\Psi_2$ to be the Gaussian curvature of a smooth 2-surface. It, can be
used to fix the value of the radius $R$. Observe, also that it implies
that the combination
\[
\frac{\mathfrak{K}\mathfrak{b}}{(\mathfrak{b}^2-\mathfrak{c}^2)^2}
\] 
must be real. It will be seen in Subsection
\ref{Section:ApplyingMars} that for the Kerr spacetime one necessarily
has that $\mathfrak{K}$ must be real and $\mathfrak{c}$ pure
imaginary. If this is the case, then $\mathfrak{b}$ must also be
real. }
\end{remark}
	
\subsection{Integrating the function $Q$}
Equation \eqref{RePsi2} can be used to compute the explicit
form of the function $Q$ appearing in the line element
\eqref{NewMetricZ}. Taking into account \eqref{ConnectionZ}
and \eqref{Psi2AxialSymmetry}, equation \eqref{RePsi2}  implies then 
\[
(Q Q')'= \frac{\mathfrak{K} R^2}{(\mathfrak{b} +
\mathfrak{c}\psi)^3}+\frac{\overline{\mathfrak{K}}
R^2}{(\overline{\mathfrak{b}} + \overline{\mathfrak{c}}\psi)^3}\,. 
\]
This expression can be readily integrated to get
\begin{equation}
Q^2 = C_2 + C_1 \psi +
\frac{\mathfrak{K}R^2}{\mathfrak{c}^2(\mathfrak{b}+\mathfrak{c}\psi)}+
\frac{\overline{\mathfrak{K}}R^2}{\overline{\mathfrak{c}}^2(\overline{\mathfrak{b}}+\overline{\mathfrak{c}}\psi)} 
\label{GeneralSolutionFrameCoefficient}
\end{equation}
with $C_1$ and $C_2$ some real integration constants which are fixed
using the conditions in \eqref{RegularityP}. A direct computation
shows then that
\begin{subequations}
\begin{eqnarray}
&& C_1=
		\frac{\mathfrak{K}R^2}{\mathfrak{c}(\mathfrak{b}^2-\mathfrak{c}^2)}+ \frac{\overline{\mathfrak{K}}R^2}{\overline{\mathfrak{c}}(\overline{\mathfrak{b}}^2-\overline{\mathfrak{c}}^2)},
		\label{ConstantC1}  \\	
&& C_2=
		-\frac{\mathfrak{K}R^2\mathfrak{b}}{\mathfrak{c}^2(\mathfrak{b}^2-\mathfrak{c}^2)}
		-\frac{\overline{\mathfrak{K}}R^2\overline{\mathfrak{b}}}{\overline{\mathfrak{c}}^2(\overline{\mathfrak{b}}^2-\overline{\mathfrak{c}}^2)}.
		\label{ConstantC2}
\end{eqnarray}
\end{subequations}
	
	\begin{remark}
		{\em The constants $R$, $\mathfrak{K}$, $\mathfrak{b}$ and
			$\mathfrak{c}$ in \eqref{GeneralSolutionFrameCoefficient} and
			\eqref{ConstantC1}-\eqref{ConstantC2} are subject to the constraint
			\eqref{ConditionGaussBonnet} arising from the Gauss-Bonnet identity. } 
	\end{remark}
	
	\begin{remark}
		{\em In order to ensure the regularity of the function $Q$ and of the
			associated curvature of $\mathcal{Z}$, it is necessary that the ratio
			$-\mathfrak{b}/\mathfrak{c} \in \mathbb{C}\setminus[-1,1]$ ---that is,
			$-\mathfrak{b}/\mathfrak{c}$ can lie in any point of the complex plane
			except the interval $[-1,1]$ on the real axis. Moreover, in order for
			the expression \eqref{GeneralSolutionFrameCoefficient} to be well
			defined, the constants $R$, $\mathfrak{K}$, $\mathfrak{b}$ and
			$\mathfrak{c}$ have to be such that the left hand side of the
			expression is non-negative for $\psi\in[-1,1]$.}
	\end{remark}
	
\begin{remark}
	{\em Note that by writing out equation \eqref{ImPsi2}
          explicitly, and by making use of the present setup ---along
          with the explicit $\psi$ dependence of $Q$ and $\Psi_2$---the imaginary part of $\tau$  gets to be uniquely determined as 
\begin{equation}\label{Imtau2}
\mbox{Im}(\tau) = -\frac{\sqrt{2}\,R}{Q(\psi)}\int_{-1}^{\psi}\, \mbox{Im}(\Psi_2(\psi')) \mbox{d}\psi'\,.
\end{equation}
Note, finally, that by making use of the gauge freedom we have in the
black hole holograph construction ---for a related discussion and an
application see the last paragraph of Subsection
\ref{DeterninancyOnGeom}--- the real part of $\tau$ can be set to
zero by performing and axially symmetric boost
transformation with parameter $b=b(\psi)$ given by
\begin{equation}\label{Retau2}
b=\exp\left(-\frac{\sqrt{2}\,R}{Q(\psi)}\int_{-1}^{\psi} \mbox{Re}(\tau(\psi'))\, \mbox{d}\psi'\right)\,.	
\end{equation}	
Remarkably, the axial symmetry of the setup guarantees that
  $\mbox{Im}(\delta\log b)=0$ and, in turn, that the imaginary part of $\tau$ remains unchanged. This, in particular, implies that \eqref{Imtau2} holds independently of the choice made for the axially symmetric boost transformation or, equivalently, for the real part of $\tau$.	}
\end{remark}


\subsection{Summary: distorted black holes with Killing spinors}
	
Summarising the discussion of the previous section we get the following:

\begin{proposition}
\label{Proposition:MetricsZ}
There exists a five (real) parameter family of smooth axial symmetric
2-metrics ${\bm \sigma}$ on $\mathcal{Z}\approx \mathbb{S}^2$ such
that $\kappa_1$ is a solution to the constraints
\[
\eth^2 \kappa_1=0 \quad {\rm and} \quad \overline{\eth}^2 \kappa_1 =0\,,
\]
and such that the curvature condition
\[
\kappa_1^3 \Psi_2 =\mathfrak{K}
\]
also holds.
\end{proposition}

\begin{remark}
{\em By appealing now to the general black hole holograph
construction \cite{Rac14}, as summarised in Theorem
\ref{Theorem:ExistenceCIVPSpacetime}, it follows then that from the
five parameter family of initial data---comprised by the metrics
referred to in Proposition \ref{Proposition:MetricsZ}, along with
pertinent form of $\tau$ determined by \eqref{Imtau2}---on $\cZ$,
there exists a \emph{five parameter family of distorted black hole
configurations} the members of which are uniquely determined
everywhere in the domain of dependence of the initial data surface,
$\mathcal{H}_1\cup\mathcal{H}_2$. }
\end{remark}

\section{Enforcing the Hermiticity of the Killing vector}
\label{Section:HermiticityKillingVector}
In Theorem \ref{Theorem:SpinorialCharacterisationKerr}, the assumption
that the spinor $\xi_{AA'}$ constructed from the Killing spinor
$\kappa_{AB}$ is Hermitian is needed in order to show that the
spacetime is isometric to the Kerr solution. Recall that using
equations \eqref{XiExpanded1}-\eqref{XiExpanded4} the components of
$\xi_{AA'}$ can be expressed in terms of derivatives of the Killing
spinor components $\kappa_{0}, \kappa_{1}$ and
$\kappa_{2}$. Accordingly, the Hermiticity condition leads to further
restrictions on the components $\kappa_{0}, \kappa_{1}$ and
$\kappa_{2}$.  A consequence of the following proposition is that it
suffices to impose restrictions only on the hypersurfaces $\Hone$ and
$\Htwo$.

\begin{proposition}
\label{Proposition:WaveEquationKillingVector}
Let $\kappa_{AB}$ be a solution to equation
\eqref{KillingSpinorWaveEquation}. Then the spinor field $\xi_{AA'}$ satisfies the wave equation
\begin{equation}
\square\xi_{AA'} = -\Psi_{A}{}^{BCD}H_{A'BCD}
\end{equation}
\end{proposition}
\begin{proof}
Follows by commuting derivatives, and using \eqref{KillingSpinorWaveEquation}.
\end{proof}

An immediate consequence of this result is that
\begin{equation}\label{eq: Hermxi}
\square\left(\xi_{AA'}-\overline{\xi}_{AA'}\right) = \overline{\Psi}_{A'}{}^{B'C'D'}\overline{H}_{AB'C'D'} -\Psi_{A}{}^{BCD}H_{A'BCD}
\end{equation}

\medskip
Assuming that the conditions of Lemmas \ref{KillingSpinorH1} and
\ref{KillingSpinorH2} are satisfied and $H_{A'ABC}$ vanishes
$\mathcal{H}_1\cup\mathcal{H}_2$. Then, in virtue of \eqref{eq:
Hermxi}, $\xi_{AA'}-\overline{\xi}_{AA'}$ must also vanish everywhere
on the domain of dependence of $\mathcal{H}_1\cup\mathcal{H}_2$,
guaranteeing thereby that the vector $\xi_{AA'}$ is Hermitian there.

\medskip
This verifies then the following:
\begin{proposition}
\label{Proposition:AxialSymmetryBifurcationSphereRe} Assume
that the spacetime obtained from the characteristic initial value
problem in $D(\mathcal{H}_1\cup\mathcal{H}_2)$ admits a Killing spinor
$\kappa_{AB}$ such that conditions \eqref{CharacteristicKSID} and
\[ 
\eth (\kappa_1 +\overline\kappa_1)=0
\]
 hold on $\mathcal{Z}$. Then
$\xi_{AA'}=\nabla^P{}_{A'}\kappa_{AP}$ is a real axial Killing vector
field on $\mathcal{Z}$.
\end{proposition}

\subsection{Some immediate restrictions}

The Hermiticity of the Killing vector $\xi_{AA'}$ 
is equivalent to the relations
\begin{equation}
\xi_{00'}=\overline{\xi}_{00'}\,, \quad \xi_{01'}=\overline{\xi}_{10'}\,, \quad
\xi_{10'}=\overline{\xi}_{01'}\,, \quad \xi_{11'}=\overline{\xi}_{11'}\,.
\label{HermiticityCondition}
\end{equation}
These conditions will be imposed on $\Hone$ and $\Htwo$ separately.

\medskip
\noindent
\textbf{Conditions on $\Hone$.} On $\Hone$, using the explicit expressions
\eqref{XiExpanded1}-\eqref{XiExpanded4}, the first condition
in \eqref{HermiticityCondition} is trivially satisfied, and the remaining conditions can be shown to be equivalent to
\begin{subequations}
\begin{align}
\delta(\kappa_{1}+\overline{\kappa}_{1})&=0, \label{HermiticityH1a}\\
\overline{\delta}(\kappa_{1}+\overline{\kappa}_{1})&=0, \label{HermiticityH1b}\\
\Delta\kappa_{1}+\tau\kappa_{2} &\quad \mbox{real}, \label{HermiticityH1c}
\end{align}
\end{subequations}
on $\Hone$. In fact, it is straightforward to show that on $\Hone$
\begin{equation*}
D\delta(\kappa_{1}+\overline{\kappa}_{1}) =D\overline{\delta}(\kappa_{1}+\overline{\kappa}_{1})=0\,.
\end{equation*}
Thus, it suffices to impose conditions
\eqref{HermiticityH1a}-\eqref{HermiticityH1b} only on
$\cZ$. In other words, the Hermiticity condition on $\Hone$ is equivalent to
\begin{align*}
&\mbox{Re}(\kappa_{1})  \quad \text{constant on }\cZ, \\
&\Delta\kappa_{1}+\tau\kappa_{2} \quad \mbox{real}\quad  \text{ on } \Hone.
\end{align*}

\medskip
\noindent
\textbf{Conditions on $\Htwo$.} Secondly, on $\Htwo$, the last
condition in \eqref{HermiticityCondition} is trivially satisfied and
the remaining conditions are equivalent to
\begin{subequations}
\begin{align}
\delta(\kappa_{1}+\overline{\kappa}_{1})&=0, \label{HermiticityH2a}\\
\overline{\delta}(\kappa_{1}+\overline{\kappa}_{1})&=0, \label{HermiticityH2b}\\
D\kappa_{1}\quad \mbox{real}, \label{HermiticityH2c}
\end{align}
\end{subequations}
on $\Htwo$. Again, it is straightforward to show that on $\Htwo$
\begin{equation*}
\Delta\delta(\kappa_{1}+\overline{\kappa}_{1}) =\Delta\overline{\delta}(\kappa_{1}+\overline{\kappa}_{1})=0.
\end{equation*}
Consequently, it suffices to impose conditions
\eqref{HermiticityH2a}-\eqref{HermiticityH2b} on $\cZ$. 

\medskip
Combining the discussion of the previous two paragraphs one concludes that the spinor
field $\xi_{AA'}$ is Hermitian on $\Hone\cup\Htwo$ if
and only if we have
\begin{subequations}
\begin{align}
\kappa_{1}+\overline{\kappa}_1 = \ \text{const\ on }&\cZ, \label{HermiticityBasic1}\\
\Delta\kappa_{1}+\tau\,\kappa_{2} \quad  \text{ real\  on\ }  &\Hone, \label{HermiticityBasic2}\\
D\kappa_{1}\quad  \text{ real\ on\ }  &\Htwo. \label{HermiticityBasic3}
\end{align}
\end{subequations}

\subsection{Hermiticity in terms of conditions at $\cZ$}
In this section it is shown that conditions
\eqref{HermiticityBasic2}-\eqref{HermiticityBasic3} can be replaced by restrictions on $\cZ$.

\medskip
\noindent
\textbf{Analysis on $\Htwo$.} Start by considering condition
\eqref{HermiticityBasic3}. From the transport equation
\eqref{KillingSpinorWaveEquationNP2H2} on $\Htwo$, and equation
\eqref{NPHZ7}, we have that
\begin{equation*}
2 \Delta D \kappa_1 = \delta\overline{\delta}\kappa_1 +\overline{\delta} \delta
\kappa_1 +4\tau \overline{\delta}\kappa_1
 - (3\alpha+\overline{\beta})\delta\kappa_1  -  (3\overline{\alpha}+\beta)\overline{\delta}\kappa_1
 +2\Psi_2 \kappa_1 
\end{equation*}
on $\Htwo$. Taking a further $\Delta$-derivative we obtain
\begin{equation*}
2 \Delta\Delta D \kappa_1 = \Delta(\delta\overline{\delta} +\overline{\delta} \delta)
\kappa_1 +4\tau \Delta\overline{\delta}\kappa_1
 - (3\alpha+\overline{\beta})\Delta\delta\kappa_1  -  (3\overline{\alpha}+\beta)\Delta\overline{\delta}\kappa_1
 +2\Psi_2 \Delta\kappa_1.
\end{equation*}
We can commute the $\Delta$-derivative with the $\delta$ and $\overline{\delta}$ derivatives to obtain
\begin{equation*}
2 \Delta\Delta D \kappa_1 = (\delta\overline{\delta} +\overline{\delta} \delta)
\Delta\kappa_1 +4\tau \overline{\delta}\Delta\kappa_1
 - (3\alpha+\overline{\beta})\delta\Delta\kappa_1  -  (3\overline{\alpha}+\beta)\overline{\delta}\Delta\kappa_1
 +2\Psi_2 \Delta\kappa_1.
\end{equation*}
Note that all the terms on the right are proportional to intrinsic
derivatives of $\Delta\kappa_{1}$, which by \eqref{NPHH25} is
proportional to $\kappa_{2}$ and its intrinsic derivatives on
$\Htwo$. As shown in subsection \ref{subsection:kappa2},  unless our
spacetime is the Minkowski solution, the component $\kappa_{2}$ must
vanishes on $\Htwo$. It follows then that
\begin{equation*}
\Delta\Delta D\kappa_{1}=0\quad\text{ on }\Htwo\,.
\end{equation*}
This is a second order ordinary differential equation along the
generators of $\Htwo$. Therefore, the requirement that $D\kappa_{1}$
is real on $\Htwo$ is equivalent to requiring that $D\kappa_{1}$ and
$\Delta D\kappa_{1}$ are real on $\cZ$.

\medskip
\noindent
\textbf{Analysis on $\Hone$.} An analogous argument apply in case of
condition \eqref{HermiticityBasic2}. Take first a $D$-derivative
along the generators of $\Hone$ and use the transport equation
\eqref{KillingSpinorWaveEquationNP2H1} on $\Hone$, along with the
assumption that $\kappa_{0}$ vanishes in $\Hone$ to obtain 
\begin{equation*}
2D(\Delta\kappa_{1} + \tau\kappa_{2})=
 \delta \overline{\delta} \kappa_1
+\overline{\delta}\delta \kappa_1  - (\alpha -\overline{\beta})
\delta\kappa_1 -
(\overline{\alpha}-\beta)\overline{\delta}\kappa_1
 +2 \Psi_2 \kappa_1.
\end{equation*}
Taking a further $D$-derivative one gets 
\begin{equation}
2DD(\Delta\kappa_{1} + \tau\kappa_{2})=
 D(\delta \overline{\delta}
+\overline{\delta}\delta) \kappa_1  - (\alpha -\overline{\beta})
D\delta\kappa_1 -
(\overline{\alpha}-\beta)D\overline{\delta}\kappa_1
 +2 \Psi_2 D\kappa_1.
\end{equation}
By commuting the $D$ derivatives with the $\delta$ and $\overline{\delta}$ derivatives, we obtain
\begin{align*}
2DD(\Delta\kappa_{1} + \tau\kappa_{2})=& (\delta \overline{\delta}
+\overline{\delta}\delta) D\kappa_1 - (3\alpha+\overline{\beta})\delta D\kappa_{1} - (3\overline{\alpha}+\beta) \overline{\delta}D\kappa_{1} \\
& + \left(\delta\overline{\tau}+\overline{\delta}\tau +4\alpha\overline{\alpha}+2\alpha\beta +2\overline{\alpha}\overline{\beta}+2\Psi_{2}\right)D\kappa_{1}.
\end{align*}
Note that all terms on the right hand side are proportional to
$\delta$ and $\overline{\delta}$ derivatives of $D\kappa_{1}$, which by
\eqref{NPHH16} are proportional to $\kappa_{0}$ and its
$\delta$ and $\overline{\delta}$ derivatives on $\Hone$. Therefore,
again, unless our spacetime is the Minkowski solution, $\kappa_{0}=0$
holds on $\Hone$. Accordingly one has that 
\begin{equation*}
DD\left(\Delta\kappa_{1}+ \tau\kappa_{2}\right)=0\quad\text{on }\Hone.
\end{equation*}
Again, the latter is a second order ordinary differential equation along the generators of
$\Hone$, and so the requirement that $\Delta\kappa_{1}+
\tau\kappa_{2}$ is real on $\Hone$ is equivalent to requiring that
$\Delta\kappa_{1}+ \tau\kappa_{2}$ and $D\left(\Delta\kappa_{1}+
\tau\kappa_{2}\right)$ are real on $\cZ$.

\medskip
Summarising the results of this section we have:

\begin{lemma}
The spinor field $\xi_{AA'}$ is Hermitian on $\Hone\cup\Htwo$, and thereby on the domain of dependence of $\Hone\cup\Htwo$, if and only if the conditions 
\begin{eqnarray*}
&& \kappa_{1}+\overline{\kappa}_1 = \text{const}\,, \\
&& D(\kappa_{1}-\overline{\kappa}_{1})=0\,,\\
&&\Delta D(\kappa_{1}-\overline{\kappa}_{1})=0\,, \\
&& \Delta(\kappa_{1}-\overline{\kappa}_{1})+ 
\tau\,\kappa_{2}-\overline{\tau}\,\overline{\kappa}_{2} = 0, \\
&& D\left(\Delta(\kappa_{1}-\overline{\kappa}_{1})+ 
\tau\,\kappa_{2}-\overline{\tau}\,\overline{\kappa}_{2}\right)=0\,,
\end{eqnarray*}
are satisfied on $\cZ$.
\end{lemma}
Note that some of these conditions are redundant. For example, we know
that $D\kappa_{1}$ vanishes on $\cZ$ due to equation \eqref{NPHH16}
and the vanishing of $\kappa_{0}$, and so clearly
$D(\kappa_{1}-\overline{\kappa}_{1})$ also vanishes on $\cZ$. A similar
argument using equation \eqref{NPHH25} can be used to show that
$\Delta(\kappa_{1}-\overline{\kappa}_{1})+
\tau\kappa_{2}-\overline{\tau}\overline{\kappa}_{2}$ vanishes on $\cZ$. We can also use the requirement that $\mbox{Re}(\kappa_{1})$ is constant on $\cZ$ to show that
the other two conditions are equivalent. Indeed, we have that 
\begin{align*}
D\left(\Delta(\kappa_{1}-\overline{\kappa}_{1})+ 
\tau\kappa_{2} -\overline{\tau}\overline{\kappa}_{2}\right) &= D\Delta(\kappa_{1}-\overline{\kappa}_{1}) -2\tau\overline{\delta}\kappa_{1} +2\overline{\tau}\delta\overline{\kappa}_{1} \\
&= \Delta D(\kappa_{1}-\overline{\kappa}_{1}) +\tau\overline{\delta}(\kappa_{1}-\overline{\kappa}_{1}) + \overline{\tau}\delta(\kappa_{1} -\overline{\kappa}_{1}) -2\tau\overline{\delta}\kappa_{1} +2\overline{\tau}\delta\overline{\kappa}_{1} \\
&= \Delta D(\kappa_{1}-\overline{\kappa}_{1})
\end{align*}
where \eqref{NPHZ7}, the commutator $[\Delta, D]$, and the vanishing
of $D\tau$ (see Table \ref{Table:CharacteristicInitialData}), along
with the conditions $\delta\kappa_{1}=-\delta\overline{\kappa}_{1}$
and $\overline{\delta}\kappa_{1}=
-\overline{\delta}\overline{\kappa}_{1}$, have been used. We compute
now $\Delta D\kappa_{1}$. Eliminating $D\kappa_{2}$ by using
\eqref{NPHH17} the transport equation
\eqref{KillingSpinorWaveEquationNP2H2} on $\cZ$ can be seen to reduce
to
\begin{align*}
2\Delta D\kappa_{1} &= (\delta\overline{\delta} +\overline{\delta}\delta)\kappa_{1} -(3\alpha+\overline{\beta})\delta\kappa_{1} -(3\overline{\alpha}+\beta)\overline{\delta}\kappa_{1} -(2\overline{\alpha}+2\beta)D\kappa_{2} +2\Psi_{2}\kappa_{1} \\
&= (\delta\overline{\delta} +\overline{\delta}\delta)\kappa_{1} -(3\alpha+\overline{\beta})\delta\kappa_{1} +(\overline{\alpha}+3\beta)\overline{\delta}\kappa_{1} +2\Psi_{2}\kappa_{1}\,.
\end{align*}
Replacing $\delta$ and $\overline{\delta}$ derivatives with the 
$\eth$ and $\overline{\eth}$operators  we obtain
\begin{equation*}
2\Delta D\kappa_{1}=(\eth\overline{\eth} +\overline{\eth}\eth)\kappa_{1} -(2\alpha+2\overline{\beta})\eth\kappa_{1}+ (2\overline{\alpha} + 2\beta)\overline{\eth}\kappa_{1} +2\Psi_{2}\kappa_{1}\,.
\end{equation*}
The imaginary part of this equation is given by
\begin{align*}
2\Delta D(\kappa_{1}-\overline{\kappa}_{1}) &= \left(\eth\overline{\eth} +\overline{\eth}\eth\right)(\kappa_{1}-\overline{\kappa}_{1}) +2\Psi_{2}\kappa_{1} -2\overline{\Psi}_{2}\overline{\kappa}_{1} \\
&= 2\,\Big[\big(\eth\overline{\eth}\kappa_{1} +2\Psi_{2}\kappa_{1}\big)-\big(\overline{\eth}\eth\overline{\kappa}_{1}+2\overline{\Psi}_{2}\overline{\kappa}_{1}\big)\Big]\,,
\end{align*}
where in the second step the constancy of $\mbox{Re}(\kappa_{1})$ on
$\cZ$, along with the commutator \eqref{laplace} applied to the spin
weight zero quantity $\kappa_{1}$, was used.

\medskip
Summarising, we have that:

\begin{lemma}
\label{Lemma:HermiticityRefined}
The spinorial field $\xi_{AA'}$ is Hermitian on $\Hone\cup\Htwo$ if
and only if on $\cZ$ we have
\begin{subequations}
\begin{eqnarray}
&& \kappa_{1}+\overline{\kappa}_1 = \text{const}\,, \label{HermiticityRefined1}\\
&& \eth\overline{\eth}\kappa_{1} +2\Psi_{2}\kappa_{1}\ \ \text{is real.} 
\label{HermiticityRefined2}
\end{eqnarray}
\end{subequations}
\end{lemma}

\section{Identifying the Kerr spacetime}
\label{Section:Kerr}
In this section we make use of Theorem
\ref{Theorem:SpinorialCharacterisationKerr} to identify the values of
the parameters $\mathfrak{K}$, $\mathfrak{b}$ and $\mathfrak{c}$
defining the coefficient $\kappa_1$ on $\mathcal{Z}$ which correspond to the Kerr
solution. To this end, we first identify conditions ensuring that the
Killing vector associated to the Killing spinor is Hermitian.

\subsection{Imposing the Hermiticity of $\xi_{AA'}$}
Recall that the conditions ensuring the Hermiticity of the spinor
$\xi_{AA'}$ have been given in Lemma
\ref{Lemma:HermiticityRefined}. Accordingly, we now proceed to
evaluate conditions
\eqref{HermiticityRefined1}-\eqref{HermiticityRefined2} in the
explicit solution \eqref{AxisymmetricKappa}.

\medskip
\noindent
\textbf{Condition \eqref{HermiticityRefined1}.} From
\eqref{AxisymmetricKappa} it readily follows that
\[
\kappa_1 + \overline{\kappa}_1 = (\mathfrak{c} +\overline{\mathfrak{c}})\psi + (\mathfrak{b}+\overline{\mathfrak{b}}). 
\]
Thus, condition \eqref{HermiticityRefined1} requires that
$\mathfrak{c}+\overline{\mathfrak{c}}=0$ so that one can write
\[
\mathfrak{c} = \mbox{i} c, \qquad c\in \mathbb{R}.
\]
Accordingly, expression \eqref{AxisymmetricKappa} simplifies to 
\[
\kappa_1 = \mbox{i} c \psi + \mathfrak{b}
\]
so that
\begin{equation}
\Psi_2 = \frac{\mathfrak{K}}{(\mathfrak{b}+ \mbox{i} c \psi )^3}.
\label{Psi2Refined}
\end{equation}

\medskip
\noindent
\textbf{Condition \eqref{HermiticityRefined2}.} A direct computation
shows that 
\begin{eqnarray*}
&& \eth \overline{\eth}\kappa_1 = \left( \frac{1}{\sqrt{2}R} Q \partial_\psi +
   \frac{1}{\sqrt{2}R}\partial_\psi Q \right)\left(\frac{1}{\sqrt{2}R}
   Q \partial_\psi \kappa_1\right)\\
&& \phantom{\eth \overline{\eth}\kappa_1} = \frac{Q^2}{2R^2}
   \partial^2_\psi \kappa_1 + \frac{Q}{R^2}\partial_\psi
   Q \partial_\psi \kappa_1.\\
&& \phantom{\eth \overline{\eth}\kappa_1} = \frac{Q}{R^2}\partial_\psi
   Q \partial_\psi \kappa_1,
\end{eqnarray*}
where in the last line it has used that the expression for $\kappa_1$
given by \eqref{AxisymmetricKappa} satisfies $\partial^2_\psi
\kappa_1=0$. As $\partial_\psi \kappa_1 = \mbox{i} c$, it readily
follows that
\[
\eth\overline\eth \kappa_1 + 2 \Psi_2 \kappa_1 = \frac{\mbox{i}c
  }{R^2}QQ' + \frac{2\mathfrak{K}}{(\mathfrak{b} +
  \mbox{i}c\psi)^2}. 
\]
Thus, condition \eqref{HermiticityRefined2} implies that 
\begin{equation}
(Q^2)'= 
  \frac{2\mbox{i}R^2\mathfrak{K}}{c(\mathfrak{b}+\mbox{i}c\psi)^2}
-\frac{2\mbox{i}R^2\overline{\mathfrak{K}}}{c(\overline{\mathfrak{b}}-\mbox{i}c\psi)^2}.
\label{HermiticityRefined2Substituted}
\end{equation}
As a consequence of the solution
\eqref{GeneralSolutionFrameCoefficient}, the above expression is
automatically satisfied so that condition \eqref{HermiticityRefined2}
does not add any further restrictions. 

\begin{lemma}
For the family of 2-metrics on $\mathcal{Z}$ given by Proposition
\ref{Proposition:MetricsZ}, the spinor $\xi_{AA'}$ associated to the
Killing spinor $\kappa_{AB}$ is Hermitian if and only if the
coefficient $\kappa_1$ on $\mathcal{Z}$ is of the form
\[ 
\kappa_1 =\mathfrak{b} + \mbox{i} c \psi\,.
\]
\end{lemma}

\subsection{Applying Mars's characterisation}
\label{Section:ApplyingMars}
If the spinor $\xi_{AA'}$ is Hermitian, then the associated 
Killing form is well defined and on $\mathcal{Z}$ the norm of the self-dual
Killing form, $\mathcal{H}^2$, associated to the Killing spinor $\kappa_{AB}$ is given
by
\[
\mathcal{H}^2 = -36\kappa_1^2 \Psi_2^2.
\]
Moreover, the Ernst potential $\chi$ takes on $\mathcal{Z}$, up to a
(possibly complex) constant $\mathfrak{x}\in\mathbb{C}$, the form
\[
\chi = \mathfrak{x}-18\kappa_1^2\Psi_2.
\]
Making use of the relation 
\eqref{MathfrakConstantRelation} to eliminate $\Psi_2$ one obtains 
\[
\mathcal{H}^2=-\frac{36\mathfrak{K}^2}{\kappa^4_1}\,, \qquad \chi = \mathfrak{x}-\frac{18\mathfrak{K}}{\kappa_1}.
\]
Now, in order to identify the Kerr spacetime via Theorem
\ref{Theorem:SpinorialCharacterisationKerr}, we set $\mathfrak{x}=1$
so that 
\[
1-\chi =\frac{18\,\mathfrak{K}}{\kappa_1}
\]
from which, in turn, one readily obtains that
\begin{eqnarray*}
&& (1-\chi)^4 = \frac{18^4\,\mathfrak{K}^4}{\kappa_1^4}\\
&& \phantom{(1-\chi)^4} = - \left( \frac{18^4\mathfrak{K}^2}{36}
   \right) \mathcal{H}^2.
\end{eqnarray*}
The previous expression allows to identify the constant $\mathfrak{l}$
in Theorem
\ref{Theorem:SpinorialCharacterisationKerr} given, in terms of the parameters used above, as
\[
\mathfrak{l} \equiv \frac{36}{18^4\,\mathfrak{K}^2}\,. 
\]
Thus, in order to have the Kerr spacetime $\mathfrak{l}$ must be real
and positive which can only be satisfied if $\mathfrak{K}$ is non-zero
and real, i.e.~$\mathfrak{K}=K\in \mathbb{R}\setminus\{0\}$.  Finally,
a direct computation using the constraint \eqref{ConditionGaussBonnet}
shows that $\mathfrak{b}=b\in\mathbb{R}$ ---cfr. Remark
\ref{ConditionGaussBonnet}.

\medskip
We summarise the discussion in the following:

\begin{proposition}
\label{Proposition:Kerr}
The members of the family of 2-metrics given in Proposition
\ref{Proposition:MetricsZ} giving rise to solutions to the vacuum
Einstein field equations on $D(\mathcal{H}_1\cup\mathcal{H}_2)$, which are
isometric to a member of the 2-parameter Kerr family of metrics are
characterised by the conditions
\[
\mathfrak{b}, \; \mathfrak{K}\in \mathbb{R}, \qquad \mathfrak{c} \in
\mathbb{C}\setminus\mathbb{R}.
\] 
These conditions fix the value of the component of the Weyl tensor $\Psi_2$ on
$\mathcal{Z}$. 
\end{proposition}

\subsubsection{Relation to the standard parameters of the Kerr family}
 From the previous discussion it follows that one can write
\begin{equation}
\Psi_2 = \frac{K}{(b+ \mbox{i}\,c\,\psi)^3}\,, \qquad b,\,c,\,K\in
\mathbb{R}\,. 
\label{Psi2Final}
\end{equation}
Clearly, $\Psi_2$ as given above is regular
everywhere on $\mathcal{Z}$ ---and accordingly, also the Gaussian
curvature of $\mathcal{Z}$. Now, observing that $b$ is an arbitrary
normalisation constant of the Killing spinor we conclude that the
representation of the Kerr family of spacetimes has two independent
constants ---as it should be expected!

\medskip
In order to relate the real parameters  $b,\,c,\,K$ with the standard
mass ($m$) and angular momentum ($a$) parameters of the Kerr family,
we recall that in a dyad $\{o^A,\,\iota^A\}$ consisting of principal
spinors of $\Psi_{ABCD}$, the only non-zero component of the Weyl
spinor is given, in terms of standard \emph{Boyer-Lindquist coordinates},  by
\[
\Psi_2 = -\frac{m}{r-\mbox{i} a \cos\theta}
\]
---see e.g. \cite{AndBaeBlu15}. In this dyad the Killing spinor takes
the form
\[
\kappa_{AB} = \frac{2}{3}(r-\mbox{i}a \cos\theta) o_{(A}\iota_{B)}.
\]
The normalisation in the above expression of the Killing spinor is
chosen so that the associated Killing vector has the form
\[
\xi^a =(\partial_t)^a.
\]
In Boyer-Lindquist coordinates the bifurcation sphere is determined by
the condition
\[
r = r_+, \qquad r_+\equiv m +\sqrt{m^2-a^2}.
\]
Thus, at the bifurcation sphere the component $\Psi_2$ of the Weyl
tensor takes the form
\begin{equation}
\Psi_2 = -\frac{m}{(r_+-\mbox{i} a \cos\theta)^3}.
\label{Psi2PrincipalDirections}
\end{equation}

\medskip
Now, in order to make contact with the framework of R\'{a}cz's
holograph construction we observe that although the spin dyad
associated to the null tetrad $\{l^a,\, n^a,\, m^a,\, n^a \}$
introduced in Section \ref{Section:CharacteristicInitialValueProblem}
is, in general, not aligned with the principal directions of the Weyl
tensor, it happens to be aligned at the bifurcation sphere
$\mathcal{Z}$. As the component $\Psi_2$ is invariant under
spin-boosts one can readily identify the expressions \eqref{Psi2Final}
and \eqref{Psi2PrincipalDirections} ---that is, one has
\[
\frac{K}{(b+ \mbox{i}\,c\,\psi)^3} = -\frac{m}{(r_+-\mbox{i} a \cos\theta)^3},
\]
so that, essentially, the constants $K$, $c$ and $b$ correspond,
respectively, to the values of the mass parameter, angular momentum
parameter and the value of the radial Boyer-Lindquist coordinate at
the event horizon. 

\section{Final remarks}
\label{Section:Conclusions}

As mentioned earlier, all the distorted electrovaccum black hole
spacetimes can be represented within R\'{a}cz's black hole holograph
construction \cite{Rac07,Rac14}. In this paper a systematic
investigation of a specific subset of these spacetimes was carried
out. This subset was chosen by requiring the existence of a Killing
spinor field in the pure vacuum case. The primary aim was to
identifying the freedom we have in choosing initial data for the
Killing spinor on the horizon of the underlying distorted vacuum black
hole. In accordance with R\'{a}cz's black hole holograph construction by
fixing merely one of the Killing spinor components on the bifurcation
surface the Killing spinor gets to be uniquely determined everywhere
in the domain of dependence of the horizons.

The motivation for the use of a Killing spinor field can be traced
back to the following conceptual issue raised already in
\cite{Rac07,Rac14}: Recall first that the Kerr family of vacuum black
holes represents only a critical point in the space of the distorted
vacuum black hole spacetimes. It is natural to ask then, what sort of
geometric selection rule, imposed only on the space bifurcation
surface, singles out the only asymptotically flat stationary vacuum
black hole spacetimes distinguished by the black hole uniqueness
theorems?

To get a clearer perspective of the results of the present paper it is
worth recalling some of the details of the black hole uniqueness
proofs. Note, first, that asymptotic flatness as an assumption is a
completely natural requirement if one is interested in the properties
of black holes which are completely isolated in space. It is not
surprising then that the black hole uniqueness theorems (see
e.g.~Refs.~\cite{Car71,Car73a,bunting,HE73}) all assume asymptotic
flatness of the domain of outer communications of the selected vacuum
spacetimes. Indeed, the black holes uniqueness proofs ---using the
black hole rigidity theorem of Hawking \cite{Hawk72, HE73} (claiming
that an asymptotically flat stationary electrovac black hole spacetime
is either static or stationary axisymmetric)--- can be traced back to
proving the uniqueness of solutions to an elliptic boundary value
problem \cite{israel1,israel2,Car71,Car73a,mazur,bunting}.  The
relevant elliptic equations are derived from the Einstein's equations
on ``$t=const$'' hypersurfaces (or---based on Hawking's black hole
rigidity theorem---on a suitable factor space of them), whereas the
boundary conditions are specified at the bifurcation surface and at
spacelike infinity \cite{bunting}.

In view of the great detail of information on the geometry of the
bifurcation surface provided by the presented investigations, one may
ask which part of them were actually used in the black hole uniqueness
proofs. The short answer is that almost none. More precisely, it was
only assumed that the geometry at the bifurcation surface is regular
and that the ``$t=const$'' hypersurfaces smoothly extend to this
surface. The validity of this latter assumption had been verified in a
series of papers either for generic metric theories of gravity
\cite{rw1,rw2} or in general relativity with the inclusion of various
matter fields \cite{frw,r1}. Nevertheless, the assumptions concerning
the geometry were never as detailed as given in the present paper. One
might be puzzled by this, but from the perspective of 
the black hole holograph construction \cite{Rac07,Rac14} it becomes
clear immediately that in identifying the Kerr family of black hole
solutions in the black hole holography construction one cannot refer
to the asymptotic properties. Accordingly, all the information we may
use must be restricted to the bifurcation surfaces which plays the
role of ``holograms'', as these compact two-dimensional carriers store
all the information concerning the geometry of the associated
four-dimensional stationary black hole spacetimes.

After having the selection rules identified in case of vacuum
configurations it is of obvious interest to get them also in the
electrovaccum case. In this way a completely new, quasi-local, type of black
hole uniqueness proofs can be established in the four dimensional
case.  Note, however, that---as it was also proposed in
\cite{Rac07,Rac14}--- in virtue of the large variety of stationary
black hole, black ring and other type of ``black'' objects in higher
dimensions it would be even more important to generalize the
techniques and concepts applied here to higher dimensions. The
corresponding investigations and constructions would definitely
deserve further attention.

\section*{Acknowledgements}
The calculations in this article have been carried out using the suite
{\tt xAct} in {\tt Mathematica} for manipulation of tensors and
spinors ---see \cite{xAct}.  JAVK thanks the hospitality of the Wigner
Institute in the course of a visit related to this work. IR and JAVK
also thanks the hospitality of Erwin Schr\"odinger Institute for
Mathematics and Physics of the University during the programme
\emph{Geometry and Relativity}.  IR was supported by the POLONEZ
programme of the National Science Centre of Poland which has received
funding from the European Union`s Horizon 2020 research and innovation
programme under the Marie Sk{\l}odowska-Curie grant agreement
No.~665778.




\begin{thebibliography}{10}

\bibitem{AndBaeBlu15}
L.~Anderson, T.~B\"{a}ckdahl \& P. Blue,
\newblock{\em Spin geometry and conservation laws in the Kerr
  spacetime},
\newblock in {\em One hundred years of General Relativity}, edited by
L. Bieri \& S.-T. Yau, page 183, International Press Boston, 2015.

\bibitem{BaeVal10b}
T.~B\"{a}ckdahl \& J.~A. {Valiente Kroon},
\newblock {\em On the construction of a geometric invariant measuring the
  deviation from Kerr data},
\newblock Ann. Henri Poincar\'e {\bf 11}, 1225 (2010).

\bibitem{bunting}  G.L.~Bunting: \textsl{Proof of the uniqueness conjecture
	for black holes}, Ph.\,D. Thesis, University of New England, Admirale (1987)


\bibitem{Car68b}
B.~Carter,
\newblock {\em Hamilton-Jacobi and Schr\"odinger separable solutions of
  Einstein's equations},
\newblock Comm. Math. Phys. {\bf 10}, 280 (1968).

\bibitem{Car71}  B. Carter: 
\newblock \textit{Axisymmetric black hole has only two
	degrees of freedom}, 
\newblock Phys. Rev. Lett. \textbf{26}, 331-333 (1971)

\bibitem{Car73a}
B.~Carter,
\newblock {\em Black hole equilibrium states},
\newblock in {\em Black holes ---les astres occlus}, edited by C.~DeWitt \&
  B.~DeWitt, page~61, Gordon and Breach, 1973.

\bibitem{ColVal16a}
M.~J. Cole \& J.~A. {Valiente Kroon},
\newblock {\em Killing spinors as a characterisation of rotating black hole
  spacetimes},
\newblock Class. Quantum Grav. {\bf 33}, 125019 (2016).

\bibitem{DobLewPaw18}
D.~Dobkowski-Rylko, J.~Lewandowski, \& T.~Pawlowski,
\newblock {\em A local version of the no-hair theorem},
\newblock in {\tt arXiv:1803.05463 [gr-qc]}, 2018.

\bibitem{Fri81a}
H.~Friedrich,
\newblock {\em On the regular and the asymptotic characteristic initial value
  problem for {Einstein}'s vacuum field equations},
\newblock Proc. Roy. Soc. Lond. A {\bf 375}, 169 (1981).

\bibitem{frw}  H. Friedrich, I. R\'acz and R.M. Wald: \textit{On the
	rigidity theorem for spacetimes with a stationary event horizon or a
	compact  Cauchy horizon}, Commun. Math. Phys. \textbf{204}, 691-707
(1999)

\bibitem{GarVal08b}
A.~Garc\'{\i}a-Parrado \& J.~A. {Valiente Kroon},
\newblock {\em Kerr initial data},
\newblock Class. Quantum Grav. {\bf 25}, 205018 (2008).

\bibitem{GarVal08a}
A.~Garc\'{\i}a-Parrado \& J.~A. {Valiente Kroon},
\newblock {\em Killing spinor initial data sets},
\newblock J. Geom. Phys. {\bf 58}, 1186 (2008).

\bibitem{Hawk72}  S.W.~Hawking: \textit{Black holes in general relativity},
Commun. Math. Phys. \textbf{25}, 152-166 (1972)

\bibitem{HE73}  S.W.~Hawking and G.F.R.~Ellis: \textsl{The large scale
	structure of space-time}, Cambridge University Press (1973)

\bibitem{hiw} S. Hollands, A. Ishibashi and R.M. Wald:
\newblock  {\it A higher
	dimensional stationary rotating black hole must be axisymmetric},
\newblock 
Commun. Math. Phys. {\bf 271}, 699-722 (2007) 

\bibitem{israel1}  W. Israel: \textit{Event horizons in static vacuum
	space-times}, Phys. Rev. \textbf{164}, 1776-1779 (1967)

\bibitem{israel2}  W. Israel: \textit{Event horizons in static electrovac
	space-times}, Commun. Math. Phys. \textbf{8}, 245-260 (1968)


\bibitem{Kan96b}
J.~K\'{a}nn\'{a}r,
\newblock {\em On the existence of $\mbox{{C}}^\infty$ solutions to the
  asymptotic characteristic initial value problem in general relativity},
\newblock Proc. Roy. Soc. Lond. A {\bf 452}, 945 (1996).

\bibitem{Luk12}
J.~Luk,
\newblock {\em On the Local Existence for the Characteristic Initial Value Problem in General Relativity}, 
Int. Math. Res. Not., {\bf 2012}, 4625–4678 (2012).

\bibitem{Mar00}
M.~Mars,
\newblock {\em Uniqueness properties of the Kerr metric},
\newblock Class. Quantum Grav. {\bf 17}, 3353 (2000).

\bibitem{xAct}
J.~M. Mart\'{\i}n-Garc\'{\i}a,
\newblock http://www.xact.es, 2014.


\bibitem{mazur} P.O. Mazur: \textit{Proof of uniqueness of the Kerr-Newman
	black hole solutions}, J. Phys. A: Math. Gen. \textbf{15}, 3173-3180 (1982)

\bibitem{PenRin84}
R.~Penrose \& W.~Rindler,
\newblock {\em Spinors and space-time. {V}olume 1. {T}wo-spinor calculus and
  relativistic fields},
\newblock Cambridge University Press, 1984.

\bibitem{rw1} I. R\'{a}cz and R.M. Wald: \textit{Extension of spacetimes
	with Killing horizon}, Class. Quant. Grav. \textbf{9}, 2643-2656 (1992)

\bibitem{rw2}  I. R\'{a}cz and R.M. Wald: \textit{Global extensions of
	spacetimes describing asymptotic final states of black holes, }Class. Quant.
Grav. \textbf{13}, 539-553 (1996)

\bibitem{r1} I. R\'acz: \textit{On further generalisation of the
	rigidity theorem for spacetimes with a stationary event horizon or a
	compact Cauchy  horizon}, Class. Quant. Grav. {\bf 17}, 153-178 (2000)

\bibitem{Rac07}
I.~R\'acz,
\newblock {\em Stationary black holes as holographs},
\newblock Class. Quantum Grav. {\bf 24}, 5541–5571 (2007).

\bibitem{Rac14}
I.~R\'acz,
\newblock {\em Stationary black holes as holographs II},
\newblock Class. Quantum Grav. {\bf 31}, 035006 (2014).

\bibitem{Ren90}
A.~D. Rendall,
\newblock {\em Reduction of the characteristic initial value problem to the
  Cauchy problem and its application to the Einstein equations},
\newblock Proc. Roy. Soc. Lond. A {\bf 427}, 221 (1990).

\bibitem{Ste91}
J.~Stewart,
\newblock {\em Advanced general relativity},
\newblock Cambridge University Press, 1991.

\bibitem{Wal94}
R.~M. Wald,
\newblock {\em Quantum field theory in curved spacetime and black hole
  thermodynamics},
\newblock Chicaho Lectures in Physics, The University of Chicago Press, 1994.

\bibitem{WalPen70}
M.~Walker \& R.~Penrose,
\newblock {\em On quadratic first integrals of the geodesic equation for type
  $\{22\}$ spacetimes},
\newblock Comm. Math. Phys. {\bf 18}, 265 (1970).

\end{thebibliography}

\end{document}